\documentclass{imsart}
\usepackage{float,booktabs}
\usepackage{hyperref}
\usepackage{xr-hyper}
\usepackage[numbers]{natbib}
\usepackage{amsmath}
\usepackage{mathrsfs}

\usepackage{mathabx}

\usepackage{amsthm}

\usepackage{amsfonts}
\usepackage{graphicx}
\usepackage{placeins}
\usepackage{algorithm2e}
\usepackage{caption}
\usepackage{array}
\usepackage{multicol}
\usepackage{color}
\usepackage{bbm}
\usepackage{tikz}
\usepackage{mhequ}
\usepackage{enumerate}
\usepackage{arydshln}
\usepackage[font={small,it}]{caption}
\usepackage[colorinlistoftodos,prependcaption,textsize=tiny]{todonotes}
\usepackage{accents}

\startlocaldefs

\newcommand{\ind}{\mathbbm 1}

\DeclareMathOperator{\InvGamma}{InvGamma}
\DeclareMathOperator{\diag}{diag}

\DeclareMathOperator{\var}{var}
\DeclareMathOperator{\No}{N}

\DeclareMathOperator{\tr}{tr}
\DeclareMathOperator{\KL}{KL}
\DeclareMathOperator{\cov}{cov}

\newcommand{\lvertiii}{{\vert\kern-0.25ex\vert\kern-0.25ex\vert}}    
\newcommand{\rvertiii}{{\vert\kern-0.25ex\vert\kern-0.25ex\vert}}

\newcommand{\1}{\mathbbm 1}
\newcommand{\mc}[1]{\mathcal #1}

\DeclareMathOperator{\En}{En}
\newcommand\munderbar[1]{\underaccent{\bar}{#1}}

\newtheorem{theorem}{Theorem}[section]
\newtheorem{remark}[theorem]{Remark}

\newtheorem{lemma}[theorem]{Lemma}

\newtheorem{corollary}[theorem]{Corollary}
\newtheorem{assumption}[theorem]{Assumption}

\newtheorem{lem}[theorem]{Lemma}

\newcommand{\be}{\begin{equs}}
\newcommand{\ee}{\end{equs}}
\def\ci{\perp\!\!\!\perp}
\newcommand{\bb}[1]{\mathbb #1}
\newcommand{\R}{\bb R}
\renewcommand{\P}{\mathcal P}
\newcommand{\X}{\mathbf X}
\newcommand{\E}{\mathbf E}
\newcommand{\bone}{\mathbf 1}
\newcommand{\bigO}{\mathcal O}

\DeclareMathOperator{\Cauchy}{Cauchy}

\def\TV{{ \mathrm{\scriptscriptstyle TV} }}

\DeclareMathOperator{\TVar}{TV}

\newcommand{\beginsupplement}{%
        \setcounter{table}{0}
        \renewcommand{\thetable}{S\arabic{table}}%
        \setcounter{figure}{0}
        \renewcommand{\thefigure}{S\arabic{figure}}%
        \setcounter{section}{0}
        \renewcommand{\thesection}{S\arabic{section}}        
        \setcounter{page}{1}
        \pagenumbering{roman}
}
\endlocaldefs

\begin{document}
\allowdisplaybreaks
% paper-specific macros go here

\begin{frontmatter}
\title{Bayes Shrinkage at GWAS scale: Convergence and Approximation Theory of a Scalable MCMC Algorithm for the Horseshoe Prior}
\runtitle{Horseshoe MCMC Convergence}

\begin{aug}
\author{\fnms{James E.} \snm{Johndrow}\ead[label=e1]{johndrow@stanford.edu}},
\author{\fnms{Paulo} \snm{Orenstein}\ead[label=e2]{pauloo@stanford.edu}}
\and
\author{\fnms{Anirban} \snm{Bhattacharya}\ead[label=e3]{anirbanb@stat.tamu.edu}}

\runauthor{J.E. Johndrow, P. Orenstein, and A. Bhattacharya}

\affiliation{Stanford University\thanksmark{m1} and Texas A\&M University\thanksmark{m2}}

\address{390 Serra Mall\\
Stanford, CA, 94305 USA \\
\vspace{-3mm}\phantom{E-mail: }\printead{e1}, \printead*{e2}}

\address{Ireland Street\\
College Station, TX 77843 USA\\
\printead{e3}}
\end{aug}

\begin{abstract}
The horseshoe prior is frequently employed in Bayesian analysis of high-dimensional models, and has been shown to achieve minimax optimal risk properties when the truth is sparse. While optimization-based algorithms for the extremely popular Lasso and elastic net procedures can scale to dimension in the hundreds of thousands, algorithms for the horseshoe that use Markov chain Monte Carlo (MCMC) for computation are limited to problems an order of magnitude smaller. This is due to high computational cost per step and growth of the variance of time-averaging estimators as a function of dimension. We propose two new MCMC algorithms for computation in these models that have improved performance compared to existing alternatives. One of the algorithms also approximates an expensive matrix product to give orders of magnitude speedup in high-dimensional applications. We prove that the exact algorithm is geometrically ergodic, and give guarantees for the accuracy of the approximate algorithm using perturbation theory. Versions of the approximation algorithm that gradually decrease the approximation error as the chain extends are shown to be exact. The scalability of the algorithm is illustrated in simulations with problem size as large as $N=5,000$ observations and $p=50,000$ predictors, and an application to a genome-wide association study with $N=2,267$ and $p=98,385$. The empirical results also show that the new algorithm yields estimates with lower mean squared error, intervals with better coverage, and elucidates features of the posterior that were often missed by previous algorithms in high dimensions, including bimodality of posterior marginals indicating uncertainty about which covariates belong in the model. 
\end{abstract}


\end{frontmatter}

\maketitle

\section{Introduction}
Modern applications in genetics and other areas of biology have stimulated considerable interest in statistical inference in the high-dimensional setting where the number of predictors $p$ is much larger than the number of observations $N$, and the truth is thought to be sparse or consist mostly of small signals. Regression models are frequently employed in this context. Consider a Gaussian linear model with likelihood
\be
L(z \mid W \beta, \sigma^2) &= (2 \pi \sigma^2)^{-N/2} e^{-\frac1{2 \sigma^2} (z-W\beta)'(z-W\beta)}, \label{eq:Likelihood}
\ee
where $W$ is a $N \times p$ matrix of covariates, $\beta \in \R^p$ is assumed to be a sparse vector, and $z \in \R^N$ is an $N$-vector of response observations. 
A common hierarchical Bayesian approach employs a Gaussian scale-mixture prior on $\beta$ of the form
\be
\begin{aligned}
\beta_j \mid \sigma^2, \eta, \xi &\stackrel{iid}{\sim} \No(0,\sigma^2 \xi^{-1} \eta_j^{-1}), \quad \eta_j^{-1/2} \stackrel{iid}{\sim} \Cauchy_{[0,b^{-1/2}]}(0,1), \quad j=1, \ldots, p,  \\
\xi^{-1/2} &\sim \Cauchy_{[a_\xi,b_\xi]}(0,1),\quad \sigma^2 \sim \InvGamma(\omega/2, \omega/2),
\end{aligned}
\label{eq:prior_main}
\ee
where $0 \le a_\xi < b_\xi \le \infty$ and $b \ge 0$, $\omega$ is a fixed prior hyperparameter, and $\Cauchy_{[\underline{a}, \overline{a}]}(0,1)$ is the standard Cauchy distribution restricted to the interval $[\underline{a}, \overline{a}]$. This is a slight generalization of the Horseshoe prior \cite{carvalho2010horseshoe}, whose original version used standard half-Cauchy priors on the local and global scales, that is, $a_\xi = b = 0, b_\xi = \infty$. The truncation of the prior on $\xi$ was introduced in \cite{van2017adaptive} for theoretical tractability, while we additionally truncate $\eta_j$ for our convergence analysis of the MCMC algorithms developed herein.  

%This is referred to as the Horseshoe prior \cite{carvalho2010horseshoe}. The truncation of the prior on $\xi$ is necessary to achieve minimax-optimal rates in \cite{van2017adaptive}, who put $\ell=1$ and $u=p^2$.
The prior structure \eqref{eq:prior_main} induces approximate sparsity in $\beta$ by shrinking most components aggressively toward zero while retaining the true signals \citep{polson2010shrink}. In this sense, the prior \eqref{eq:prior_main} approximates the properties of point-mass mixture priors \citep{johnson2012bayesian,scott2010bayes, george1997approaches}, which allow some components of $\beta$ to be exactly zero a posteriori. From the point of view of testing the hypotheses $H_{0j}: \beta_j = 0$, the global precision parameter $\xi$ controls how many of the hypotheses are true, while the local precisions $\eta_j$ control which of the hypotheses are true. In the normal means setting, where $W = I_N$ (the $N \times N$ identity matrix), this prior has been shown to achieve the minimax adaptive rate of contraction when the true $\beta$ is sparse \cite{van2014Horseshoe,van2017adaptive}. Moreover, under certain conditions the marginal credible intervals have asymptotically correct frequentist coverage \cite{van2017uncertainty}. Although early literature on the horseshoe prior justified it as a continuous approximation of the point-mass mixture prior, over time it has come to be recognized as a good prior choice in high-dimensional settings in its own right \cite{bhadra2017lasso}.

Despite the popularity of the horseshoe in the literature, there currently is a lack of MCMC algorithms that scale to large $(N, p)$, owing to expensive linear algebra and slow mixing of the corresponding Markov chain. The current state of the art algorithm for large $p$ is \cite{bhattacharya2016fast}, which has only been employed successfully up to about $p=10,000$, while the recently proposed algorithm of \cite{hahn2017efficient} scales very well in $N$ but is less efficient than the exact algorithm we propose here when $p \gg N$ (see \cite[Section 3]{hahn2017efficient}). Another recent proposal is \cite{makalic2016simple}, but compares only to the implementation in the \texttt{monomvn} package for \texttt{R}, which is very slow relative to \cite{bhattacharya2016fast}. The lack of scalable algorithms has kept a useful model designed for high-dimensional regression out of many modern high-dimensional applications such as genome-wide association studies (GWAS), which often have $N$ in the thousands and $p$ in the hundreds of thousands or more. Moreover, no MCMC algorithms for the horseshoe have yet been shown to be geometrically ergodic, which is the usual route to guaranteeing rapid convergence to the posterior, the existence of central limit theorems, and finite-time risk bounds for time averages and other pathwise quantities. This is despite geometric ergodicity results for the Bayes Lasso Gibbs sampler \cite{khare2013geometric} of \citet{park2008bayesian} and the Dirichlet-Laplace Gibbs sampler \cite{pal2014geometric} of \citet{bhattacharya2015dirichlet}. 

Here, we propose a new MCMC algorithm for the horseshoe that exhibits faster mixing than existing algorithms. Our basic approach to improving mixing is to make more extensive use of block updating. We prove that the associated Markov chain is geometrically ergodic. While the proof of the geometric ergodicity result utilizes a Lyapunov function with some similarities to those employed in \cite{khare2013geometric} and \cite{pal2014geometric}, certain delicate modifications were required to accommodate the pole at zero in the Cauchy prior on the local scales.
We then propose a numerical approximation that substantially reduces the cost per step of the algorithm in large $(N, p)$ settings. Specifically, we make fast approximations to several matrix products exploiting the sparsity structure of the posterior. We prove bounds on the approximation error to the posterior both for invariant measures of the approximate algorithm and for finite-time Ces\'{a}ro averages. These results utilize and expand on recent work on perturbation theory for geometrically ergodic Markov chains. In the process, we prove a general lemma showing that one can typically work in unweighted metrics and nonetheless obtain an approximation error bound in the metric weighted by a Lyapunov function, which substantially reduces the effort needed to establish guarantees of approximation accuracy for large classes of unbounded functions. We further show that if one gradually reduces the approximation error as the chain extends, it is possible to construct exact algorithms that utilize only approximate transition kernels. These latter results are very general and apply to a wide array of algorithms constructed from approximate kernels.

Perturbation theory has been a recent focus of the theoretical MCMC literature \citep{johndrow2017error, rudolf2017perturbation, pillai2014ergodicity}, as well as algorithm development \citep{bardenet2017markov, korattikara2014austerity, welling2011bayesian}. Earlier examples of perturbation theory for Markov chains under stronger ergodicity conditions include \citet{mitrophanov2005sensitivity} and \citet{roberts1998convergence}. This theoretical literature provides conditions under which finite-length paths from the approximate kernel $\P_\epsilon$ give provably good approximations to the posterior.  This approach is attractive from at least two perspectives: (1) it suggests the possibility of overcoming computational challenges for Bayesian inference in high-dimensional or large sample settings by replacing computational bottlenecks with faster numerical approximations, and (2) it allows practitioners to move beyond the setting of choosing a $\P$ that has exactly the ``right'' invariant measure from a set of alternatives that in practice is quite small. 

Unfortunately, the practical success of this strategy has thus far been fairly limited. Recent activity has focused on using subsamples or ``minibatches'' of data in the large $N$ setting to create an analogue of stochastic gradient methods for MCMC. As \citet{bardenet2017markov} point out, achieving provably good approximations with significant computational advantage using subsampling typically requires the posterior to be well-approximated by a Gaussian, which is unlikely to be the case in large $p$ applications. Accurate approximations using minibatches typically require the construction of control variates \cite{pollock2016scalable,baker2017control,bardenet2017markov}, which in practice can be time-consuming, particularly when the target is high-dimensional and near-sparse in most directions, and the important directions are not known \emph{a priori}. 

Here we show strong approximation error guarantees for an algorithm that is not based on subsampling yet gives orders of magnitude speedup in large $p$ settings. This is one of the first demonstrations we are aware of in which the perturbation strategy has resulted in a practically significant algorithmic advance in the high-dimensional setting when the posterior is not remotely close to a Gaussian. In particular, the horseshoe posterior differs from a Gaussian not only in the tails, but also in its ``center,'' since the posterior will often have many modes. Because the critical feature of our algorithm is exploitation of sparsity, we expect that a similar strategy could succeed in other canonical high-dimensional Bayesian models. 

We then compare by simulation the exact and approximate algorithms in a range of large $p$, small $N$ regression settings, and show that the approximation is empirically very accurate and has orders of magnitude lower computational cost per step than the exact algorithm. We conclude by utilizing the approximate algorithm to estimate the horseshoe on a GWAS dataset with $N=2,267, p=98,385$, which is an order of magnitude higher dimensional than the datasets considered by \cite{bhattacharya2016fast}. We compare these results to point estimates obtained using the Lasso, and show that while there is broad agreement in which variables are important, the horseshoe estimated using our approximate algorithm exhibits the expected behavior of shrinking the larger signals less and the smaller signals more than Lasso. We also show that the our approximate algorithm more accurately recovers nuanced features of the posterior compared to the exact algorithm of \cite{bhattacharya2016fast}, such as bimodality of marginals when the true signal is near the minimax threshold of detection. These bimodal marginals indicate uncertainty about which variables belong in the model, which is an often-touted argument for the use of Bayesian procedures compared to frequentist methods such as the Lasso which return only a single selected model.

\section{Algorithms} \label{sec:Algos}
Our program is a theoretical and empirical evaluation of two MCMC algorithms for approximation of the horseshoe posterior obtained by combining the likelihood \eqref{eq:Likelihood} with the hierarchical prior \eqref{eq:prior_main}. Both are blocked Metropolis-within-Gibbs algorithms. The first algorithm targets the exact posterior, while the latter allows some bias in order to reduce the computation time per step when $p$ is large. We begin by defining the update rule for both algorithms and summarizing their empirical performance before turning to theoretical results in the next section. For sake of brevity, we suppress dependence on $z$ and $W$ in the full conditionals of the state variables. 

\subsection{Exact Algorithm}
We first define some quantities that will be used repeatedly. Let
\be \label{eq:BasicQuantities}
\begin{aligned}
D &= \diag(\eta_j^{-1}), \quad M_\xi = I_N + \xi^{-1} W D W' \\
p(\xi \mid \eta) &= |M_\xi|^{-1/2} \left( \frac{\omega}2 + \frac12 z'M_\xi^{-1} z \right)^{-(N+\omega)/2} \frac{\1\{\xi \in (a_\xi,b_\xi)\}}{\sqrt{\xi} (1+\xi)}
\end{aligned}
\ee
A blocked Metropolis-within-Gibbs algorithm that targets the exact horseshoe posterior is given by the update rule
\be \label{eq:ExactAlgorithm}
\begin{aligned}
&\text{1. sample } \eta \sim p(\eta \mid \xi,\beta,\sigma^2) \,\propto\, \prod_{j=1}^p \frac{1}{1+\eta_j} e^{-\frac{\beta_j^2 \xi \eta_j}{2\sigma^2}} \1\{\eta_j > b\}. \\
&\text{2. propose } \log(\xi^*) \sim \No(\log(\xi),s), \text{ accept } \xi \text{ w.p. } \frac{p(\xi^* \mid \eta) \xi^*}{p(\xi \mid \eta) \xi}. \\
&\text{3. sample } \sigma^2 \mid \eta,\xi \sim \text{InvGamma}\left( \frac{\omega + N}2, \frac{\omega + z'M_\xi^{-1} z}{2} \right). \\
&\text{4. sample } \beta \mid \eta,\xi,\sigma^2 \sim \No\left( (W'W + (\xi^{-1} D)^{-1})^{-1} W'z, \sigma^2 (W'W + (\xi^{-1} D)^{-1})^{-1} \right).
\end{aligned}
\ee
We refer generically to the Markov transition operator defined by this update rule as $\P$. 
We require the truncation of the local precisions $\eta_j$ to prove geometric ergodicity when $p>N$ (though not when $p \le N$). An inspection of the proof of the prior concentration for the horseshoe in \cite{chakraborty2016bayesian} reveals that the same concentration result goes through with the prior on the local scales truncated above, which suggests the statistical optimality is maintained with the truncated prior \cite{bhattacharya2016bayesian}.

The algorithm in \eqref{eq:ExactAlgorithm} is new, though it is related to the algorithm of \cite[Supplement]{polson2014bayesian} and that of \cite{bhattacharya2016fast}. It differs from \cite{polson2014bayesian} in that the second step targets $p(\xi \mid \eta)$  rather than $p(\xi \mid \eta,\beta,\sigma^2)$ as in \cite{polson2014bayesian}, and thus blocks together $(\beta,\sigma^2,\xi)$ instead of only $(\beta,\sigma^2)$. It also differs from \cite{bhattacharya2016fast}, which did not do any blocking of $\beta,\sigma^2,\xi$. Moreover, whereas \cite{polson2014bayesian} and \cite{bhattacharya2016fast} used slice sampling targeting $p(\eta \mid \xi,\beta,\sigma^2)$, we develop an exact rejection sampler to sample the $\eta_j$s independently. The rejection sampler exploits that the full conditional density of $\eta_j$ is log-convex to build a piecewise upper envelope which can be conveniently sampled from, with careful choices of the pieces ensuring very high acceptance rates; we defer the details to a supplemental document. Being able to sample the $\eta_j$s exactly is convenient as it avoids the introduction of additional $p$ latent variables in the slice sampler, and also simplifies the convergence analysis of the Markov chain.

Like \cite{bhattacharya2016fast}, we use an efficient method for sampling from the Gaussian full conditional for $\beta$. The details of this method are relevant for understanding our approximate sampler, so we briefly summarize it here. To sample from $\beta \mid \eta,\xi,\sigma^2$, the following three steps suffice
\be \label{eq:GaussianSamplingTrick}
\begin{aligned}
\text{sample } u &\sim \No(0, \xi^{-1} D) \text{ and } f \sim \No(0, I_N) \text{ independently} \\
\text{ set } v &= W u + f, \quad v^\ast = M^{-1}_\xi (z/\sigma - v),  \\
\text{set } \beta  &= \sigma(u + \xi^{-1} D W' v^\ast).
\end{aligned}
\ee
Notice that this algorithm -- and indeed, all but one step of \eqref{eq:ExactAlgorithm} -- requires computing $M_\xi$ defined in \eqref{eq:BasicQuantities} and solving $M_\xi v^* = (z/\sigma-v)$ for $v^*$. When $p$ is large, the computational bottleneck of the algorithm in \eqref{eq:ExactAlgorithm} is, perhaps surprisingly, just computing the matrix $W D W'$, which is needed to compute $M_\xi$. This has computational cost $N^2 p$, which dominates every other calculation in the algorithm when $p > N$. In the next section, we propose an approximate sampler that has lower computational cost per step.

\subsection{Approximate Algorithm} \label{sec:ApproxAlgo}
To reduce computational cost per step, we employ an approximation of the matrix product $\xi^{-1} WDW'$. The horseshoe prior is designed for the sparse setting, where most of the true $\beta$'s are zero or very small. In this case, the horseshoe posterior will tend to concentrate strongly around zero for most of the true nulls, thus endowing it with its minimax adaptive properties. Of course, this means that the posterior has a great deal of structure, since we can typically expect it to be tightly concentrated around the origin in a subspace of dimension approximately $(p-s)$, where $s$ is the unknown number of non-nulls. 

We can exploit this structure to create very accurate approximations of $\xi^{-1} WDW'$. For entries of $\beta_j$ to be shrunk to near zero, the precision $\xi \eta_j$ must be very large, as can be seen from \eqref{eq:GaussianSamplingTrick}. When this is the case, the $j$th column of $W$ does not contribute much to the $N \times N$ matrix $\xi^{-1} WDW'$. An important practical consequence of this, hitherto unexplored, is that once the MCMC algorithm begins to converge, the matrix $\xi^{-1} WDW'$ will typically be well-approximated by hard-thresholding $D$ in (\ref{eq:BasicQuantities}), resulting in
\be \label{eq:WDWApprox}
M_\xi \approx M_{\xi,\delta} :\, = I_N + \xi^{-1} WD_\delta W', \quad D_\delta = \diag(\eta_j^{-1} \bone(\xi^{-1}_{\max} \eta_j^{-1} > \delta)) 
\ee
for ``small'' $\delta$, where $\xi_{\max} = \max(\xi,\xi^*)$ in the first step of \eqref{eq:ExactAlgorithm} (the choice of $\delta$ is considered in Section \ref{sec:ApproxError}). This thresholding step reduces computational cost considerably, since the columns of $W$ corresponding to the zero diagonal entries of $D_\delta$ can just be ignored. Letting 
\be \label{eq:ActiveSet}
S = \{ j : \xi^{-1}_{\max} \eta_j^{-1} > \delta\}, 
\ee
we can also write the approximation as
%\be
$M_{\xi,\delta} = I_N + \xi^{-1} W_S D_S W_S'$, 
%\ee
where $W_S$ consists of the columns of $W$ with indices in the set $S$, and $D_S$ consists of the rows and columns of $D$ with indices in the set $S$. This makes clear the computational advantages of thresholding. 

Using this strategy, we define an approximate algorithm that uses the same update rule as in \eqref{eq:ExactAlgorithm}, with only two changes:
\begin{enumerate}
 \item $M_\xi$ is replaced by $M_{\xi,\delta}$ everywhere that it appears in \eqref{eq:ExactAlgorithm}; and
 \item In the final step of \eqref{eq:GaussianSamplingTrick}, the quantity $D W'$ is replaced by $D_\delta W'$. 
\end{enumerate}
We denote the Markov transition operator corresponding to this variation of \eqref{eq:ExactAlgorithm} by $\P_\epsilon$. The subscript $\epsilon$ is meant to indicate that $\P_\epsilon$ is ``close'' to $\P$ in some suitable metric on probability measures. The choice of metric and how close $\P$ is to $\P_\epsilon$ as a function of the current state and $\delta$ are the focus of Section \ref{sec:PerturbationBounds}. %For now, we just give some representative empirical results.

The primary motivation behind the approximate algorithm is to improve per-iteration computational complexity without sacrificing accuracy. As discussed earlier, when the truth is sparse, we expect a large subset of $\{\xi^{-1} \eta_j^{-1}\}_{j \le p}$ to be small {\em a posteriori}, and hence thresholding the entries smaller than a small threshold $\delta$ should not affect the accuracy of the algorithm. Thresholding those small entries has significant computational advantages. The speedup from this approximation is best when $p$ is large relative to $N$ and the truth is sparse or close to sparse, so that most entries of $\beta$ are shrunk to near zero. Critically, coordinates that are thresholded away at iteration $k$ \emph{need not be thresholded away at iteration} $(k+1)$, and in practice the set of variables that escapes the threshold does change considerably from one iteration to another. This can occur because the thresholded coordinates are never actually set to zero or omitted, but rather sampled from a Gaussian that closely approximates the exact full conditional. Thus, we are not sacrificing the primary benefit of Bayesian methods for sparse regression: estimates of uncertainty about the set of true signals are still valid.

Consider the computational cost of extending the Markov chain by a single step. The exact algorithm needs to calculate $|M_\xi|$, $z' M_\xi^{-1} z$ and solve a linear system in $M_{\xi}$ in each iteration, each of which requires $O(N^3)$ operations. Further, formation of the matrix $M_\xi$ itself requires computation of $W D W'$, which has complexity $O(N^2 p)$. The approximate algorithm on the other hand needs to calculate $W D_\delta W'$, with a subset of the diagonal entries of $D_\delta$ being zero. With $S$ as in \eqref{eq:ActiveSet} denoting the {\em active set} of variables which escape the threshold, set 
$$
s_{\delta} = |S| = \sum_{j=1}^p \bone(\xi_{\max}^{-1} \eta_j^{-1} > \delta). 
$$
Also, let $D_S$ denote the $s_\delta \times s_\delta$ sub-matrix of $D$ and $W_S$ the $N \times s_\delta$ sub-matrix of $W$ resulting from picking out the non-thresholded diagonal entries/columns indexed by $S$. When the truth is sparse, $s_\delta \ll p$ after a few iterations and $W D_\delta W' = W_S D_S W_S'$, which costs $N^2 s_\delta$, providing significant savings. 
A second level of computational savings can be made when $s_\delta < N$, whence $W D_\delta W'$ is a reduced-rank approximation to $W D W'$. In such cases, our implementation altogether replaces the calculation of $W D_\delta W'$ and  
the formation of $M_{\xi, \delta}$ by directly calculating 
\be
M_{\xi, \delta}^{-1} = (I_N + \xi^{-1} W D_\delta W')^{-1} = I_N - W_S \big( \xi D_S^{-1} + W_S'W_S \big)^{-1} W_S',
\ee
using the Woodbury matrix identity. The calculation of $z' M_{\xi, \delta}^{-1} z$ and $M_{\xi, \delta}^{-1} (z/\sigma - v)$ are performed by substituting the above expression of $M_{\xi, \delta}^{-1}$, which only requires solving $s_\delta \times s_\delta$ systems, and has overall complexity $s_\delta^3 \vee s_\delta N$. The determinant of $I+\xi^{-1} W D_\delta W'$ is then computed by (a) performing a singular value decomposition of $W_S D_S^{1/2}$, which costs $\bigO(s_\delta^2 N)$, and then (b) calculating the eigenvalues as $1+\mathrm{s}^2$, where $\mathrm{s}$ is a vector of the singular values of $W_S D_S^{1/2}$, $(N-s_\delta)$ of which are identically zero. Accounting for the calculation of $W u$ performed when sampling $\beta$, which costs $\bigO(Np)$, the per step computational cost of the approximate algorithm when $s_\delta < N$ is order $(s_\delta^2 \vee p) N$. Thus by exploiting the sparse structure of the target, the algorithm achieves similar computational cost \emph{per step} to coordinate descent algorithms for Lasso and Elastic Net \cite[Sections 2.1, 2.2]{friedman2010regularization}. 

Before we conclude this section, we provide some additional insight into the consequences of the approximation for $\beta$. The effects of the modified updates for $\xi$ and $\sigma^2$ are relatively direct to see; however those for $\beta$ modify  multiple steps of the algorithm in \cite{bhattacharya2016fast}. Define $\Gamma := \xi^{-1} D$ and $\Gamma_\delta = \xi^{-1} D_\delta$. The approximate algorithm for $\beta$ sets
\begin{align*}
\beta = \Gamma_\delta W' M_\delta^{-1} z + \sigma \, (u - \Gamma W' M_\delta^{-1} v). 
\end{align*}
Since $(u, v)$ is jointly Gaussian, $\beta$ obtained above continues to have a Gaussian distribution, $\beta \sim N(\mu_\delta, \sigma^2 \Sigma_\delta)$, with 
\begin{align*}
\mu_\delta :\,= \Gamma_\delta W' M_\delta^{-1} z, \quad \Sigma_\delta :\,= \mbox{cov}(u - \Gamma W' M_\delta^{-1} v). 
\end{align*}
Some further simplifications (see Appendix A.1 for details) yields,
\begin{align}\label{eq:m_del}
\mu_\delta = (\mu_S; 0_{(p-s_\delta) \times 1}), \quad \mu_S = (W_S' W_S + \Gamma_S^{-1})^{-1} W_S' z,
\end{align}
and\footnote{When we write $\mu_\delta = (m_S; 0_{(p-s_\delta) \times 1})$, we simply mean that the sub-vector of $\mu_\delta$ corresponding to the indices in $S$ is $m_S$ while the rest are zero. Similarly, blocks of $\Sigma_\delta$ are defined by $S$ and $S^c$.} 
\begin{align}\label{eq:Sig_del}
\Sigma_{\delta} = 
\begin{bmatrix}
(W_S' W_S + \Gamma_S^{-1})^{-1} & - \Gamma_S W_S' M_S^{-1} W_{S^c} \Gamma_{S^c} \\
  & \Gamma_{S^c}
\end{bmatrix}.
\end{align}
Writing $\beta = (\beta_S; \beta_{S^c})$, we have $\E(\beta_{S^c}) = 0$, i.e, the entries of $\beta$ outside the active set are drawn from a zero mean distribution. Second, the marginal distribution of $\beta_S$ is $N\big( (W_S' W_S + \Gamma_S^{-1})^{-1} W_S' z, \, \sigma^2(W_S' W_S + \Gamma_S^{-1})^{-1} \big)$, which would exactly be the full conditional distribution of $\beta$ if the model was fitted with the current set of active variables.  

\subsection{Empirical performance}
Our main motivation for pursuing a variety of theoretical results about this algorithm is that it \emph{performs very well empirically.} In the final section, we apply the approximate algorithm to a GWAS dataset with $N=2267$ and $p=98,385$. This is about an order of magnitude larger in $p$ than any other application of horseshoe for linear models that we are aware of. Our exact algorithm has per-step linear complexity in $p$ and quadratic complexity in $N$, while the approximate algorithm actually has per-step linear complexity in $N$ and $p$ in many cases. It therefore competes with coordinate descent for the Lasso in terms of per-step computational cost. While the dependence of the mixing properties of the Markov chain on dimension is not considered theoretically, empirically we find that both algorithms are insensitive to $N$ and $p$ on typical metrics like autocorrelations and effective sample sizes.

We give a brief empirical comparisons of our two algorithms to the algorithm of \cite{bhattacharya2016fast} based on a simulation with $N=2,000$ and $p=20,000$, where the true $\beta$ consists of a sparse sequence of signals of varying sizes. We use $\delta = 10^{-4}$ for the approximate algorithm; choosing $\delta$ is considered in detail in Section \ref{sec:ApproxError}, along with a complete description of the simulation setup. The left panel of Figure \ref{fig:XiAC} compares autocorrelations for $\log(\xi)$ for the ``old'' algorithm of \cite{bhattacharya2016fast} to our exact algorithm (``new'') and our approximate algorithm. We focus on $\xi$ since this parameter is known to mix poorly in MCMC algorithms for the horseshoe \cite{polson2014bayesian}. Both of our algorithms improve mixing considerably. The right panel of Figure \ref{fig:XiAC} shows the distribution of effective samples per second, a measure of overall computational efficiency, over a number of parameters for the three algorithms. The new algorithm has slightly worse performance at the median than the old algorithm because of the slightly higher per-step cost of the blocked sampler, but performs much better for the slowest-mixing parameters. The approximate algorithm is about 50 times more efficient by this metric than the exact algorithm, and this gap widens with increasing $p$.

\begin{figure}[h]
 \centering
 \begin{tabular}{cc}
 \includegraphics[height=0.15\textheight]{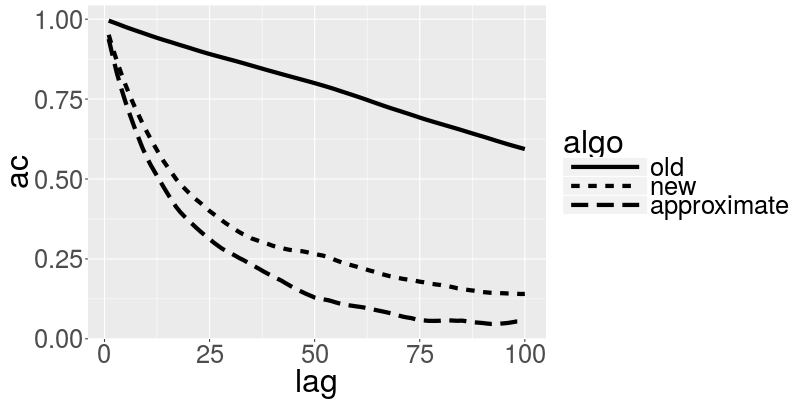} &  \includegraphics[height=0.15\textheight]{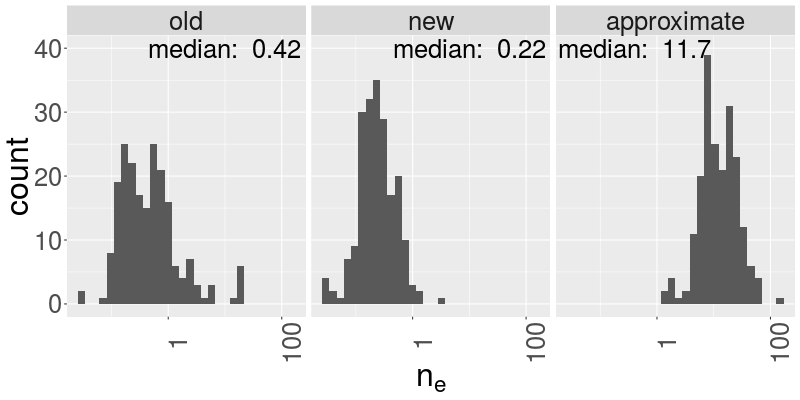}
 \end{tabular}
 \caption{Left: Estimated autocorrelations for $\log(\xi)$ for the three algorithms. Right: Effective samples per second for the three algorithms} \label{fig:XiAC}
\end{figure}

Figure \ref{fig:beta10} shows trace plots and density estimates for a single entry of $\beta$ for the three algorithms. This particular $\beta_j$ corresponds to a true signal of moderate size, and the resulting posterior marginal is bimodal, reflecting uncertainty about whether it is a signal or a null. Our exact and approximate algorithms both apparently mix well and result in visually similar estimates of the posterior marginal, while the old algorithm appears to become stuck at zero after a few thousand iterations, and the higher mode is lost after discarding a burn-in. Although this is a single entry of $\beta$, we later perform a more complete empirical comparison and find that the new algorithm outperforms the old algorithm on every metric we consider, while there is little discernible difference between the exact and approximate algorithms when $\delta = 10^{-4}$. Intuitively, the choice of $\delta$ should depend only weakly on dimension, since the matrix $WDW'$ is always regularized by the identity. Thus a ``small'' value of $\delta$ is one that is small relative to the eigenvalues of the identity, which are all 1. This is discussed in more detail in Section \ref{sec:ApproxError}.

\begin{figure}[h]
 \centering
 \includegraphics[width=0.7\textwidth]{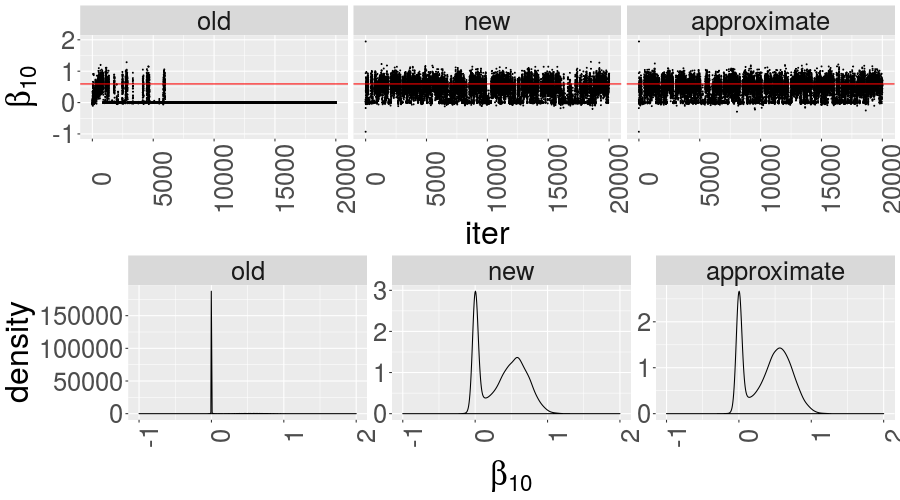}
 \caption{Trace plots (with true value indicated) and density estimates for one entry of $\beta$.} \label{fig:beta10}
\end{figure}

\section{Theoretical results} \label{sec:Theorems}
We now give results on the convergence of the exact algorithm and the accuracy of the approximate algorithm. We conclude with several general results showing that it is possible to construct exact algorithms using only approximate kernels. Any proofs not appearing immediately after the theorem statement are deferred to the appendix.

\subsection{Exact algorithm}
We first give some background on convergence rates of Markov chains. Our presentation follows closely that of \citet{hairer2011yet}. Let $\P$ be a Markov transition operator on a measurable state space $\X$. Let $x$ denote a generic element of the state space $\X$, so here $x = (\beta,\eta,\xi,\sigma^2)$. Denote by $P(\X)$ the set of probability measures on $\X$. We shall follow the general convention of \citep{hairer2011yet} to denote the action of $\P$ on a measurable function $f:\X \to \R$ and a probability measure $\nu$ on $\X$ by
\be 
\P f(x) = \int f(y) \P(x, y) dy, \quad \nu \P(A) = \int_{\X} \P(x, A) \nu(dx). 
\ee
We will state a form of geometric ergodicity of $\P$ that follows from two standard assumptions about $\P$, the first of which is existence of a Lyapunov function. 
\begin{assumption} \label{ass:Lyapunov}
 There exists a function $V : \X \to [0,\infty)$ and constants $0 < \gamma < 1$ and $K > 0$ such that
 \be
 (\P V)(x) \equiv \int V(y) \P(x,dy) \le \gamma V(x) + K.
 \ee
\end{assumption}
Lyapunov functions have been important in the study of stochastic stability at least since \citet{khasminskii1980stochastic}; their role in proving exponential convergence rates for Markov chains is thoroughly set down in the influential text of \cite{meyn1993markov} (see also \cite{rosenthal1995minorization}). Throughout, we will study convergence in a total variation norm weighted by the Lyapunov function. Define
\be
d_{\theta}(x,y) = \1\{x \ne y\} (2 + \theta V(x) + \theta V(y))
\ee
for some $\theta > 0$, and define a Lipschitz seminorm on measurable functions and associated dual metric on probability measures by
\be
\lvertiii \varphi \rvertiii_{\theta} &= \sup_{x \ne y} \frac{|\varphi(x) - \varphi(y)|}{d_{\theta}(x,y)} \\
d_{\theta}(\nu_1,\nu_2) &= \inf_{\Upsilon \in \mc C(\nu_1,\nu_2)} \int d_{\theta}(x,y) \Upsilon(dx,dy) = \sup_{\lvertiii \varphi \rvertiii_{\theta}<1} \int \varphi(x) (\nu_1 - \nu_2)(dx)
\ee
where $\mc C(\nu_1,\nu_2)$ is the space of couplings of the probability measures $\nu_1,\nu_2$. Notice that if $\theta=0$, we recover the ordinary unweighted total variation metric. %\abtodo{Say here that if $\theta= 0$, $d_\theta$ is the TVD? Also, below, we have used $d_{\theta V}$ only once it seems.}  
\citet{hairer2011yet} shows that $d_\theta$ is a $\theta V$-weighted total variation norm
\be \label{eq:WeightedTV}
d_{\theta}(\nu_1,\nu_2) = \int (1+\theta V(x)) |\nu_1 - \nu_2|(dx).
\ee

Geometric ergodicity of $\P$ is often proved by showing minorization on sublevel sets of $V$. We give the form of this condition used in \cite{hairer2011yet}.
\begin{assumption} \label{ass:Minorization}
For every $R>0$ there exists $\alpha \in (0,1)$ (depending on $R$) such that, for $\mc S(R) = \{x : V(x) < R\}$,
\be \label{eq:Minorization}
\sup_{x,y \in \mc S(R)} \|\delta_x \P - \delta_y \P\|_{\TV} \le 2 (1-\alpha).
\ee
\end{assumption}

While weaker conditions are also used, typically they are unnecessary when the collection of kernels are all absolutely continuous with respect to Lebesgue measure. We then have the following variation of Harris' theorem from \citet{hairer2011yet}.
\begin{theorem}[\citet{hairer2011yet}, Theorem 3.1] \label{thm:HM2011}
 Suppose $\P$ satisfies assumptions \ref{ass:Lyapunov} and \ref{ass:Minorization}. Then there exists $\bar \alpha \in (0,1)$ and $\theta > 0$ such that
 \be \label{eq:StrictContraction}
 d_{\theta}(\nu_1 \P, \nu_2 \P) \le \bar \alpha d_{\theta}(\nu_1,\nu_2)
 \ee
 for any two probability measures $\nu_1,\nu_2 \in P(\X)$. That is, $\P$ is geometrically ergodic (or $V$-uniformly ergodic).
\end{theorem}
Observe that iterating the estimate in \eqref{eq:StrictContraction} and letting $\nu_1 = \nu$, and $\nu_2 = \nu^\ast$, the invariant measure, gives the result
\be
d_{\theta}(\nu \P^n,\nu^\ast) \le \bar \alpha^n d_{\theta}(\nu,\nu^\ast),
\ee
so that convergence toward the target occurs at an exponential rate. This differs slightly from the usual notion of geometric ergodicity in the MCMC literature given in \eqref{eq:Gergo} below. In addition to being somewhat more convenient, the two results are essentially equivalent.
\begin{remark}
 Suppose $\P$ satisfies Theorem \ref{thm:HM2011}. Then there exists a $C< \infty$ such that
 \be \label{eq:Gergo}
 \sup_{|\varphi| < 1+V} \int \varphi(y) (\delta_x \P^n - \nu^\ast)(dy) \le C \bar \alpha^n V(x).
 \ee
\end{remark}
Thus we will typically use the definition of geometric ergodicity in \eqref{eq:StrictContraction}, with the understanding that it implies the usual notion in \eqref{eq:Gergo}. Our first result shows that the exact algorithm is geometrically ergodic. 
\begin{theorem} \label{thm:Gergo}
 Let $\P$ be the Markov transition operator with update rule given by  \eqref{eq:ExactAlgorithm}. Then for any $c < 1/2$
 \begin{enumerate}
 \item The function
\be \label{eq:Lyapunov}
V\big(\eta, \beta, \sigma^2, \xi\big) = \frac{\|W \beta\|^2}{\sigma^2} + \xi^2 + \sum_{j=1}^p \left[\frac{\sigma^{2c}}{|\beta_j|^{2c}} + \frac{\eta_j^c  |\beta_j|^c}{\sigma^{c}} + \eta_j^c \right] 
\ee
 is a Lyapunov function of $\P$, even if no truncation of the prior on $\eta$ is used (i.e. $b=0$ in \eqref{eq:prior_main}). 
 \item If $b>0$ in \eqref{eq:prior_main}, $\P$ is geometrically ergodic in the sense of \eqref{eq:StrictContraction}.
 %i.e. there exists $\theta \in (0,1)$ and $\bar \alpha \in (0,1)$ such that for any two probability measures $\nu_1,\nu_2$, $d_{\theta}(\nu \P, \nu_2 \P) \le \bar \alpha d_{\theta}(\nu_1,\nu_2)$. %\abtodo{Just say geometrically ergodic in the sense of (12)?}
 \end{enumerate}
\end{theorem}

\begin{remark}
 If $p \le N$ and $W$ is full-rank, then $\P$ is geometrically ergodic without any truncation of the prior on $\eta$. That is, one can take the constant $b=0$ in \eqref{eq:prior_main}.
\end{remark}

The Lyapunov function in \eqref{eq:Lyapunov} is somewhat unusual in that -- as a function of $\beta_j^2 \sigma^{-2}$ -- it both grows at infinity and has a pole at zero, whereas most commonly encountered Lyapunov functions simply grow at infinity and are bounded on compact sets containing the origin. This is necessary because showing the minorization condition in Assumption \ref{ass:Minorization} requires a uniform bound on the total variation distance between any two densities in the set 
\be
\left\{ p_1(\eta_j \mid \beta_j,\sigma^2,\xi) = \frac{e^{-m_j}}{\Gamma\left(0,m_j(1+b) \right)} \frac{1}{1+\eta_j} e^{-m_j \eta_j} \1\{\eta_j > b\} : (\beta,\sigma^2,\xi) \in \mc S(R) \right\}
\ee
for any $0 < R < \infty$ and every $j$, where $\Gamma(0,x)$ is the upper incomplete gamma function defined in \eqref{eq:incomp_gam}, $\mc S(R)$ is as defined in Assumption \ref{ass:Minorization} and $m_j = \beta_j^2 \xi/(2 \sigma^2)$. Clearly, minorization requires bounding $\beta_j^2 \sigma^{-2}$ away from infinity inside sublevel sets $\mc S(R)$ of $V$, since $p_1(\eta_j \mid \beta_j,\sigma^2,\xi) \to \delta_0(\eta_j)$ as $\beta_j^2 \sigma^{-2} \to \infty$, and $\delta_0(\eta_j)$ has total variation distance 1 from every measure with a continuous density. We must also bound $\beta_j^2 \sigma^{-2}$ away from zero for every $j$, since the limiting distribution is improper there. So the Lyapunov function must have sublevel sets in which $\beta_j^2 \sigma^{-2}$ is bounded away from both zero and infinity. The former is accomplished by the term $\sum_j \sigma^{2c} |\beta_j|^{-2c}$ regardless of the values of $N$ and $p$. However, truncation of the prior on $\eta$ is needed to achieve the latter. If $p>N$, then the function $\|W\beta \|^2 \sigma^{-2}$ is constant in the kernel of the linear function $W : \R^p \to \R^N$, and the only other appearance of positive powers of $\beta$ in \eqref{eq:Lyapunov} is in the term $\eta_j^c |\beta_j|^c \sigma^{-c}$. Thus, in order to ensure that $\beta$ cannot go to infinity in the kernel of $W$, we must have $\eta_j$ bounded away from zero in sublevel sets. In contrast, when $p \le N$ and $W$ is full rank, the term $\|W \beta\|^2 \sigma^{-2}$ is enough to keep $\beta_j^2 \sigma^{-2}$ bounded away from infinity in sublevel sets.

\subsection{Perturbation bounds} \label{sec:PerturbationBounds}
We now turn to proving error bounds for the approximate algorithm. Often in algorithm development, it is useful to identify computational bottlenecks, then design some computationally faster numerical approximation to alleviate the bottleneck. For example, the algorithm in section \ref{sec:ApproxAlgo} substitutes an approximation of the matrix $WDW'$ that is fast to compute when $\eta^{-1}$ is near sparse. This defines some new Markov operator $\P_\epsilon$. One then typically wants to know that the long-time dynamics of $\P_\epsilon$ will approximate those of $\P$. For example, we might ask whether the invariant measure(s) of $\P_\epsilon$ (assuming they exist) are close to the invariant measure $\nu^\ast$ of $\P$, or whether the usual time-averaging estimator 
\be
n^{-1} \sum_{k=0}^{n-1} \varphi(X_k^\epsilon)
\ee
for $X_k^\epsilon \sim \nu \P_\epsilon^k$ gives a good approximation to expectations under $\nu^\ast$. This is referred to as perturbation theory. This approach has significant advantages over studying $\P_\epsilon$ directly. For example, it is not necessary to show separately that $\P_\epsilon$ is geometrically ergodic (though in many cases it is), nor is it necessary that $\P_\epsilon$ has a unique invariant measure. Moreover, closeness of the invariant measure(s) of $\P_\epsilon$ to $\nu^\ast$ can be demonstrated as a corollary of bounds on the dynamics of the two chains.

Perturbation bounds for uniformly ergodic Markov operators date at least to \cite{mitrophanov2005sensitivity}, but more recent work \cite{johndrow2017error, rudolf2017perturbation, pillai2014ergodicity} focuses on the unbounded state space setting and the use of Lyapunov functions. One effectively needs two conditions, the first of which is some pointwise control of the kernel approximation error, typically in the same metric used to study convergence. In our setting one such condition is
\begin{assumption} \label{ass:Closeness}
The approximate kernel $\P_\epsilon$ satisfies
\be
\sup_{x \in \X} \|\delta_x \P - \delta_x \P_\epsilon\|_{\TV} \le \frac{\epsilon}2.
\ee
\end{assumption}
This differs from the basic error control assumption in both \cite{johndrow2017error} and \cite{rudolf2017perturbation}, which used variants of the condition $d_1(\delta_x \P,\delta_x \P_\epsilon) \le \epsilon (1+ \kappa V(x))$. In Theorem \ref{thm:Perturbation}, we show that the two conditions are essentially equivalent when one has control over stochastic stability of $\P_\epsilon$ via a Lyapunov function. It is often convenient if this is also a Lyapunov function of $\P$, so that the same weighted norms can be used to metrize convergence.
\begin{assumption} \label{ass:CommonLyapunov}
There exists $K_\epsilon > 0$ and $\gamma_\epsilon \in (0,1)$ such that
\be
(\P_\epsilon V)(x) \le \gamma_\epsilon V(x) + K_\epsilon.
\ee
\end{assumption}
Before stating our main results we point out an important general property of Lyapunov functions and weighted total variation metrics that allows us to use Assumption \ref{ass:Closeness} instead of an approximation error condition in $d_\theta$.
\begin{remark}
 If $\P$ has a Lyapunov function, then there must exist a Lyapunov function $V$ of $\P$ for which $V^2$ is also a Lyapunov function. In particular, if $\tilde V$ is a Lyapunov function of $\P$, then $V = \tilde V^{1/2}$ is a Lyapunov function of $\P$ whose square is also a Lyapunov function. Moreover, if $\P$ satisfies Assumption \ref{ass:Minorization} for $V^2$, then it also satisfies Assumption \ref{ass:Minorization} for $V$. Thus, if Theorem \ref{thm:HM2011} holds in the weighted total variation norm built on $V^2$, then it also holds in the weighted total variation norm built on $V$.
\end{remark}
\begin{proof}
The first part is proved in \cite{meyn1993markov}, but the argument is simple so we reproduce it here. Let $\tilde V$ be a Lyapunov function of $\P$ and put $V = \tilde V^{1/2}$. There exist $\tilde \gamma$ and $\tilde K$ so that
\be
(\P \tilde V)(x) \le \tilde \gamma \tilde V(x) + \tilde K.
\ee
By Jensen's inequality
 \be
 (\P \tilde V^{1/2})(x) \le (\P \tilde V(x))^{1/2} \le \sqrt{\tilde \gamma \tilde V(x) + \tilde K} &\le \sqrt{\tilde \gamma} \sqrt{ \tilde V(x)} + \sqrt{\tilde K} \\
 &\equiv \gamma V(x) + K.
 \ee
 and thus $V$ is a Lyapunov function for which $V^2$ is also a Lyapunov function. For the second part, we only need to show minorization on sublevel sets of $V$. Since $\P$ satisfies Assumption \ref{ass:Minorization} for $\tilde V$, for every $R^2> 0$ there exists $\alpha_{R^2} \in (0,1)$ (depending on $R$) so that
 \be
 \sup_{x,y \in \tilde{\mc S}(R^2)} \| \delta_x \P - \delta_y \P\|_{TV} \le 2(1-\alpha_{R^2})
 \ee
 for $\tilde{\mc S}(R^2) = \{x : \tilde V(x) < R^2\}$. But then Assumption \ref{ass:Minorization} is also satisfied for $V$, since letting $\mc S(R) = \{ x : V(x) < R\}$, we have
 \be
 \sup_{x,y \in \mc S(R)} \| \delta_x \P - \delta_y \P\|_{TV} \le 2(1-\alpha_{R})
 \ee
\end{proof}

The next result shows that under Assumptions \ref{ass:CommonLyapunov} and \ref{ass:Closeness}, we can obtain various bounds on the accuracy of $\P_\epsilon$.
\begin{theorem} \label{thm:Perturbation}
 Let $V$ be a Lyapunov function of $\P$ and $\P_\epsilon$ for which $V^2$ is also a Lyapunov function of $\P$ and $\P_\epsilon$. Suppose $\P$ satisfies Assumptions \ref{ass:Lyapunov} and \ref{ass:Minorization}, and $\P_\epsilon$ satisfies Assumptions \ref{ass:Closeness} and \ref{ass:CommonLyapunov}, all defined with respect to Lyapunov function $V$. Then, with $\psi(\epsilon) = C^* \sqrt{\epsilon}$ for a constant $C^*>0$, we have that for any probability measures $\nu_1,\nu_2$, 
 \be \label{eq:dPnPen}  
 \begin{aligned}                                      
  d_{\theta}(\nu_1 \P_\epsilon^n, \nu_2 \P^n) &\leq
                                       \frac{\psi(\epsilon)}{1-\bar \alpha}
                                       \left(\frac{1+ K_\epsilon}{1-\gamma_\epsilon}\right)
                                       + \psi(\epsilon) ( \nu_1 V
  ) (\bar \alpha \vee \gamma_\epsilon)^{n-1}n \\
  &+ \bar \alpha^n  d_{\theta}(\nu_1 , \nu_2), 
\end{aligned}
\ee
 which immediately implies that if $\nu_\epsilon^\ast$ is any invariant measure of $\P_\epsilon$
 \be
 d_\theta(\nu_\epsilon^\ast, \nu^\ast) &\le \frac{\psi(\epsilon)}{1-\bar \alpha} \left( \frac{1+ K_\epsilon}{1-\gamma_\epsilon} \right).
 \ee
 Furthermore, there exists $C,c_0,c_1 < \infty$ so that for any $|\varphi| < \sqrt{V}$
\be \label{eq:VariationBound}
\begin{aligned}
  \mathbf E \left( \frac1n \sum_{k=0}^{n-1} \varphi(X_k^\epsilon) - \nu^\ast \varphi  \right)^2 &\le 3 C^2 \psi(\epsilon) c_0  \\
  &+ \frac{3 C^2}{n} \left(\frac{2(1+K_{\epsilon})}{1-\gamma_\epsilon} + \frac{ \psi(\epsilon) c_1 V(x_0)}{1-\sqrt{\gamma_\epsilon}} \right) + \mathcal{O}\left( \frac{1}{n^2} \right),
  \end{aligned}
\ee
with $X_0^\epsilon = x_0$ and $X_k^\epsilon \sim \delta_{x_0} \P_\epsilon^{k-1}$. Moreover, the constants $C,c_0,c_1$ satisfy
\be
C &\le \frac{1 \wedge \nu^\ast V}{1-\bar \alpha_{(1/2)}}, \quad c_0 \le 2 + 5 \frac{K_\epsilon \vee \sqrt{K_\epsilon}}{(1-\sqrt{\gamma_\epsilon})^2} \quad c_1 = \left( 2 + \frac{\sqrt{K_\epsilon}}{1-\sqrt{\gamma_\epsilon}}  \right),
\ee
where $\nu^\ast$ is the unique invariant measure of $\P$ and $1-\bar \alpha_{(1/2)}$ is the spectral gap in the weighted total variation norm built on $V^{1/2}$ with an appropriate $\theta_{(1/2)}>\theta$. 
\end{theorem}
\begin{proof}
 The key to the result is the following Lemma. This result improves upon the result in \cite[Section 4.1]{johndrow2017error}, which showed that control in unweighted total variation was sufficient if the approximation error was tuned to the current state. This result requires only uniform control in the total variation or Hellinger distance over the entire state space.
 \begin{lemma} \label{lem:ApproxErrorV}
Suppose $V$ is a Lyapunov function of both $\P$ and $\P_\epsilon$ for which $V^2$ is also a Lyapunov function, and that Assumption \ref{ass:Closeness} holds. Then
with $\psi(\epsilon) = C^* \sqrt{\epsilon}$ we have
\be
\int (1+V(y)) |\delta_x \P - \delta_x \P_\epsilon |(dy) \le \psi(\epsilon) (1 + V(x))
\ee 
\end{lemma}
\begin{proof}
Write $\delta_x \P,\delta_x \P_\epsilon$ as densities 
\be
p(x,y) = \frac{d \delta_x \P}{d\nu}(y), \quad p_\epsilon(x,y) =  \frac{d \delta_x \P_\epsilon}{d\nu}(y)
\ee
with respect to an appropriate dominating measure $\nu$, which in our applications is just Lebesgue measure, and put $\tilde V = V^2$. Then
\be
I(x) &= \int V(y) |p(x,y) - p_\epsilon(x,y)| dy \\
I(x)^2 &= \left(\int V(y) (p^{1/2}(x,y) + p^{1/2}_\epsilon(x,y)) |p^{1/2}(x,y) - p^{1/2}_\epsilon(x,y)| dy \right)^2 \\
&\le 2 \left(\int \tilde V(y) (p(x,y) + p_\epsilon(x,y)) dy\right) \left(\int  (p^{1/2}(x,y) - p^{1/2}_\epsilon(x,y))^2 dy \right) \\
&\le 2 \left(\int \tilde V(y) (p(x,y) + p_\epsilon(x,y)) dy\right) \left(\int  |p(x,y) - p_\epsilon(x,y)| dy \right) \\
&\le 2 \left( (\gamma_0+\gamma_\epsilon) \tilde V(x) + K_0+K_\epsilon \right) \| \delta_x \P - \delta_x \P_\epsilon\|_{\TV} \\
I(x) &\le \sqrt{2} (\sqrt{\gamma_0+\gamma_\epsilon} \tilde V^{1/2}(x) + \sqrt{K_0+K_\epsilon}) \|\delta_x \P - \delta_x \P_\epsilon\|_{\TV}^{1/2} \\
&\le \sqrt{\epsilon} \left( \sqrt{K_0+K_\epsilon} + \sqrt{\gamma_0 + \gamma_\epsilon} V(x) \right) \\
&\le 2 \sqrt{\epsilon} \sqrt{K_0 + K_\epsilon+2} \left( \frac12 + V(x) \right)
\ee 
where we used the fact that $\gamma_0+\gamma_\epsilon < 2$. So finally we obtain
\be
\int (1+\sqrt{V(y)}) |\delta_x \P - \delta_x \P_\epsilon |(dy) &\le 2 \sqrt{\epsilon} \sqrt{K_0 + K_\epsilon+2} \left(\frac12 + V(x) \right) + \frac{\epsilon}2 \label{eq:BoundV} \\
&\le \psi(\epsilon) (1 +V(x)) 
\ee
for $\psi(\epsilon) = 2 \sqrt{\frac{\epsilon}2} \sqrt{K_0 + K_\epsilon+2}$, and we used the fact that $\epsilon < 2$ so $2\sqrt{ \epsilon} > \epsilon/2$.
\end{proof}
Now, \eqref{eq:dPnPen} follows from \cite[equation (10)]{johndrow2017error} and \eqref{eq:VariationBound} follows from \cite[Theorem 1.11]{johndrow2017error} after substituting $\psi(\epsilon)$ for $\epsilon$.
\end{proof}

The next result follows immediately from \eqref{eq:WeightedTV} and \eqref{eq:BoundV}.
\begin{corollary} \label{cor:ClosenessWtd}
 Suppose $\P$ satisfies Assumption \ref{ass:Closeness} and $V,V^2$ are Lyapunov functions of $\P$. Then with $\psi(\epsilon) = C^* \sqrt{\epsilon}$ 
 \be
 d_1(\delta_x \P, \delta_x \P_\epsilon) \le \psi(\epsilon) (1+ V(x)).
 \ee
\end{corollary}

Although these bounds are fairly transparent, a few comments are in order. First, all of the error bounds decrease to zero at rate $\sqrt{\epsilon}$. Second, $\P_\epsilon$ has an asymptotic bias proportional to $\sqrt{\epsilon} (1-\bar \alpha)^{-1}$, and all of the constants will be small when $\sqrt{\epsilon}$ is small relative to the spectral gap $1-\bar \alpha$. The implication is that there is more ``room'' to use approximations when the exact chain mixes rapidly, and the bias will be small when $\epsilon$ is small relative to the spectral gap. Moreover, it seems that if $\epsilon$ is gradually decreased to zero as the chain extends, one can achieve an exact algorithm using only approximate kernels; the conditions under which this occurs are made precise in the next section. 

In finite time, the practical tradeoff is between using a longer path from $\P_\epsilon$ with larger $\epsilon$, which results in larger bias but smaller variance, or a shorter path from $\P_\epsilon$ with smaller $\epsilon$, which has smaller bias but much larger variance. These tradeoffs are evident from \eqref{eq:VariationBound}, which gives an estimate of the squared error risk for the time-averaging estimator. Morally this is no different from choosing between two Markov kernels with the same invariant measure, where one mixes slowly but has low computational cost per step, and one mixes rapidly but has high computational cost per step.

The next result shows that our approximate horseshoe algorithm satisfies Assumptions \ref{ass:Closeness} and \ref{ass:CommonLyapunov}.
\begin{theorem} \label{thm:Close}
Let $\P_\epsilon$ be the Markov transition operator that uses the same update rule as $\P$ in Theorem \ref{thm:Gergo}, but approximates $WDW'$ and $D W$ by $WD_\delta W'$ and $D_\delta W'$ as in Section \ref{sec:ApproxAlgo} with a fixed value of $\delta$. Then 
\begin{enumerate}
\item $V(x)$ defined in \eqref{eq:Lyapunov} is a Lyapunov function of $\P_\epsilon$.
\item There exists a constant $C>0$ depending on $W,z$ such that for any $x \in \X$, 
\be \label{eq:PointwiseEtaBeta}
\sup_{x \in \X} \| \delta_x \P - \delta_x \P_\epsilon\|_{\TV} \le C \sqrt{\delta} + \bigO(\delta), 
\ee
where $\delta$ is the threshold tuning parameter for the matrix approximation in \eqref{eq:WDWApprox}.
\end{enumerate}
The result remains true even if no truncation of the prior on $\eta$ is used (i.e. we have $b=0$ in \eqref{eq:ExactAlgorithm}). In addition, if $b>0$ in \eqref{eq:ExactAlgorithm}, then $\P_\epsilon$ is geometrically ergodic.
\end{theorem}

It follows that Theorem \ref{thm:Perturbation} holds for our approximate algorithm using the Lyapunov function defined by the square root of \eqref{eq:Lyapunov}.

This result gives both a guarantee that taking $\delta$ sufficiently small, one can achieve any desired level of approximation error, and the rate at which the approximation error goes to zero with $\delta$. Of course, without knowing exactly the value of all of the constants, we cannot give exact estimates of the approximation error for any $\delta$. Section \ref{sec:ApproxError} focuses on choosing $\delta$ in practice.

\subsection{Exact algorithms using only approximate kernels} \label{sec:Exact}
We now give results showing how to construct ``exact'' versions of algorithms that only use approximating kernels. These results hold under general conditions and are not specific to the algorithms in Section \ref{sec:Algos}. 

In the MCMC literature, an algorithm is typically considered exact if
\be
\| \nu \P^k - \nu^\ast \|_{\TV} \to 0
\ee
as $k \to \infty$ for any starting measure $\nu$. Similarly, one might require that time averages converge to expectations under the target, e.g.
\be
\lim_{n \to \infty} \E \left( n^{-1} \sum_{k=0}^{n-1} \varphi(X_k^\epsilon) - \nu^\ast \varphi \right)^2 = 0 
\ee
for some large class of functions $\varphi$. Of course, in most cases any pathwise quantity from $\P$ will still have bias for any finite running time (though see the method of \cite{jacob2017unbiased} on de-biasing using couplings), and since one only ever has access to finite-time pathwise quantities, there is little practical difference between this guarantee and that given by Theorem \ref{thm:Perturbation}. Nonetheless, this property is often seen as desirable. The following result shows that we can achieve this by employing a sequence of approximating transition kernels $\P_{\epsilon_k}$ at step $k$, and taking $\epsilon_k \to 0$ as $k \to \infty$ at a slow rate. Notice that while this result uses the approximation condition $d_\theta(\delta_x \P, \delta_x \P_\epsilon) \le \epsilon_k (1+ V(x))$, this is implied by Assumption \ref{ass:Closeness} by \eqref{eq:WeightedTV} and \eqref{eq:BoundV} and the fact that the norms $d_1,d_\theta$ are equivalent (see \cite{hairer2011yet}). 
\begin{theorem}
Let $\{\epsilon_k\} \in [0,1]^{\infty}$. Consider a Markov chain $\{X_k\}$ defined by $X_0 \sim \nu$, $X_k \mid X_{k-1} \sim \P_{\epsilon_k}(X_{k-1},\cdot)$, and denote $\P_{\epsilon_1} \P_{\epsilon_2} \cdots \P_{\epsilon_n} \equiv \prod_{k=1}^n \P_{\epsilon_k}$. 
 Suppose that for every $\epsilon_k$, 
 \be
 d_\theta(\delta_x \P,\delta_x \P_{\epsilon_k}) \le \epsilon_k (1+V(x)),
 \ee
 and that for every $\epsilon \in [0,1)$,
 %\be
 $(\P_{\epsilon} V)(x) \le \gamma_{\epsilon} V(x) + K_\epsilon$.
 %\ee
 Suppose further that
 %\be
 $\tilde \gamma = \sup_{\epsilon \le 1} \gamma_\epsilon < 1$ and $\tilde K = \sup_{\epsilon \le 1} K_\epsilon < \infty$.
 %\ee
 Then if
 \be \label{eq:SumEpsAlpha}
 \lim_{n \to \infty} \sum_{k=0}^{n} \epsilon_{n-k} \bar \alpha^k = 0,
 \ee
 we have
 \be
 \lim_{n \to \infty} \bigg\| \nu \prod_{k=0}^{n} \P_{\epsilon_k} - \nu^\ast \bigg\|_{\TV} = 0,
 \ee
 and if
 \be \label{eq:SumSqrtEps}
 \lim_{n \to \infty} n^{-2} \sum_{k=1}^{n} \sum_{j=1}^n \sqrt{\epsilon_j \epsilon_k} = 0, 
 \ee
 then for any function $|\varphi| < \sqrt{V}$,
 \be
\lim_{n \to \infty} \E\left( n^{-1} \sum_{k=0}^{n-1} \varphi(X_k^\epsilon) - \nu^\ast \varphi \right)^2 = 0.
 \ee
\end{theorem}
\begin{proof}
 Let $\tilde \P_k$ be the $k$-step transition kernel $\tilde \P_k = \prod_{j=1}^{k} \P_{\epsilon_j}$. By \cite[Corollary 1.6]{johndrow2017error}
 \be
 d_\theta(\nu_1 \tilde \P_n, \nu_2 \P^{n}) \le \bar \alpha d_\theta(\nu_1 \tilde \P_{n-1}, \nu_2 \P^{n-1}) + \epsilon_n (1+\nu_1 \tilde \P_{n-1} V).
 \ee
 Iterating this estimate we obtain
 \be
 d_\theta(\nu_1 \tilde \P_n, \nu_2 \P^{n}) \le \bar \alpha^n d_\theta(\nu_1,\nu_2) + \sum_{k=0}^{n-1} \bar \alpha^k \epsilon_{n-k}  (1+\nu_1 \tilde \P_{n-k-1} V). 
 \ee
 Since $V$ is a Lyapunov function of $\P_{\epsilon_k}$ for every $\epsilon_k$, we obtain using the uniform bound on the constants
 \be
 \nu_1 \tilde \P_k V &\le \nu_1 \tilde \P_{k-1} (\gamma_{\epsilon_k} V + K_{\epsilon_k}) \\
 &\le \prod_{j=1}^k \gamma_{\epsilon_j} \nu_1 V + \sum_{j=0}^{k-1} \gamma_{\epsilon_j} K_{\epsilon_j} \\
 &\le \tilde \gamma^k (\nu_1 V) + \frac{\tilde K}{1-\tilde \gamma} \le \frac{\nu_1 V +\tilde K}{1-\tilde \gamma}
 \ee
 so
 \be
 d_\theta(\nu_1 \tilde \P_n, \nu_2 \P^{n}) &\le \bar \alpha^n d_\theta(\nu_1,\nu_2) + \sum_{k=0}^{n-1} \bar \alpha^k \epsilon_{n-k} \left(1 + \frac{\nu_1 V + \tilde K}{1-\tilde \gamma} \right) \\
 &= \bar \alpha^n d_\theta(\nu_1,\nu_2) + \left(1 + \frac{\nu_1 V + \tilde K}{1-\tilde \gamma} \right) \sum_{k=0}^{n-1} \bar \alpha^k \epsilon_{n-k}.
 \ee
 Substituting $\nu_2 = \nu^\ast$ we obtain
 \be
 d_\theta(\nu_1 \tilde \P_n, \nu^\ast) &\le \bar \alpha^n d_\theta(\nu_1,\nu^\ast) + \left(1 + \frac{\nu_1 V + \tilde K}{1-\tilde \gamma} \right) \sum_{k=0}^{n-1} \bar \alpha^k \epsilon_{n-k}.
 \ee
Now using that $d_\theta$ bounds total variation from above, the result follows whenever $\lim_{n \to \infty}\sum_{k=0}^{n-1} \bar \alpha^k \epsilon_{n-k} = 0$. 
 
 To show the second part, we apply \cite[equation (45)]{johndrow2017error} to obtain
 \be
 \frac1n \sum_{k=0}^{n-1} \varphi(X_k^\epsilon) - \nu^\ast \varphi = \frac{U(X_0^\epsilon)-U(X_n^\epsilon)}{n} + \frac1n M_n^\epsilon + \frac1n \sum_{k=0}^{n-1} (\P_{\epsilon_{k+1}}-\P) U(X_k^\epsilon)
 \ee
 where $\tilde \varphi = \varphi - \nu^\ast \varphi$,
 \be
 U = \sum_{k=0}^{\infty} \P \tilde \varphi, \quad M_n^\epsilon = \sum_{k=1}^n m_k^\epsilon, \quad m_{k+1}^\epsilon = U(X_k^\epsilon) -\P_{\epsilon_{k+1}} U(X_k^\epsilon). 
 \ee
Put  $C = \frac{(1 \vee \mu V^{1/2})}{\theta_{(1/2)}(1-\bar \alpha_{(1/2)})}$, where $1-\bar \alpha_{(1/2)}$ is the spectral gap of $\P$ in the weighted total variation norm built on $V^{1/2}$ with an appropriate $\theta_{(1/2)}$. Simple modifications to the calculations in \cite{johndrow2017error} using the uniform bounds on $K_\epsilon$ and $\gamma_\epsilon$ give
\be \label{eq:MartingalePart}
\E\left[\left(\frac1n M_n^{\epsilon} \right)^2\right]  &\le 2 C^2 \left( \frac1n +
\frac{\tilde K}{n\{1-\tilde \gamma\}} + \frac{1-\tilde \gamma^n}{n^2 
\{1-\tilde \gamma\}} V(x_0) \right),
\ee
and
\be
\E [(\P_{\epsilon_k} -\P) U(X_k^\epsilon) (\P_{\epsilon_j}-\P) U(X_j^{\epsilon})] \le C^2 \sqrt{\epsilon_k \epsilon_j} \left[ c_0 +   c_1 \tilde \gamma^{(j \wedge k)/2} V(x_0) \right]
\ee
so
\be 
\sum_{k=0}^{n-1} \sum_{j=0}^{n-1} \mathbf E \left[ (\P_{\epsilon_{k+1}}-\P) U(X_k^{\epsilon}) (\P_{\epsilon_{j+1}}-\P) U(X_j^{\epsilon}) \right] &\le \sum_{k=0}^{n-1} \sum_{j=0}^{n-1} C^2 \sqrt{\epsilon_{k+1} \epsilon_{j+1}} \left[ c_0 +  c_1 \tilde \gamma^{(j \wedge k)/2} V(x_0) \right] \\
&\le C^2 \left[ c_0 +   c_1 V(x_0) \right] \sum_{k=1}^{n} \sum_{j=1}^{n} \sqrt{\epsilon_k \epsilon_j} \\
n^{-2} \sum_{k=0}^{n-1} \sum_{j=0}^{n-1} \mathbf E \left[ (\P_{\epsilon_{k+1}}-\P) U(X_k^{\epsilon}) (\P_{\epsilon_{j+1}}-\P) U(X_j^{\epsilon}) \right] &\le \frac1{n^2} C^2 \left[ c_0 + c_1 V(x_0) \right] \sum_{k=1}^{n} \sum_{j=1}^{n} \sqrt{\epsilon_k \epsilon_j} \\
&\le \frac1{n^2} \tilde C \sum_{k=1}^{n} \sum_{j=1}^{n} \sqrt{\epsilon_k \epsilon_j} \label{eq:PerturbationPart}
\ee
Finally
\be \label{eq:BoundaryPart}
\frac{(U(X_0^\epsilon) -U(X_n^\epsilon))^2}{n^2} &\le \frac{4C^2}{n^2} \left( 1 + (1+\tilde \gamma^n) V(x_0) + \frac{\tilde K}{1-\tilde \gamma} \right). 
\ee
The quantities \eqref{eq:MartingalePart} and \eqref{eq:BoundaryPart} converge to zero at rate $n^{-1}$ and $n^{-2}$, respectively. The quantity in \eqref{eq:PerturbationPart} will converge to zero whenever
\be
\lim_{n \to \infty}  n^{-2} \sum_{k=1}^{n} \sum_{j=1}^{n} \sqrt{\epsilon_k \epsilon_j} = 0;
\ee
a sufficient condition is that $\sum_{k=0}^\infty \sum_{j=0}^\infty \sqrt{\epsilon_k \epsilon_j}$ is finite, in which case convergence to zero occurs at rate $n^{-2}$. This completes the proof of the second part.
\end{proof}

We note that the condition in \eqref{eq:SumSqrtEps} is exceedingly weak. Essentially, partial sums $\sum_{k=0}^{n-1} \epsilon_k$ need only grow slower than $n$. This is easily satisfied by many divergent series, such as the harmonic series.  The following remark shows that $\epsilon_k \to 0$ is sufficient for \eqref{eq:SumEpsAlpha} to hold.
\begin{remark} \label{rem:ConvergenceOfEpsilonSeries}
 Suppose $\epsilon_k \to 0$. Then $\lim_{n \to \infty} \sum_{k=0}^{n-1}\bar \alpha^k \epsilon_{n-k} = 0$.
\end{remark}
\begin{proof}
For any $k \le 0$ define $\epsilon_k = 1$. We have that
\be
\lim_{n \to \infty} \sum_{k=0}^{n-1} \bar \alpha^k \epsilon_{n-k} &\le \lim_{n \to \infty} \sum_{k=0}^{\infty} \bar \alpha^k \epsilon_{n-k} \\
&= \sum_{k=0}^{\infty} \bar \alpha^k (\lim_{n \to \infty} \epsilon_{n-k}) \\
&= 0, 
\ee
where in the second step we used that $\sum_{k=0}^{\infty} \bar \alpha^k \epsilon_{n-k} \le \sum_{k=0}^{\infty} \bar \alpha^k = (1-\bar \alpha)^{-1} < \infty$ and applied the dominated convergence theorem. This completes the proof.
\end{proof}
Taken together, the results in this section indicate that if one takes $\epsilon_k$ to zero, an exact algorithm can be obtained, including guarantees that time averages converge to expectations under the posterior measure uniformly over a large class of functions. This guarantee is similar to the guarantee of exactness for overdamped Langevin taking step sizes to 0 \cite{durmus2016high}. However, even one step of exact Langevin is always computationally infeasible, so there is no upper bound on the computational cost of the algorithm as the step sizes decrease to zero. In contrast, in our setting the exact algorithm is just a polynomial time rather than a linear time algorithm per step. The behavior of sequences of approximate kernels with decreasing approximation error can also be compared with that of pseudo-marginal MCMC, for which one usually can choose between a slowly mixing algorithm that has low per step cost and a faster mixing algorithm that has higher cost, both of which are exact. Similarly, in our case there is clearly some sequence $\epsilon_k$ that optimally trades off bias and variance for a fixed computational budget, but we must typically rely on empirical analysis to assess this. As such, it is often more practical to choose a single $\epsilon$ empirically, which we consider in the next section.

\section{Analysis of approximation error} \label{sec:ApproxError}
The results in the previous section show that the pointwise approximation error in $d_\theta$ decreases at rate $\delta^{-1/2}$. However, because we do not have a quantitative estimate of the spectral gap $1-\bar \alpha$, it is difficult to know how small $\delta$ needs to be to make the bias terms in Theorem \ref{thm:Perturbation} small. Here we give both some heuristic arguments for what a ``small enough'' value of $\delta$ will typically be, as well as an empirical analysis of the bias induced by different fixed values of $\delta$. The latter is done by comparing paths from the exact and approximate algorithms for different values of $\delta$ for problem sizes where it is feasible to run the exact algorithm. Of course, one can always achieve an asymptotically exact algorithm by using a decreasing sequence $\epsilon_k$ of approximation errors per the results of Section \ref{sec:Exact}. However, in practice simplicity is often highly valued, so the analysis in this section is aimed at choosing a default value of $\delta$ to be used in cases where the approximation error is held fixed over time.

\subsection{Heuristics for choosing $\delta$}
The threshold $\delta$ should satisfy $\delta \ll 1$. To see why, 
suppose that at most $N$ of the diagonal entries of $D$ are nonzero and the columns of $W$ corresponding to nonzero diagonal entries of $D$ are orthogonal. In this case 
\be
M_\xi = I + \xi^{-1} WDW' = W(I+\xi^{-1} D)W
\ee
with eigenvalues $(1+\eta_j^{-1} \xi^{-1})$ for $j : \lambda_j \ne 0$, so $M^{-1}$ has eigenvalues $(1+\eta_j^{-1} \xi^{-1})^{-1}$. In this idealized setting, an approximation $D_\delta$ that thresholds some of the nonzero $\eta_j$ will be very accurate precisely when $\xi^{-1} \eta_j^{-1} \ll 1$. This heuristic will be approximately correct whenever the columns of $W$ are not highly collinear. When highly collinear columns do exist, thresholding $\xi^{-1} \eta_j^{-1}$ at $\delta$ can result in a much worse approximation, since $m$ highly collinear columns with small $\eta_j$ correspond approximately to a single eigenvector of $\xi^{-1} WD W'$ with eigenvalue $m$ times what the above calculation would suggest. Thus, if we know that $W$ has low coherence, then a value of $\delta = 10^{-2}$ might be sufficient. To compensate for the likely presence of highly collinear columns in applications, we suggest taking $\delta = 10^{-4}$. We have detected no evidence of loss of accuracy using this threshold, even with highly dependent design matrices. In practice, the overriding factor in choosing such thresholds is computational budget, so we suggest taking $\delta$ as small as possible while satisfying computational constraints. 

\subsection{Empirical analysis}
We now empirically assess the approximation error for time averages by running the approximate algorithm for different values of $\delta$. The results in the following sections are based on a series of simulations in which the data are generated from
\be \label{eq:simulation}
w_i &\stackrel{iid}{\sim} \No_p(0,\Sigma) \\
z_i &\sim \No(w_i \beta, 4) \\
\beta_j &= \begin{cases} 2^{-(j/4-9/4)} & j<24 \\ 0 & j>23 \end{cases},
\ee
In contrast to typical simulations studies for shrinkage priors, in which signals are typically either zero or large relative
to the residual variance, we use a decreasing sequence of signals. The largest signal size is $4$, while $18$ out of the $23$ signals are smaller than the residual variance. For all of the problem sizes that we consider, this
results in bimodal marginal posterior for at least some of the $\beta_j$, increasing the difficulty of sampling from the
target. We consider two cases for $\Sigma$: the identity and $\Sigma_{ij} = \phi^{|i-j|}$. The latter is the covariance
matrix for an autoregressive model of order 1 with autoregressive coefficient $\phi$ and stationary variance $(1-\phi^2)^{-1}$. Throughout, we put $\phi = 0.9$ when
simulating a dependent design. Because all of the nonzero signals are in the first 23 elements of $\beta$, all of the $\beta_j$
corresponding to true signals will be highly correlated a posteriori, again considerably increasing the difficulty of efficiently sampling from the target.

For analysis of the approximation error, we simulate from \eqref{eq:simulation} with $N=1,000$ and $p=10,000$ for $\delta = 10^{-2},10^{-3},10^{-4}$, and $10^{-5}$. We also run the exact algorithm twice with different random number seeds. We collect paths of length 20,000 from each simulation after discarding a burn-in of 5,000. For the first 100 entries of $\beta$, which includes the 23 non-nulls and 77 nulls, we compute (1) correlation of pathwise means between the exact and approximate algorithm, (2) correlation of pathwise variances between the exact and approximate algorithms, and (3) Kolmogorov-Smirnov statistics for comparing the approximate algorithm to the exact algorithm. Each metric is also computed between the paths from the exact algorithm using a different random number seed. This last measurement gives some notion of how much variation one can expect in the estimates just due to MCMC error. A value of $\delta$ that performs similarly by these metrics to another copy of the exact algorithm initiated with a different random seed is thus one that achieves almost undetectable approximation error.

The Kolmogorov-Smirnov statistics are shown in Figure \ref{fig:delta-ks}. While $\delta = 10^{-2}$ or $10^{-3}$ have significant bias for at least some of the marginals, when $\delta = 10^{-4}$, none of the Kolmogorov-Smirnov statistics are greater that 0.1, and most are less than $10^{-1.5} \approx 0.03$. A slight inprovement is seen in decreasing $\delta$ to $10^{-5}$, for which the distribution of Kolmogorov-Smirnov statistics is hardly distinguishable from the distribution from a replicate simulation using the exact algorithm initiated using a different random number seed. 

\begin{figure}[h]
 \includegraphics[width=0.6\textwidth]{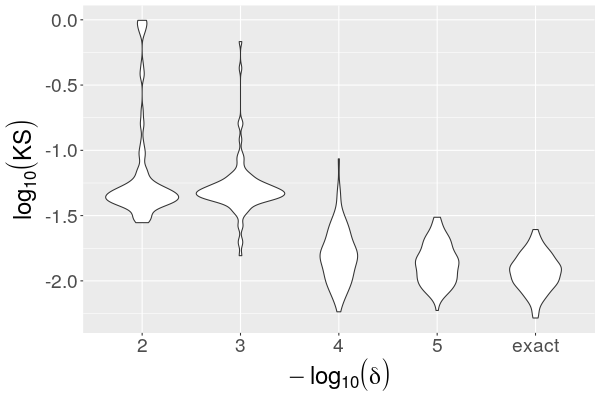}
 \caption{Distribution of Kolmogorov-Smirnov statistics comparing the marginals of 100 entries of $\beta$ for different values of $\delta$. } \label{fig:delta-ks}
\end{figure}

Table \ref{tab:DeltaCors} shows correlations between the means and variances of 100 entries of $\beta$ estimated using the exact and approximate algorithms. Similarly to the case of Kolmogorov-Smirnov statistics, significant disagreement is seen betweent the exact and approximate algorithms for $\delta = 10^{-2}$ or $10^{-3}$, but they are virtually indistinguishable for $\delta = 10^{-4}$ or $10^{-5}$. On the basis of these results, we typically choose $\delta = 10^{-4}$ in subsequent simulations.

% latex table generated in R 3.5.1 by xtable 1.8-3 package
% Fri Aug 31 11:14:11 2018
\begin{table}[ht]
\centering
\begin{tabular}{rlrr}
  \hline
 & $-\log_{10}(\delta)$ & mean & variance \\ 
  \hline
1 & 2 & 0.98 & 0.39 \\ 
  2 & 3 & 1.00 & 0.78 \\ 
  3 & 4 & 1.00 & 0.99 \\ 
  4 & 5 & 1.00 & 1.00 \\ 
  5 & exact ($\delta = 0$) & 1.00 & 1.00 \\ 
   \hline
\end{tabular}
\caption{Correlations between estimates of means and varianes of $\beta$ based on pathwise time averages for different values of $\delta$ } \label{tab:DeltaCors}
\end{table}

\section{Analysis of computational cost}

\subsection{Estimating dependence of constants on problem size}
The results of Section \ref{sec:Theorems} prove that the exact algorithm converges toward the posterior at an exponential rate, and give explicit bounds on the approximation error of time averages from $\P_\epsilon$ as a function of path length $n$. Moreover, we know the rate at which the computational complexity of taking one step from $\P$ or $\P_\epsilon$ grows with $N$ and $p$. However, rates are not always informative about the actual computational cost of an algorithm in finite dimensions, since one typically does not have sharp estimates of the constants. In particular, the spectral gap $1-\bar \alpha$ that appears in the results of Section \ref{sec:Theorems} often depends on $N$ and $p$. It is typically very difficult to determine theoretically how $\bar \alpha$ depends on $N,p$ in multistep Gibbs samplers like that in \eqref{eq:ExactAlgorithm} (see e.g. \cite{johndrow2018mcmc} for a very simple example that nonetheless required extensive calculation). 

However, one can conduct an empirical analysis of computational cost in the following way. If we take $\epsilon = 0$ in \eqref{eq:VariationBound}, the bound becomes
\be
\mathbf E \left( \frac1n \sum_{k=0}^{n-1} \varphi(X_k) - \nu^\ast \varphi  \right)^2\le \frac{3}{n} \left( \frac{1 \vee \mu V}{1-\bar \alpha_{(1/2)}} \right)^2 \left(\frac{2(1+K_{\epsilon})}{1-\gamma_\epsilon} \right) + \mathcal{O}\left( \frac{1}{n^2} \right).
\ee
It follows that the asymptotic (in $n$) variance of time averages of geometrically ergodic Markov chains is proportional to $(1-\bar \alpha_{(1/2)})^{-1}$. This term can be thought of as an upper bound on the sum 
\be \label{eq:IntegratedAutocov}
\tau^2_\varphi := \sum_{k=0}^{\infty} \cov(\varphi(X_0),\varphi(X_k)) \le \frac1{1-\bar \alpha_{(1/2)}}
\ee
for worst-case functions $|\varphi| < 1+V$ with $X_0 \sim \nu$. A common approach to study how this constant varies as a function of $N,p$ is thus to choose some collection of functions (usually coordinate projections) and compute an estimate of \eqref{eq:IntegratedAutocov} via plugging in pathwise estimates of the covariances obtained after discarding a burn-in and truncating the sum. This is taken to be an estimate of the asymptotic variance of $\varphi$. Numerous other estimators are available; see \cite{flegal2010batch}. Of course, there is no way to reliably find worst-case functions $\varphi$, but the empirical estimates at least give some sense of how this quantity behaves for statistically ``important'' functions like coordinate projections.

Estimates of $\tau_\varphi$ are referred to as MCMC standard error, and there is a significant literature
on the properties of different estimators (see \cite{flegal2010batch} for a rigorous treatment). 
Several of these estimators are implemented in the \texttt{R} package
\texttt{mcmcse}. We have consistently found the overlapping batch means estimator with the theoretically optimal $n^{1/3}$ 
batch size to perform the best, and we use this estimator throughout the paper. The asymptotic variance should be estimated
after discarding the initial portion of the path; we discard 5,000 scans.

Using estimates of $\tau^2_\varphi$ for
coordinate projections, we empirically analyze the effect of problem size on the required path length as follows. Suppose
that the relationship $\tau^2_\varphi = B N^{a_1} p^{a_2}$ for constants $B, a_1, a_2$ dictates the growth rate of $\tau^2_\varphi$ with $N$ and $p$; that is to say, the asymptotic variance grows like a polynomial in $N,p$. 
Then,  
$$
\log(\tau^2_\varphi) = \log(B) + a_1 \log(N) + a_2 \log(p),
$$
and thus one can obtain a 
rough estimate of the order of $\tau^2_\varphi$ in $p$ and $N$ from a regression of $\log(\hat \tau^2_\varphi)$ on 
$\log(N) + \log(p)$. We propose to compare estimates $\widehat a_1, \widehat a_2$ of $a_1,a_2$ across different algorithms as a way to empirically evaluate the relative computational complexity arising from the growth of the asymptotic variance.

A related pathwise quantity is the effective sample size $n_e$, which is usually defined as
\be \label{eq:Te}
n_e = \frac{\var_{\nu^*}(\varphi) n}{\tau^2_\varphi}, 
\ee
an adjustment to the path length $n$ to reflect how much the asymptotic variance, $\tau^2_\varphi$, is inflated by autocorrelation. 
Clearly, $n_e$ is proportional to the reciprocal of the asymptotic variance, so larger $n_e$ is better.
To estimate $n_e$ from paths of length $n$, we employ the procedure in \texttt{mcmcmse}, again using the overlapping batch means estimator with $n^{1/3}$ batch size and discarding 5,000 initial iterations.   

\subsection{Cost per step}
Table \ref{tab:t_reg} shows estimates of coefficients from a regression of $\log(t)$ on $\log(N)+\log(p)$ for the old, new, and approximate algorithms, where $t$ is computation time in seconds. These estimates are based on 20 simulations from the model in \eqref{eq:simulation} with $N$ sampled uniformly at random from integers between $200$ and $1,000$
and $p$ sampled uniformly at random from integers between $1,000$ and $5,000$. The algorithm was run for 20,000 iterations and total wall clock time
recorded. Computation was performed on multicore hardware with 12 threads, so matrix multiplications contribute less to the wall clock time than do
matrix decompositions, resulting in the lower than expected exponents on $N,p$. Thus, these estimates are meant to reflect the actual performance on modern 
multicore hardware. Moreover, the computation time of the approximate algorithm is likely non-constant
in $N,p$. For larger dimensions, the initial few iterations are likely to dominate the total computation time, since the benefits do not emerge until 
the algorithm locates most of the true nulls. This cost could be largely eliminated by ``warm starting'' the algorithm at, say, the cross-validated 
Lasso solution, which can be computed in nearly linear time in $N,p$. This approach could deliver a significant advantage in cases where lasso and horseshoe largely agree about the set of ``important'' variables, as in the application in Section \ref{sec:GWAS}.

% latex table generated in R 3.4.3 by xtable 1.8-2 package
% Fri Mar 30 10:50:24 2018
%\begin{table}[ht]
\begin{figure}[ht]
\captionof{table}{Estimates from regression of $\log(t)$ on $\log(N)+\log(p)$.} \label{tab:t_reg}
\begin{tabular}{lrrr}
  \hline
dimension & old & new & approximate \\ 
  \hline
$\log(N)$ & 1.6478 & 1.6847 & 0.5449 \\ 
  $\log(p)$ & 0.7204 & 0.6065 & 0.3392 \\ 
   \hline
\end{tabular}
\end{figure}
%\begin{tablenotes}
%\item JKL, just keep laughing; MN, merry noise.
%\end{tablenotes}
%\end{threeparttable}
%\end{table}

\subsection{Cost related to variance of time averages}
To assess the cost due to increased variance of the time-averaging estimator as a function of $N,p$, we conduct another set of simulations. We focus on the performance of the approximate algorithm, since its much lower computational cost per step allows a wider range of values of $N,p$ in the simulation study, improving the reliability of the results (recall that the approximate algorithm is also geometrically ergodic by Theorem \ref{thm:Close}). The results that follow are based on two simulation studies from the setup in \eqref{eq:simulation}, each consisting of 20 independent simulations in which $N$ was sampled uniformly at random from the integers between $1,000$ and $5,000$ and $p$ was sampled uniformly at random from the integers between $10,000$ and $50,000$. In the first simulation study, we use an independent design. In the second simulation study, we use a correlated design with AR-1 structure and autocorrelation 0.9 as described above. The approximate algorithm was run for 20,000 iterations. Calculations of effective sample sizes $n_e$ and standard errors were based on the final 15,000 iterations.  

The left panels of Figure \ref{fig:box_p_ess} shows the distribution of $n_e$ based on the first 100 entries of $\beta$, the corresponding entries of $\eta$, $\log(\xi)$, and $-2\log(\sigma)$ as a function of $p$; each simulation also has a different value of $N$. No variation by $p$ is evident in either the independent or correlated design case.  The right panels of Figure \ref{fig:box_p_ess} shows the analogous result, but as a function of $N$. A slight increase in effective sample size as $N$ increases is possible. There is apparently little difference in $n_e$ when the design matrix is correlated compared to independent design.

\begin{figure}[h]
%\centering
\begin{tabular}{cc}
 \includegraphics[width=0.5\textwidth]{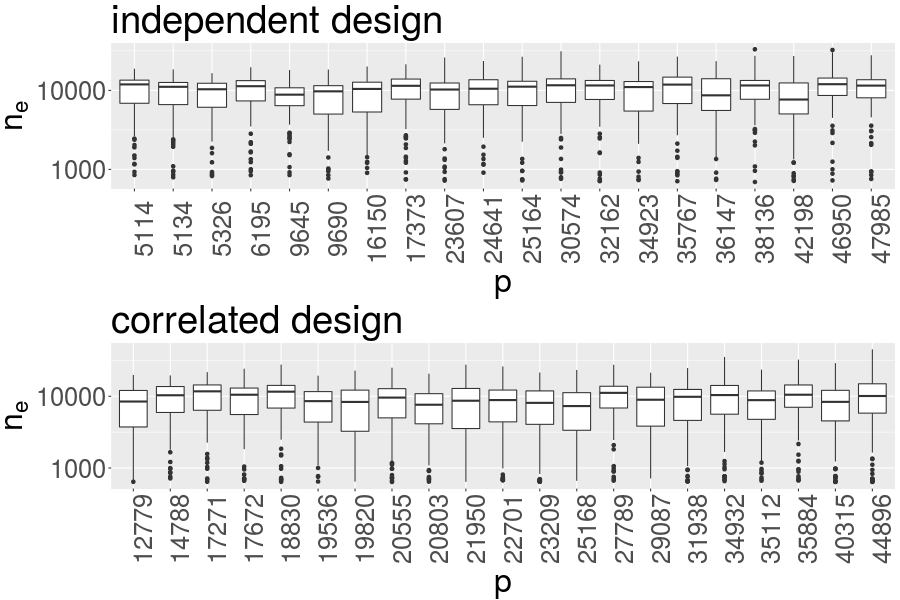} &  \includegraphics[width=0.5\textwidth]{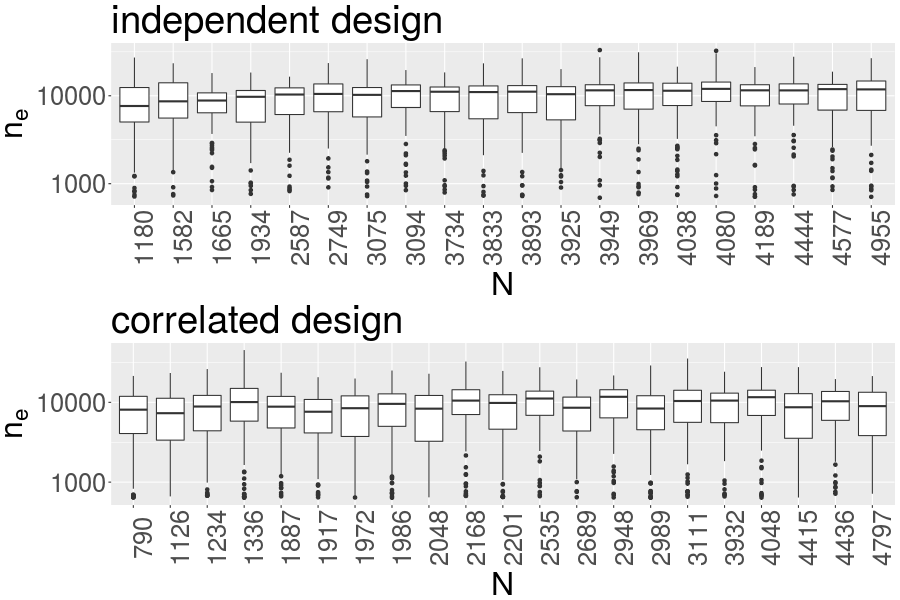}
\end{tabular}
 \caption{Left panels: Effective sample sizes $n_e$ for 100 entries of $\beta$, 100 entries of $\eta$, $\log \xi$ and $-2 \log \sigma$. The 100 entries of $\beta, \eta$ include those corresponding to all of the true signals, and 76 null signals. The horizontal axis indicates the value of $p$ used in each of the 20 simulations. Right panels: analogous to left panels, except that the horizontal axis indicates the value of $N$ used in each of the 20 simulations rather than value of $p$.} \label{fig:box_p_ess}
\end{figure}

Tables \ref{tab:reg_independent} and \ref{tab:reg_dependent} show results of a linear model with specification
\be
\log(\hat{\sigma}^2_{\varphi_j}(i)) = a_0 + a_1 \log(N_i) + a_2 \log(p_i) + b_j + \epsilon_{ij}
\ee
where $\varphi_j(x)=x_j$ is $j$th the coordinate projection of the partial state vector 
\be
x = (\beta_{1:100},\eta_{1:100},\log(\xi),-2\log(\sigma))
\ee
and $i=1,\ldots,20$ indexes the simulation scenario. Clearly, some coordinates tend to mix better than others, 
and the coordinate-specific intercepts allow for this variation. Results for independent design are shown in
Table \ref{tab:reg_independent} and for dependent design in Table \ref{tab:reg_dependent}. The small, negative
coefficient estimates suggest that if anything the Markov chain actually mixes slightly more rapidly as $N,p$ increase. 
Thus, there is little evidence that a longer path is needed to achieve fixed Monte Carlo error as $N,p$ grow. 

% latex table generated in R 3.4.3 by xtable 1.8-2 package
% Fri Mar 30 10:44:13 2018
%\begin{table}[ht]
%\centering
\begin{figure}[ht]
\captionof{table}{estimated parameters from regression of $-\log(n_e)$ on $\log(N)+\log(p)$, independent design} \label{tab:reg_independent}
\begin{tabular}{rrrrr}
  \hline
 & Estimate & Std. Error & t value & Pr($>$$|$t$|$) \\ 
  \hline
(Intercept) & -7.65 & 0.32 & -23.65 & 0.00 \\ 
  log(N) & -0.17 & 0.03 & -5.39 & 0.00 \\ 
  log(p) & -0.03 & 0.02 & -1.83 & 0.07 \\ 
   \hline
\end{tabular}
\end{figure}
%\end{table}

% latex table generated in R 3.4.3 by xtable 1.8-2 package
% Fri Mar 30 10:48:37 2018
%\begin{table}[ht]
%\centering
\begin{figure}[ht]
\captionof{table}{estimated parameters from regression of $-\log(n_e)$ on $\log(N)+\log(p)$, dependent design} \label{tab:reg_dependent}
\begin{tabular}{rrrrr}
  \hline
 & Estimate & Std. Error & t value & Pr($>$$|$t$|$) \\ 
  \hline
(Intercept) & -7.56 & 0.54 & -14.03 & 0.00 \\ 
  log(N) & -0.15 & 0.03 & -5.04 & 0.00 \\ 
  log(p) & -0.06 & 0.04 & -1.47 & 0.14 \\ 
   \hline
\end{tabular}
\end{figure}
%\end{table}

\section{Statistical performance}
Because the old algorithm can sometimes become trapped in potential wells (see Figure \ref{fig:beta10}), computational cost is not a complete measure of the difference between the old and new algorithms. In this section, we analyze the performance of the three algorithms in the estimation of $\beta$, which is typically the focus of inference. We again use the simulation setup in \eqref{eq:simulation} with $N$ sampled uniformly at random from the integers between $200$ and $1,000$ and $p$ sampled uniformly at random from the integers between $1,000$ and $5,000$. Mean squared error (MSE) for estimation of $\beta$ by MCMC time averages is shown in the left panel of Figure \ref{fig:mse}. There is no discernible difference between the performance of the new and approximate algorithms, but the old algorithm has about double the MSE at the median over the 20 simulations. Similarly, median empirical coverage of 95 percent credible intervals is about 90 percent for the old algorithm, and in only one case did the empirical coverage achieve 95 percent. In contrast, the new and approximate algorithms have median empirical coverage of about 93 percent, and never exhibited empirical coverage below 90 percent. We know from \citet{van2017uncertainty} that credible intervals for intermediate-sized signals cannot achieve the nominal coverage, even asymptotically. Since our simulation involves a sequence of decreasing signals, undoubtedly some of them fall into this ``intermediate'' categorization. As such, the performance of the new and approximate algorithms with respect to empirical coverage is probably near optimal. 

\begin{figure}[h]
\centering
\begin{tabular}{cc}
 \includegraphics[width=0.4\textwidth]{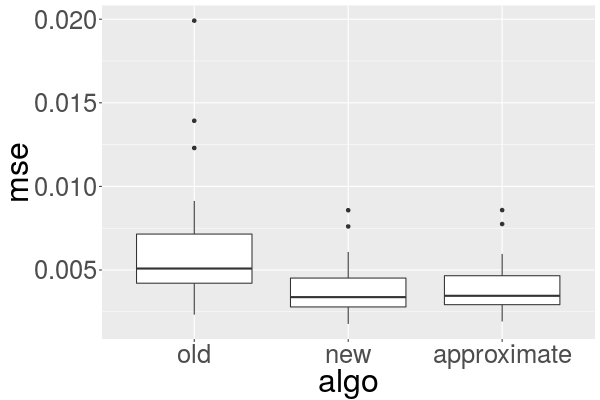} & \includegraphics[width=0.4\textwidth]{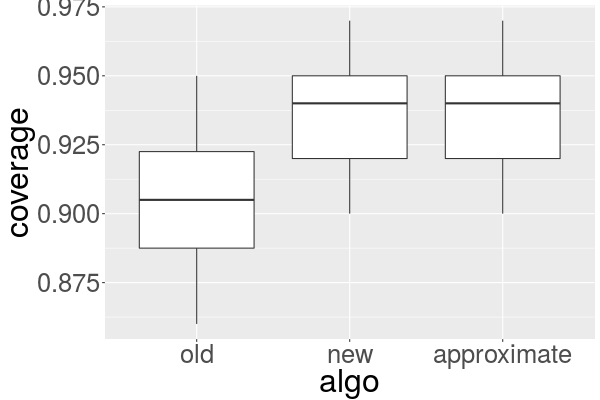} \\
\end{tabular}
 \caption{Mean squared error for estimation of $\beta$ by time averages (left) and coverage of 95 percent equal-tailed credible intervals based on time averaged quantiles (right). Boxplot is over results of 20 simulations.} \label{fig:mse}
\end{figure}

Figures \ref{fig:beta_marginals_small}, \ref{fig:beta_marginals_big}, and \ref{fig:sigma_intervals} in Supplementary Materials also evaluate statistical performance of the approximate algorithm. Figures \ref{fig:beta_marginals_small} and \ref{fig:beta_marginals_big} show posterior marginals for the first 25 entries of $\beta$ for simulations with $N=1,000$, $p=5,000$ and $N=5,000, p=50,000$, respectively, along with the true values of $\beta$. In general, the marginals have single modes centered near the truth for larger true signals, two modes with one centered near the truth and one centered at zero for intermediate sized true signals, and single modes at zero when the true signal is small or identically zero. This is consistent with the expected behavior of the Horseshoe. Figure \ref{fig:sigma_intervals} shows violin plots with indicated 95 percent credible intervals for $\sigma^2$ over 20 independent simulations each with $1,000 \le N \le 5,000$ and $5,000 \le p \le 50,000$. All but two of the intervals cover the true value of $2$. Overall, the approximate algorithm has exhibited excellent statistical performance by every metric we have considered.

\section{GWAS Application} \label{sec:GWAS}
We use the horseshoe with computation by the approximate algorithm to analyze a Genome-Wide Association Study (GWAS) dataset. The data consist of $N=2,267$ observations and $p=98,385$ single nucleotide polymorphisms (SNPs) in the genome of maize. These data have been previously studied by \cite{liu2016iterative} and \cite{zeng2017non}. Each observation corresponds to a different inbred maize line from the USDA Ames seed bank \citep{romay2013comprehensive}. As the response, we use growing degree days to silking, a measure of the average number of days exceeding a certain temperature that are necessary for the maize to ``silk.'' Maize is typically ready to harvest about 60 days after silking, so this is a measure of the length of the growth cycle for a particular line of maize, crudely controlling for temperature. This response is also considered by \cite{zeng2017non}.

We run the approximate algorithm for 30,000 iterations, discarding 5,000 iterations as burn-in. Figure \ref{fig:ess_maize} shows histograms of $n_e$ and $n_e/t$ for 200 entries of $\beta$, the corresponding 200 entries of $\eta$, $\log(\xi)$, and $-2 \log(\sigma)$. The 200 entries of $\beta, \eta$ includes the 100 entries for which the posterior mean is largest in absolute value, as well as 100 other entries. The smallest value of $n_e$ observed was 893, and the smallest value of $n_e/t$ 0.05. The median values were 4531 and 0.24, respectively. Thus the algorithm remains quite efficient, even on a fairly large, real dataset. 

\begin{figure}[h]
\centering
 \begin{tabular}{cc}
  \includegraphics[width=0.7\textwidth]{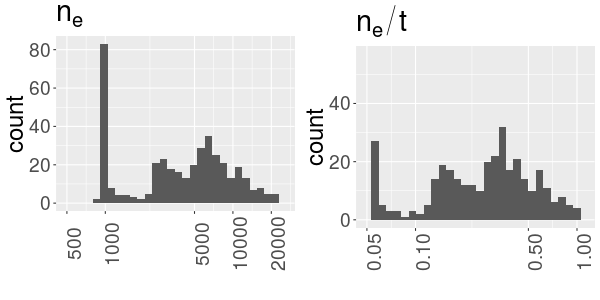}
 \end{tabular}
 \caption{Effective sample size (left) and effective samples per second (right) for maize application.} \label{fig:ess_maize}
\end{figure}

Figure \ref{fig:beta_marginals_maize} shows density plots of samples for the nine entries of $\beta$ with largest estimated absolute posterior mean, as well as the estimated posterior mean. The Lasso estimates for these parameters, with the penalty chosen by 10-fold cross-validation, are also indicated. It is clear that, even for the entries of $\beta$ for which the signal strength is largest, the Horseshoe marginals are typically bimodal, with the weight in the mode centered at zero increasing with decreasing signal strength. This suggests that the bimodal shape of the marginals may be quite common in applications, and gives some sense of the level of uncertainty about which entries correspond to true signals. The Lasso estimates for these relatively large parameters are typically shrunken toward zero relative to the Horseshoe posterior mean, a behavior that has been observed previously (see \cite{bhadra2017lasso}).

\begin{figure}[h]
\centering
 \includegraphics[width=0.7\textwidth]{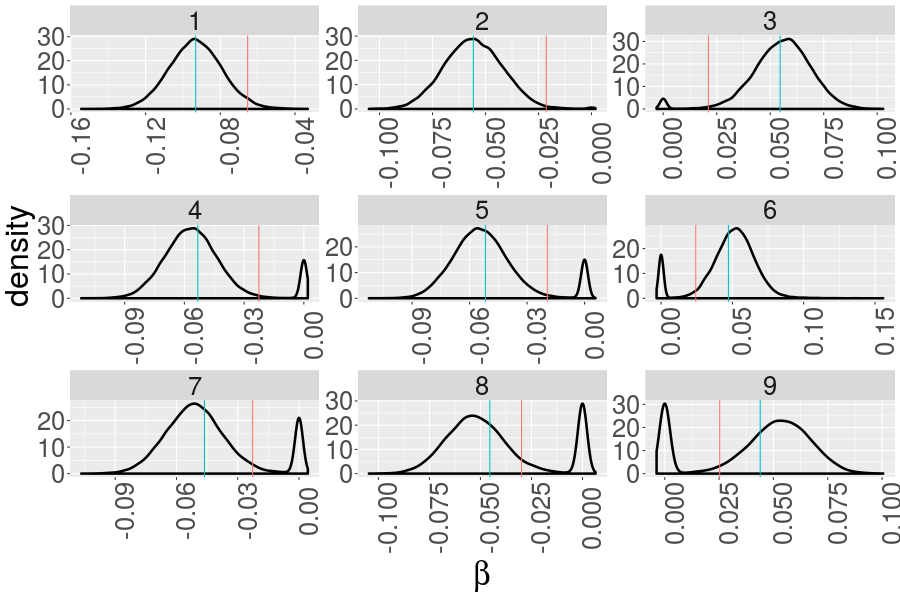}
 \caption{Density plots for the 9 entries of $\beta$ with largest estimated posterior mean along with $\hat{\E}[\beta_j \mid y]$ from Horseshoe (blue) and $\hat{\beta}_j$ from Lasso (red).} \label{fig:beta_marginals_maize}
\end{figure}

Figure \ref{fig:lasso_horse_compare} plots the number of entries of $\beta$ for which the absolute value of the corresponding Lasso or Horseshoe point estimates exceed a threshold   between $0.0005$ and $0.1$. Also shown is the size of the intersection of these two sets. For larger thresholds, the number of Horseshoe point estimates exceeding the threshold is typically larger than that for Lasso, while for smaller thresholds, this trend is reversed. This is again consistent with the tendency of Lasso to overshrink large signals and undershrink small signals \citep{bhadra2017lasso}. The size of the intersection closely tracks the minimum size of the two sets, suggesting that Lasso and Horseshoe largely agree as to which coefficients represent signals, but disagree somewhat about their magnitude. Of course, Lasso provides no notion of uncertainty in the selected variables such as that conveyed by the posterior marginals of Horseshoe.  

\begin{figure}[h]
\centering
 \includegraphics[width=0.4\textwidth]{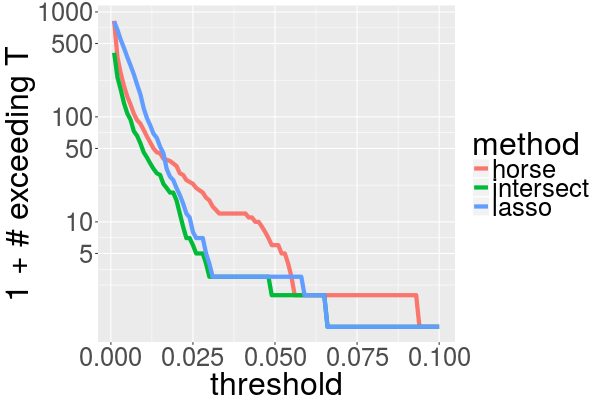}
 \caption{Plot of the number of variables for which $\hat{\E}[\beta_j \mid y]>T$ (Horseshoe) or $\hat{\beta}_j > T$ (where $\hat{\beta}_j$ is the Lasso estimate) vs $T$ (threshold) for $T \ge 0.0005$. } \label{fig:lasso_horse_compare}
\end{figure}

\section{Discussion}
It is now ten years since the Bayesian Lasso and the associated Gibbs sampling algorithm were proposed in \cite{park2008bayesian}, eight years since the Horseshoe prior appeared in \cite{carvalho2010horseshoe}, and 22 years since the landmark Lasso paper of \citet{tibshirani1996regression}. 
While the introduction of the least angle regression algorithm \citep{efron2004least}, and, more recently, coordinate descent algorithms \citep{friedman2010regularization}, have made $L_1$ regularized regression with hundreds of thousands of predictors possible on standalone computing hardware, no existing implementation of Bayes Lasso, Horseshoe, or any other Bayesian global-local shrinkage prior scales to this problem size. This has probably limited the adoption of these attractive Bayesian methods by practitioners, especially in the biological sciences where large $p$ is common. Regardless of the virtues of a statistical procedure, it is of little use practically if it is not computable.

Here we have for the first time offered computational algorithms for Horseshoe that can scale to hundreds of thousands of predictors. The algorithms have strong theoretical convergence and approximation error guarantees. Our approximate algorithm has the same computational cost per step as coordinate descent for elastic net and Lasso when the truth is sparse, though naturally more computation time is required to obtain a Markov chain of the requisite length than to obtain a single path of Lasso solutions. However, one gains more information from the Horseshoe, perhaps most critically some measure of uncertainty in which $\beta_j$ correspond to true signals. The Bayesian community has long recommended against selecting single models without reporting uncertainty in which model is selected, but has often not provided algorithms that implement the recommended methods in large $p$ problems. This has perhaps contributed to the growing importance of selective inference over Bayesian methods, as practitioners have mostly adopted the strategy of selecting a single model. We hope that the computational strategies outlined here will contribute to growth in the use of Bayesian methods in high-dimensional settings, and that exploiting sparsity and other special structure of the target will be more widely adopted as a means to develop accurate approximate MCMC algorithms for modern applications. 

\subsection*{Acknowledgements}
The authors thank Professors Xiaolei Liu and Xiang Zhou for sharing the Maize GWAS data. Prof. Bhattacharya's research is supported by the National Science Foundation (NSF DMS 1613156) and an NSF CAREER award (DMS 1653404).

%\newpage
%\FloatBarrier
\appendix
%\section*{Appendix}
\section{Preliminaries}
We introduce notation and make several observations that are used throughout the appendix. For a square matrix $A$, $\tr(A)$ denotes its trace. We use $I_d$ to denote the $d \times d$ identity matrix. For an $m \times r$ matrix $A$ (with $m > r$), $s_i(A) :\,= s_i = \sqrt{\lambda_i}$ for $i = 1, \ldots, r$ denote the singular values of $A$, where $\lambda_1 \ge \lambda_2 \ge \ldots \ge \lambda_r \ge 0$ are the eigenvalues of $A' A$. The largest and smallest (non-zero) singular values are $s_{\max}(A) = s_1(A)$ and $s_{\min}(A) = s_r(A)$. Unless otherwise stated, $\|A\| :\,=s_{\max}(A)$ denotes the operator norm of a matrix. We often make use of the standard facts $\|A B\| = \|B A\| \le \|A\| \, \|B \|$, $\|A + B\| \le \|A\| + \|B\|$, and $\|A^{-1}\| = 1/s_{\min}(A)$. We use $\succeq$ to denote the partial order on the space of nonnegative definite (nnd) matrices, i.e., $A \succeq B$ implies $(A-B)$ is nnd. We also record the fact that if $A, B$ are nnd matrices of the same size, then $ABA$ is also nnd. 

\ 

For probability measures $P, Q$ on $(\mathcal{X}, \mathcal{B})$ having densities $p$ and $q$ with respect to some dominating measure $\nu$, recall the following equivalent definitions of the total variation distance 
$$
\|P - Q\|_{\TV} = \sup_{B \in \mathcal{B}} \big| P(B) - Q(B) \big| = \frac{1}{2} \, \int_{\mathcal{X}} |p-q| d\nu = \sup_{|\phi| < 1} \int \phi(p-q) d\nu. 
$$
The Kullback--Leibler (KL) divergence $\mbox{KL}(P\,||\,Q) = \int p \log(p/q) d\nu$. From Pinsker's inequality, $\mbox{KL}(P\,||\,Q) \ge 2 \|P-Q\|_{\TV}^2$. 

\

Define the multivariate normal inverse-gamma (MNIG) distribution to be the joint distribution of $(\beta, \sigma^2) \in \mathbb{R}^p \otimes \bb R_+$ defined by the hierarchy 
\begin{align}\label{eq:mnig}
\beta \mid \sigma^2 \sim N(\mu, \sigma^2 \, \Sigma), \quad \sigma^2 \sim \mbox{InvGamma}(a, a'),
\end{align}
where an $\mbox{InvGamma}(a, a')$ distribution has density proportional to \\$x^{-(a+1)} e^{-a'/x} \1_{(0, \infty)}(x)$. 
We denote the above distribution by MNIG$(\mu, \Sigma, a, a')$. 

We record a lemma which calculates the KL divergence between two MNIG distributions with the same shape parameter; a proof is provided in Appendix \ref{sec:KL_lemma}.  
\begin{lemma}\label{lem:KL_MNIG}
Suppose $p_i \sim \mathrm{MNIG}(\mu_i, \Sigma_i, a_i, a'_i)$ for $i = 0, 1$, with $a_0 = a_1$. Then, 
\be
\mathrm{KL}(p_0\,||\,p_1) = & \frac{1}{2} \bigg[ \tr(\Sigma_1^{-1} \Sigma_0 - I_p) - \log |\Sigma_1^{-1} \Sigma_0| + (\mu_1 - \mu_0)'\Sigma_1^{-1}(\mu_1 - \mu_0) \, \frac{a_0}{a'_0} \bigg] \\
& + a_0 \log(a'_0/a'_1) + \frac{(a'_1 - a'_0) a_0}{a'_0}.
\ee
\end{lemma}

\subsection{Derivation of mean and covariance for $\beta$ in the approximate chain}
Let us first derive $\mu_\delta$ in \eqref{eq:m_del}. Recalling the definition of $D_\delta$, we have $\Gamma_\delta W' = (\Gamma_S W_S'; 0_{p-s_\delta\times N})$ and $W \Gamma_\delta W' = W_S \Gamma_S W_S'$. Thus, $\mu_\delta = (\mu_S; 0_{p-s_\delta \times 1})$, with 
$$
\mu_S = \Gamma_S W_S' (I_N + W_S \Gamma_S W_S')^{-1} z = (W_S' W_S + \Gamma_S^{-1})^{-1} W_S' z.
$$
A proof of the second equality can be found in the proof of Proposition 1 in \cite{bhattacharya2016fast}. 

We now derive $\Sigma_\delta$ in \eqref{eq:Sig_del}. Again, using the definition of $D_\delta$, we have $u - \Gamma W' M_\delta^{-1} v = (u_S - \Gamma_S W_S' M_S^{-1} v \, ; \, u_{S^c})$, where $M_S = (I_N + W_S \Gamma_S W_S')$. Also, recall that $v = W u + f = W_S u_S + W_{S^c} u_{S^c} + f$, and $u_S \ci u_{S^c}$ since $\Gamma$ is diagonal, which together imply $\mbox{cov}(u_S, v) = \Gamma_S W_S'$ and $\mbox{cov}(u_{S^c}, v) = \Gamma_{S^c} W_{S^c}'$. We now derive the blocks of $\Sigma_\delta$. \\
1. We have, 
$$
\mbox{cov}(u_S - \Gamma_S W_S' M_S^{-1} v) = \Gamma_S - \Gamma_S W_S' M_S^{-1} W_S \Gamma_S = (W_S' W_S + \Gamma_S^{-1})^{-1}, 
$$
where the proof of the second equality can be found in the proof of Proposition 1 in \cite{bhattacharya2016fast}. \\
2. Next, using $u_S \ci u_{S^c}$, 
$$
\mbox{cov}(u_S - \Gamma_S W_S' M_S^{-1} v, u_{S^c}) = - \Gamma_S W_S' M_S^{-1} W_{S^c} \Gamma_{S^c}.
$$
3. Finally, $\mbox{cov}(u_{S^c}) = \Gamma_{S^c}$. 

\section{Transition densities of the exact and approximate chain}\label{sec:App_exap}
We lay down the transition densities of the exact and approximate chains. Recall $D = \mbox{diag}(\eta_j^{-1})$ and $\Gamma = \xi^{-1} D$. \\[2ex]
{\bf Exact chain.} First, a comment on the state space for the Markov chain(s) of interest. $\P(x,\cdot)$ is not defined for $x$ in the set 
\be
\mc X_0 = \{ x = (\beta,\sigma^2,\xi,\eta) : \beta_j = 0 \text{ for one or more } j\} 
\ee
Thus, we exclude $\mc X_0$ from the state space, and construct a Lyapunov function that is infinite on this set, so that points in this set are not on the boundary of sublevel sets of the Lyapunov function. By \cite[Remark 1.1]{hairer2011yet}, the Lyapunov condition and minorization on its sublevel sets are sufficient to establish exponential convergence toward a unique invariant measure. In particular, \cite[Theorem 3.1]{hairer2011yet} requires only that the state space $\X$ be a measurable space. Since $\mc X_0$ has $\nu \P^k$-measure zero for any $k > 0$ whenever it has $\nu$-measure zero, computational problems are avoided by simply not initializing the Markov chain at a point in $\mc X_0$.

Define $\R_{\setminus 0} = \R \setminus \{0\}$. Consider the Markov transition kernel $\P$ for our exact algorithm on state space $\X = \X_1 \times \X_2 = \R_+^p \times (\R_{\setminus 0}^p  \times \R_+ \times \R_+)$ with state variable $x = (\eta, x_{\backslash \eta})$, where $x_{\backslash \eta} = (\beta, \sigma^2, \xi)$.
Letting $x = (\tilde \eta, \tilde \beta, \tilde \sigma^2, \tilde \xi)$ denote the current state and $y = (\eta,\beta,\sigma^2,\xi)$ the new state, the transition kernel $\P(x,\cdot)$ has density with respect to Lebesgue measure
\be \label{eq:TransDens}
p(x,y) = p\big((\tilde \eta, x_{\backslash \eta}), (\eta, y_{\backslash \eta})\big) = p_1(\eta \mid x_{\backslash \eta}) \, p_2(y_{\backslash \eta} \mid \eta, \tilde \xi),
\ee
where
\be
\begin{aligned}
p_1(\eta \mid x_{\backslash \eta}) &= \prod_{j=1}^p p_1(\eta_j \mid x_{\backslash \eta}), \\
p_1(\eta_j \mid x_{\backslash \eta}) &= \frac{e^{-\tilde m_j}}{\Gamma(0,\tilde m_j+b \tilde m_j)} \, \frac{e^{-\tilde m_j\eta_j}}{1+ \eta_j} \,  \1_{(b, \infty)}(\eta_j),\\
p_2(y_{\backslash \eta} \mid \eta, \tilde \xi) &= p_2(\beta \mid \sigma^2, \xi, \eta) \, p_2(\sigma^2 \mid \xi, \eta) \, p_2(\xi \mid \eta, \tilde \xi),
\end{aligned}
\label{eq:betaeta_reg}
\ee
where $\tilde m_j = \tilde \xi \tilde \beta_j^2/(2 \tilde \sigma^2)$ and $\Gamma(\cdot, \cdot)$ is the incomplete Gamma function defined in \eqref{eq:incomp_gam}. We describe the various components of $p_2$ below. 

The transition kernel of the MH-within-Gibbs update for $\xi$ can be written as 
\begin{align}\label{eq:transer_MH}
p_2(\xi \mid \eta, \tilde \xi) = \alpha_\eta(\tilde \xi, \xi) \, h(\xi \mid \tilde \xi) + r_{\eta}(\tilde \xi) \, \delta_{\tilde \xi}(\xi),
\end{align}
where $h(\cdot \mid \cdot)$ is the log-normal proposal kernel for $\xi$. Here, 
\be
\alpha_\eta(\tilde \xi, \xi) = \min\left\{1, q_\eta(\tilde \xi, \xi) = \frac{p(\xi \mid \eta) \xi}{p(\tilde \xi \mid \eta) \tilde \xi} \right\}
\ee
is the probability of accepting a move to $\xi$ from $\tilde \xi$, with $p(\xi \mid \eta)$ as defined in \eqref{eq:BasicQuantities}, $\delta_{\tilde \xi}(\cdot)$ denotes a point-mass at $\tilde \xi$, and
$$
r_\eta(\tilde \xi) = 1 - \int \alpha_\eta(\tilde \xi, \xi) \, h(\xi \mid \tilde \xi) \, d \xi
$$
is the probability of staying at $\tilde \xi$.  

The full conditional $p_2(\beta \mid \sigma^2, \xi, \eta)$ is $N(\mu, \sigma^2 \Sigma)$ with $\mu = \Sigma W'z$ and $\Sigma = (W'W + (\xi^{-1} D)^{-1})^{-1}$, while $p_2(\sigma^2 \mid \xi, \eta)$ has an InvGamma distribution. Thus, using the definition of MNIG above, the full conditional for $(\beta, \sigma^2)$ from the exact algorithm, $p_2(\beta, \sigma^2 \mid \xi, \eta)$, is distributed as $\mbox{MNIG}(\mu, \Sigma, a, a')$, with 
\begin{align}\label{eq:jt_exact}
\begin{aligned}
\mu & = \Sigma W'z = \Gamma W' M^{-1} z, \ \Sigma = (W' W + \Gamma^{-1})^{-1} = \Gamma - \Gamma W' M^{-1} W \Gamma, \\
a &= (N+\omega)/2, \ a' = (z'M^{-1} z + \omega)/2.
\end{aligned}
\end{align}
In the fist line of the above display, we used various equivalent representations of $\mu$ and $\Sigma$ which follow from the Woodbury matrix identity and are encapsulated in the algorithm \eqref{eq:GaussianSamplingTrick}. 
\\[2ex]
{\bf Approximate chain.} We now describe the transition density of the approximate chain. Noting that the update for $\eta$ remains the same in the approximate algorithm, the approximate Markov kernel $\P_\epsilon(x, \cdot)$ has transition density
\begin{align}\label{eq:transker_full_approx}
p_\epsilon(x, y) = p_1(\eta \mid x_{\backslash \eta}) \, p_{2,\epsilon}(\beta, \sigma^2 \mid \xi, \eta) \, p_{2,\epsilon}(\xi \mid \eta,\tilde \xi)\quad x \in \X,
\end{align}
where $p_{2,\epsilon}(\beta, \sigma^2 \mid \xi, \eta)$ denotes the approximate full conditional of $(\beta, \sigma^2)$ resulting from the approximations of $D W$ and $WDW'$ by $D_\delta W$ and $W D_\delta W'$ described in Section \ref{sec:ApproxAlgo}, which is distributed as $\mbox{MNIG}(\mu_\delta, \Sigma_\delta, a_\delta, a_\delta')$, with 
\begin{align}\label{eq:jt_approx}
\begin{aligned}
&\mu_\delta = \Gamma_\delta W' M_{\delta}^{-1} z, \quad \Sigma_\delta = \Gamma - \big(2 \Gamma W' M_\delta^{-1} W \Gamma_\delta - \Gamma_\delta W' M_\delta^{-1} M M_\delta^{-1} W \Gamma_\delta \big), \\
& a_\delta = (N+\omega)/2, \quad a_\delta' = (z'M_{\delta}^{-1} z + \omega)/2.
\end{aligned}
\end{align}
Also, 
\begin{align}\label{eq:transer_MH}
p_{2, \epsilon}(\xi \mid \tilde \xi, \eta) = \alpha_{\eta, \epsilon}(\tilde \xi, \xi) \, h(\xi \mid \tilde \xi) + r_{\eta, \epsilon}(\tilde \xi) \, \delta_{\tilde \xi}(\xi),
\end{align}
is the approximate MH-within-Gibbs transition density obtained by the approximation of the acceptance probability $\alpha_{\eta, \epsilon}(\xi, \xi') = \min\big\{1, q_{\eta, \delta}(\xi, \xi') \big\}$ where 
\be \label{eq:QEtaDelta}
q_{\eta,\delta} = \frac{p_\epsilon(\xi \mid \eta) \xi}{p_\epsilon(\tilde \xi \mid \eta) \tilde \xi},
\ee
and $p_\epsilon(\xi \mid \eta)$ is obtained by replacing $M_\xi$ by $M_{\xi,\delta}$ in \eqref{eq:BasicQuantities}.

\section{Proof of Theorem 3.5}
We prove the second assertion, that is, prove geometric ergodicity of $\P$ for $b > 0$ by verifying the Lyapunov condition in Assumption \ref{ass:Lyapunov} and the minorization condition in Assumption \ref{ass:Minorization}. An inspection of the verification of the Lyapunov condition will show that the same proof continues to work for $b = 0$ with some minor modifications.

\subsection{Lyapunov condition for the exact chain}
We first show that 
\be \label{eq:LyapunovApp}
V\big(\eta, \beta, \sigma^2, \xi\big) = \frac{\|W \beta\|^2}{\sigma^2} + \sum_{j=1}^p \big[\sigma^{2c} |\beta_j|^{-2c} + \sigma^{-c} \eta_j^c  |\beta_j|^c + \eta_j^c \big] + \xi^2
\ee
is a Lyapunov function for $\P$ for any $c \in (0, 1/2)$. We have, 
\be
(\P V)(\tilde \eta, x_{\backslash \eta})
& = \int V(\eta, y_{\backslash \eta}) \, p\big((\tilde \eta, x_{\backslash \eta}), (\eta, y_{\backslash \eta})\big) d\eta dy_{\backslash \eta} \\
& = \int_{\eta} \bigg[ \int_{y_{\backslash \eta}} V(\eta, y_{\backslash \eta}) p_2(y_{\backslash \eta} \mid \eta, \tilde \xi) \,dy_{\backslash \eta} \bigg] p_1(\eta \mid x_{\backslash \eta}) \,d\eta \\
& = \int V_1(\eta,\tilde \xi) p_1(\eta \mid x_{\backslash \eta}) \,d\eta, \label{eq:Lyapunov_gibbs}
\ee
where 
\be
V_1(\eta,\tilde \xi) &= \E\left(\frac{\|W\beta\|^2}{\sigma^2} \,\big\vert\, \eta, \tilde \xi \right) + \E(\xi^2 \mid \eta, \tilde \xi) \\
&+ \sum_{j=1}^p \big[\E(\sigma^{2c} |\beta_j|^{-2c} \mid \eta, \tilde \xi) + \E(\sigma^{-c} |\beta_j|^c \mid \eta, \tilde \xi) + \eta_j^c \big]. 
\label{eq:V1_def}
\ee
We now aim to bound $V_1(\eta,\tilde \xi)$. We shall make repeated use of the fact that 
\be 
\E[ g(\beta) h(\sigma^2) \mid \eta] = \E \bigg[ h(\sigma^2) \, \E[g(\beta) \mid \sigma^2, \xi, \eta] \,\big\vert\, \eta\bigg]
\ee
for integrable functions $g$ and $h$, using the tower property of conditional expectations. 

\

We begin by integrating over $\beta$ to bound $\E(\|W\beta\|^2 \mid \sigma^2, \xi, \eta)$, $\E(|\beta_j|^{-2c} \mid \sigma^2, \xi, \eta)$, and $\E(|\beta_j|^c \mid \sigma^2, \xi, \eta)$, respectively. We first show that 
\be\label{eq:Wb_bd}
\E(\|W\beta\|^2 \mid \sigma^2, \xi, \eta) \le \|z\|^2 + N \sigma^2.
\ee
To that end, we have, from \eqref{eq:jt_exact},
{\small\be
\E(\|W \beta\|^2 \mid \eta) 
= \mu' W' W \mu + \sigma^2 \, \mbox{tr}[(W'W) \Sigma] 
= \| W \Sigma W' z\|^2 + \sigma^2 \, \mbox{tr}[W \Sigma W']. 
\ee}
Let us now calculate $W \Sigma W'$. Recall, 
\be 
\Gamma = \xi^{-1} D, \quad M = I_N + W \Gamma W', \quad \Sigma = \Gamma - \Gamma W' M^{-1} W \Gamma,
\ee
where the last equality follows from the Woodbury matrix identity. Then, 
\be 
W \Sigma W' 
&= W \Gamma W' - W \Gamma W' M^{-1} W \Gamma W'  = W \Gamma W'  \big[I_N - M^{-1} W \Gamma W' \big] \\
& = W \Gamma W' M^{-1}  = I_N - M^{-1},
\ee
where we have used that $M^{-1} W \Gamma W' = W \Gamma W'  M^{-1} = I_N - M^{-1}$. We then have 
$\| W \Sigma W' z\|^2 = \|(I_N - M^{-1})z\|^2 \le \|z\|^2$, and $\mbox{tr}(W \Sigma W') \le N$, delivering \eqref{eq:Wb_bd}. 

\

We next focus on $\E(|\beta_j|^{-2c} \mid \sigma^2, \xi, \eta)$ and show that for universal constants $0<C_1,C_2 < \infty$,
\be\label{eq:betapowbd1}
\E \big(|\beta_j|^{-2c} \mid \sigma^2, \xi, \eta) \le \sigma^{-2c}(C_1 \eta_j^c + C_2). 
\ee 
A formula for negative absolute moments of Gaussians is available from \cite{gradshteyn2014table} and recorded in equation \eqref{eq:normal_invmom} in Section \ref{sec:NormalHypergeo}. Specifically, let $\mu_j$ and $\sigma_j^2$ respectively denote the $j$th entry of $\mu$ and the $j$th diagonal entry of $\Sigma$ in \eqref{eq:jt_exact}. We then have, from \eqref{eq:normal_invmom}, that 
\be%\label{eq:betapow2c}
\E \big(|\beta_j|^{-2c} \mid \sigma^2, \xi, \eta) = \sigma^{-2c} \, \sigma_j^{-2c} \, \frac{2^{-c} \, \Gamma\left( \frac{1-2c}{2} \right)}{\sqrt{\pi}} \, e^{-\frac{\mu_j^2}{2\sigma^2 \sigma_j^2}} \, M\left( \frac{1-2c}2, \frac12; \frac{\mu_j^2}{2\sigma^2 \sigma_j^2} \right),
\ee
where $M(\cdot, \cdot; \cdot)$ is the confluent hypergeometric function of the first kind; see Section \ref{sec:NormalHypergeo} for definition and properties. Since $c \in (0, 1/2)$, the condition of Lemma \ref{lem:conf_dec} is satisfied, so that we can bound
\be 
e^{-\frac{\mu_j^2}{2\sigma^2 \sigma_j^2}} \, M\left( \frac{1-2c}2, \frac12; \frac{\mu_j^2}{2\sigma^2 \sigma_j^2} \right) \le 1. 
\ee  
This implies 
\be\label{eq:betapowbd}
\E \big(|\beta_j|^{-2c} \mid \sigma^2, \xi, \eta) \le \tilde C_1 \sigma^{-2c} \, \sigma_j^{-2c},
\ee
where $\tilde C_1 = \pi^{-1/2} 2^{-c} \Gamma(1/2-c)$. 

To bound the right hand side of \eqref{eq:betapowbd}, we need a lower bound on $\sigma_j^2$. 
To that end, we have $s_{\max}^2(W) \, I_p + (\xi^{-1} D)^{-1} \succeq W'W + (\xi^{-1} D)^{-1}$, where $A \succeq B$ denotes $(A-B)$ is nonnegative definite (nnd). Using the fact that $A \succeq B$ implies $B^{-1} \succeq A^{-1}$, we have $\Sigma \succeq ( s_{\max}^2(W) \, I_p + (\xi^{-1} D)^{-1} )^{-1}$. Next, use the fact that if $A \succeq B$, then $a_{jj} \ge b_{jj}$, since $(a_{jj} - b_{jj}) = e_j' (A-B) e_j \ge 0$ with $e_j$ the $j$th unit vector. This implies
\be 
\sigma_j^2 \ge \frac{1}{s_{\max}^2(W) + \xi \eta_j}, \quad \sigma_j^{-2c} \le (s_{\max}^2(W) + \xi \eta_j)^c \le (s_{\max}^{2c}(W) + b_\xi^c \eta_j^c),
\ee
where in the last step, we used that for $a, b > 0$ and $c \in (0, 1/2)$, $(a+b)^c \le (a^c + b^c)$, and that $\xi \le b_\xi$. Substitute the bound in  \eqref{eq:betapowbd} to obtain \eqref{eq:betapowbd1}.

\

Next, we consider $\E(|\beta_j|^c \mid \sigma^2, \xi, \eta)$, and show that there exist universal constants $0<C_3,C_4<\infty$ such that
\be\label{eq:betapospow}
\sum_{j=1}^p \eta_j^c \E(|\beta_j|^c \mid \sigma^2, \xi, \eta) \le C_3 \sigma^c \,\sum_{j=1}^p (\eta_j^c+1) + C_4.
\ee

We have, using that $x \mapsto x^{2/c}$ is convex for $x > 0$, that $\big[ \E(|\beta_j|^c \mid \sigma^2,\xi,\eta)\big]^{2/c} \le \E(\beta_j^2 \mid \sigma^2, \xi, \eta) = \mu_j^2 + \sigma^2 \sigma_j^2$, and hence, $\E(|\beta_j|^c \mid \sigma^2,\xi,\eta) \le (\mu_j^2 + \sigma^2 \sigma_j^2)^{c/2} \le |\mu_j|^c + \sigma^c (\sigma_j^2)^{c/2}$. Following a similar argument as in the paragraph after \eqref{eq:betapowbd}, $(W'W + (\xi^{-1} D)^{-1})^{-1} \preceq \xi^{-1} D$, implying $\sigma_j^2 \le (\xi \eta_j)^{-1}$. These together imply, 
\be 
\sum_{j=1}^p \eta_j^c \E(|\beta_j|^c \mid \sigma^2,\xi,\eta) \le \sum_{j=1}^p \eta_j^c \{ |\mu|_j^c + \sigma^c \, (\sigma_j^2)^{c/2} \} \le \sum_{j=1}^p \eta_j^c |\mu_j|^c + C \sigma^c \, \sum_{j=1}^p \eta_j^{c/2}. 
\ee
for a universal constant $C$. Next, using H\"{o}lder's inequality, 
\be
\sum_{j=1}^p \eta_j^c |\mu_j|^c \le \Big[\sum_{j=1}^p (\eta_j^c |\mu_j|^c)^{2/c} \Big]^{c/2} \, p^{1-c/2} = p^{1-c/2} \, (\|D^{-1} \mu\|^2)^{c/2}. 
\ee
We have $\|D^{-1} \mu\|^2 = w' D^{-1}\Sigma^2 D^{-1} w$, with $w = W' z$. Since $\Sigma^2 = (W'W + (\xi^{-1} D)^{-1})^{-2} \preceq \xi^2 D^2$, we have $D^{-1} \Sigma^2 D^{-1} \preceq \xi^{-2} I_p$, where we have used that if $B_1 \preceq B_2$, and $B_1, B_2, A$ are positive definite (pd), then $A B_1 A \preceq A B_2 A$. This implies $\|D^{-1} \mu\|^2 \le \xi^{-2} \|w\|^2 \le a_\xi^{-2} \|w\|^2$. 
Cascading this bound through the previous two displays, and using the inequality $x^{c/2} \le (x^c + 1)$ for $x > 0$, \eqref{eq:betapospow} is obtained. 

\

Combining the bounds \eqref{eq:betapospow}, \eqref{eq:betapowbd1}, and \eqref{eq:Wb_bd}, we obtain for universal constants $0<C_5,C_6,C_7,C_8<\infty$ not depending on $\eta$ 
\be 
\E[V \mid \sigma^2, \xi, \eta] \le C_5 \sum_{j=1}^p \eta_j^c + \frac{C_6}{\sigma^2} + \frac{C_7}{\sigma^c} + \xi^2 + \tilde C_7 
\ee
Next, take an expectation w.r.t. $\sigma^2 \mid \xi, \eta$. Note that $\E(1/\sigma^2 \mid \xi, \eta) = (N+a)/(z'M^{-1}z + b) \le (N+a)/b$, and similarly, $\E(1/\sigma^c \mid \xi, \eta)$ is also bounded above by a constant not depending on $\xi, \eta$. 
This leads to 
\be \label{eq:V1_bded}
V_1(\eta,\tilde \xi) = \E[V \mid \eta, \tilde \xi] \le C_5 \sum_{j=1}^p \eta_j^c + C_8, 
\ee
for a universal constant $C_8$ not depending on $\tilde \xi$, where we additionally used that $\xi$ is compactly supported. Notice that while $V_1$ is a function of $\tilde \xi$, the upper bound on the right side is not, a consequence of the fact that $\xi \in [a_\xi, b_\xi]$ for $0<a_\xi<b_\xi<\infty$.

\

We now proceed to bound $\E\big(V_1(\eta,\tilde \xi) \mid x_{\backslash \eta}) = \int V_1(\eta,\tilde \xi) \, p_2(\eta \mid x_{\backslash \eta})$. For a small $\varepsilon > 0$ to be chosen later, bound
\be 
\int \eta_j^c \, p(\eta_j \mid x_{\backslash \eta}) \, d\eta_j
& = \frac{ e^{-\tilde m_j} }{ \Gamma(0, \tilde m_j + b \tilde m_j) } \, \int_b^{\infty} \eta_j^c \, \frac{e^{-\tilde m_j \eta_j}}{1 + \eta_j} \, d\eta_j \\
& \le \frac{ e^{-\tilde m_j} }{ \Gamma(0, \tilde m_j + b \tilde m_j) } \,  \int_b^{\infty} \eta_j^{c-1} e^{-\tilde m_j \eta_j} \, d\eta_j \\
& = \frac{e^{-\tilde m_j}}{m_j^c} \, \frac{\Gamma(c, b \tilde m_j)}{\Gamma(0, \tilde m_j + b \tilde m_j)} \\
& \le \varepsilon \tilde m_j^{-c} + C_{\varepsilon} \\
& \le (a_\xi/2)^{-c} \varepsilon \, \tilde \sigma^{2c} |\tilde \beta_j|^{-2c} + C_{\varepsilon}
\ee
In the first inequality, we used the bound $\eta_j/(1+\eta_j) < 1$ for $\eta_j \in (b, \infty)$, while the penultimate inequality follows from Lemma \ref{lem:imcomp_gamrat2}. From \eqref{eq:V1_bded}, we then obtain 
\be 
\E\big(V_1(\eta,\tilde \xi) \mid x_{\backslash \eta}) \le \sum_{j=1}^p (a_\xi/2)^{-c} C_5 \, \varepsilon \, \sigma^{2c} |\beta_j|^{-2c} + C_9. 
\ee
Now pick $\varepsilon$ such that $(a_\xi/2)^{-c} C_5 \, \varepsilon < 1$, and we have proved that $V$ is Lyapunov. 

\subsection{Minorization condition on sublevel sets}
Consider sublevel sets of the Lyapunov function in \eqref{eq:Lyapunov}:
\be
\mc S(R) := \left\{ x :  \frac{\|W \beta\|^2}{\sigma^2} + \sum_{j=1}^p \big[\sigma^{2c} |\beta_j|^{-2c} + \sigma^{-c} \eta_j^c  |\beta_j|^c + \eta_j^c \big] + \xi^2 < R \right\}
\ee
where $D = \diag(\eta^{-1})$. Consider two points $x,y \in \X$. Observe that
\be
\|\delta_x \P - \delta_y \P\|_{\TV} &= \int |p_1(\eta \mid x_{\setminus \eta}) p_2(z_{\setminus \eta}  \mid \eta, \tilde \xi) - p_1(\eta \mid z_{\setminus \eta}) p_2(z_{\setminus \eta} \mid \eta, \tilde \xi) | d z_{\setminus \eta} d\eta  \\
&= \int  p_2(z_{\setminus \eta} \mid \eta, \tilde \xi) |p_1(\eta \mid x_{\setminus \eta}) - p_1(\eta \mid y_{\setminus \eta}) | d z_{\setminus \eta} d\eta \\
&= \int |p_1(\eta \mid x_{\setminus \eta}) - p_1(\eta \mid y_{\setminus \eta}) | d\eta \\
&= \|\delta_{x_{\setminus \eta}} \P_1 - \delta_{y_{\setminus \eta}} \P_1 \|_{\TV},
\ee
the total variation distance just between the $\eta$ conditionals started at two points $x_{\setminus \eta},y_{\setminus \eta}$. 

Let 
\be
S_{\setminus \eta}(R) =  \{x_{\setminus \eta} \in \X_2 : (x_{\setminus \eta},\eta) \in \mc S(R) \text{ for some } \eta \in \X_1   \}
\ee
Consider a point $x_{\setminus \eta} \in \mc S_{\setminus \eta}(R)$. Any such point must satisfy 
\be
\frac{|\beta_j|^c}{\sigma^c} \eta^c &< R \Rightarrow \frac{\beta_j^2 \xi}{\sigma^2} < b_{\xi} R^{2/c} \eta^{-2} \quad j=1,\ldots,p \\
\frac{\sigma^{2c}}{|\beta_j|^{2c}} &< R \Rightarrow \frac{\beta_j^2 \xi}{\sigma^2} > a_{\xi} R^{-1/c} \quad j=1,\ldots,p \\
\text{if } N \le p \text{ then } \frac{\|W \beta\|^2}{\sigma^2} &< R \Rightarrow \frac{\beta_j^2 \xi}{\sigma^2} < b_\xi R \quad j=1,\ldots,p
\ee
It follows that when $N \le p$, $a_\xi R^{-1/c} < \beta_j^2 \xi \sigma^{-2} < b_\xi R$ for every $x_{\setminus \eta} \in \mc S_{\setminus \eta}(R)$. Moreover, if $p > N$ and the prior on $\eta$ is truncated below by $b$, then we have $a_\xi R^{-1/c} < \beta_j^2 \xi \sigma^{-2} < b_\xi R^{2/c} b^{-2}$ for every $x_{\setminus \eta} \in \mc S_{\setminus \eta}(R)$. 

The remainder of the proof uses the upper bound $b_\xi R^{2/c} b^{-2}$ from the $p > N$ case; the proof for the $p \le N$ case is virtually identical and omitted. Define the interval
\be
I(R) = \left[ \frac12 a_\xi R^{-1/c}, \frac12 b_\xi R^{2/c} b^{-2} \right]
\ee
and collection of densities corresponding to the full conditional of $\eta_j \mid x_{\setminus \eta}$ for a generic $\eta_j$
\be
\mc F(R) &= \left\{f_{m_j}(\eta_j) = \frac{e^{-m_j}}{\Gamma(0,m_j(1+b))} \, \frac{e^{-m_j \eta_j}}{1+ \eta_j} \1\{\eta_j > b\}, m_j \in I(R)  \right\},
\ee
and recall that $m_j = \beta_j^2 \xi \sigma^{-2}$. We have that for any $m_j \in I(R)$
\be
\frac{e^{-m_j}}{\Gamma(0,m_j(1+b))} \, \frac{e^{-m_j \eta_j}}{1+ \eta_j} \1\{\eta_j > b\} \ge 
\frac{e^{-\frac12 b_\xi R^{2/c} b^{-2}}}{\Gamma(0,\frac12 a_\xi R^{-1/c}(1+b))} \, \frac{e^{-\frac12 b_\xi R^{2/c} b^{-2} \eta_j}}{1+ \eta_j} \1\{\eta_j > b\}
\ee
and since the function $e^{-cx}/(1+x)$ for $c>0$ is monotone decreasing in $x$, it follows
\be
\inf_{\substack{\eta_j \in (b,b+1] \\ m_j \in I(R)}} f_{m_j}(\eta_j) \ge \frac{e^{-\frac12 b_\xi R^{2/c} b^{-2}}}{\Gamma(0,\frac12 a_\xi R^{-1/c}(1+b))} \, \frac{e^{-\frac12 b_\xi R^{2/c} b^{-2} (b+1)}}{2+b} \1\{\eta_j > b\}.
\ee
Now since the transition density corresponding to $\P_1$ can be written $p_1(x_{\setminus \eta}, \eta) = \prod_{j=1}^p f_{m_j}(\eta_j)$, we have, with $m = (m_1,\ldots,m_p)$,
\be
\inf_{\substack{\eta_j \in (b,b+1]^p \\ m_j \in I(R)^p}} p_1(x_{\setminus \eta},\eta) \ge \frac{e^{-\frac{p}2 b_\xi R^{2/c} b^{-2}}}{\Gamma^p(0,\frac12 a_\xi R^{-1/c}(1+b))} \, \frac{e^{-\frac{p}2 b_\xi R^{2/c} b^{-2} (b+1)}}{(2+b)^p} \equiv C(R)>0.
\ee
Define
\be
\mc M(R) = \{ m : m_j = \beta_j^2 \xi \sigma^{-2} \text{ for some } (\beta,\xi,\sigma^2) \in \mc S_{\setminus \eta}(R) \},
\ee
and observe that $\mc M(R) \subset I(R)^p$. It follows that
\be
\inf_{x_{\setminus \eta}, y_{\setminus \eta} \in \mc S_{\setminus \eta}(R)} \int_{(b,b+1]^p} (p_1(x_{\setminus \eta},\eta) \wedge p_1(y_{\setminus \eta},\eta)) d\eta \ge C(R) \int_{(b,b+1]^p} d\eta = C(R),
\ee
so
\be
\sup_{x_{\setminus \eta},y_{\setminus \eta} \in \mc S_{\setminus \eta}(R)} \|\delta_{x_{\setminus \eta}} \P_1 - \delta_{y_{\setminus \eta}} \P_1\|_{\TV} &= 1 - \inf_{x_{\setminus \eta}, y_{\setminus \eta} \in \mc S_{\setminus \eta}(R)} \int_{(b,\infty)^p} (p_1(x_{\setminus \eta},\eta) \wedge p_1(y_{\setminus \eta},\eta)) d \eta \\
&\le 1 - \inf_{x_{\setminus \eta}, y_{\setminus \eta} \in \mc S_{\setminus \eta}(R)} \int_{(b,b+1]^p} (p_1(x_{\setminus \eta},\eta) \wedge p_1(y_{\setminus \eta},\eta)) d \eta \\
&\le 1-C(R) < 1,
\ee
completing the proof. 

\section{Proof of Theorem 3.13} \label{sec:Theorem1Proof}
We first show that $V$ continues to define a Lyapunov function for the approximate chain $\P_\epsilon$, and then bound the total variation distance between the exact and approximate transition densities. 
\subsection{Lyapunov condition for approximate chain}
The proof of this part to a large extent closely resembles the first part of the proof of Theorem \ref{thm:Gergo}, and we only point out the key features. The only place where one requires more work is to bound the trace term of $E(\|W \beta\|^2 \mid \sigma^2, \xi, \eta) = \mu_\delta' W'W \mu_\delta + \sigma^2 \, \tr(W \Sigma_\delta W')$ under $\P_\epsilon$. Proceeding as before, we can show $\mu_\delta' W'W \mu_\delta = \|(I_N - M_\delta^{-1})z\|^2 \le \|z\|^2$. Write $\tr(W \Sigma_\delta W') = \tr(W \Sigma W') + \tr(W \Delta W')$, where $\Delta = \Sigma_\delta - \Sigma$. We have already showed in the proof of Theorem \ref{thm:Gergo} that $\tr(W \Sigma W')$ is small. For the other term, write $\tr(W \Delta W') = \tr(W \Sigma^{1/2} \Sigma^{-1/2} \Delta \Sigma^{-1/2} \Sigma^{1/2} W')$. We prove in the next subsection (see equation \ref{eq:siginvdel_final}) that
\be 
\|\Sigma^{-1} \Delta\|_2 \le 8 \|W\|^2 \delta + O(\delta^2).
\ee
Since $\Sigma^{-1/2}\Delta \Sigma^{-1/2}$ is similar to $\Sigma^{-1} \Delta$, this means $\Sigma^{-1/2}\Delta \Sigma^{-1/2} \preceq C I_p$ for some constant $C>0$. This means $\tr(W \Sigma^{1/2} \Sigma^{-1/2} \Delta \Sigma^{-1/2} \Sigma^{1/2} W') \le C \tr(W \Sigma W')$, which we already know is bounded above by a constant. 

The other fact used to complete the proof is that the same two-sided bound for $\sigma_j^2$ continues to hold as before. To see this, for $j \notin S$, $\sigma_j^2 = \xi^{-1} \eta_j^{-1}$, while for $j \in S$, $\sigma_j^2 \ge 1/(s_{\max}^2(W_S) + \xi \eta_j) \ge 1/(s_{\max}^2(W) + \xi \eta_j)$, and the upper bound $\sigma_j^2 \le (\xi \eta_j)^{-1}$ holds for all $j$.

\subsection{Proof of uniform total variation bound in \eqref{eq:PointwiseEtaBeta}}
Recall we denote $x = (\beta, \sigma^2, \xi, \eta)$ for the entire state vector. We shall also call $\theta = (\beta, \sigma^2)$. The various pieces of the transition density for the exact algorithm is given in \eqref{eq:TransDens} -- \eqref{eq:jt_exact}, while the same for the approximate algorithm is given in \eqref{eq:transker_full_approx} -- \eqref{eq:QEtaDelta}.

We now proceed to bound the total variation distance between $\P(x, \cdot)$ and $\P_\epsilon(x, \cdot)$. We have, for a fixed $x \in \X$, 
\begin{align*}
& 2 \|\P(x, \cdot) - \P_\epsilon(x, \cdot) \|_{\TV} = \int | p(x, y) - p_\epsilon(x, y) | dy \\
=& \int p_1(\eta \mid x_{\setminus \eta}) \left\{ \int |p_2(y_{\setminus \eta} \mid \eta, \tilde \xi) - p_{2,\epsilon}(y_{\setminus \eta} \mid \eta, \tilde \xi)| d y_{\setminus \eta} \right\} d\eta \\
\le & \sup_\eta \int |p_2(y_{\setminus \eta} \mid \eta, \tilde \xi) - p_{2,\epsilon}(y_{\setminus \eta} \mid \eta, \tilde \xi)| d y_{\setminus \eta} \\
= & \sup_{\eta} \int \big|~ p_2(\theta \mid \xi, \eta) \, p_2(\xi \mid \eta, \tilde \xi) - p_{2,\epsilon}(\theta \mid \xi,\eta) \, p_{2,\epsilon}(\xi \mid \eta, \tilde \xi) ~\big| \, d y_{\setminus \eta}'   \\
\stackrel{(i)}{\le}& \sup_{\eta} \bigg[ \int \left\{ \int \left| p_2(\theta \mid \xi, \eta) - p_{2,\epsilon}(\theta \mid \xi, \eta)  \right| ~ d\theta \right\} \, p_2(\xi \mid \eta, \tilde \xi) \, d \xi \\
+ & \int \big| p_2(\xi \mid \eta, \tilde \xi) - p_{2, \epsilon}(\xi \mid \eta, \tilde \xi) \big| \, d \xi \bigg] \\
\stackrel{(ii)}{\le} & 2 \sup_{\xi, \eta} \| p_2(\theta \mid \xi, \eta) - p_{2,\epsilon}(\theta \mid \xi, \eta) \|_{\TV} + 2 \sup_{\xi, \tilde \xi, \eta} \big| \alpha_\eta(\tilde \xi, \xi) - \alpha_{\eta, \epsilon}(\tilde \xi, \xi) \big|.  
\end{align*}
For (i), we used triangle inequality and that $\int p_{2,\epsilon}(\theta \mid \xi, \eta) \, d\theta = 1$. For (ii), we used that 
\begin{align*}
& \int \big| p_2(\xi \mid \eta, \tilde \xi) - p_{2, \epsilon}(\xi \mid \eta, \tilde \xi) \big| \, d \xi  \\
\le & \int \big| \alpha_\eta(\tilde \xi, \xi) - \alpha_{\eta, \epsilon}(\tilde \xi, \xi) \big| \, h(\xi \mid \tilde \xi) \, d\xi + \big|r_\eta(\tilde \xi) - r_{\eta, \epsilon}(\tilde \xi) \big| \\
\le & 2 \sup_{\xi, \tilde \xi, \eta} \big| \alpha_\eta(\tilde \xi, \xi) - \alpha_{\eta, \epsilon}(\tilde \xi, \xi) \big|.
\end{align*}
Since the bound in (ii) is independent of $x$, we conclude that 
\be
\sup_{x \in \X} \|\delta_x \P - \delta_x \P_\epsilon \|_{\TV} &\le \underbrace{\sup_{\xi, \eta} \| p_2(\theta \mid \xi, \eta) - p_{2,\epsilon}(\theta \mid \xi, \eta) \|_{\TV}}_{\TV_1} \\
&+ \underbrace{\sup_{\xi, \tilde \xi, \eta} \big| \alpha_\eta(\tilde \xi, \xi) - \alpha_{\eta, \epsilon}(\tilde \xi, \xi) \big|}_{\TV_2}. \label{eq:TVCombine}
\ee
We now separately bound $\TVar_1$ and $\TVar_2$. We show that 
\begin{align*}
\TVar_1^2 & = 4 \|W\|^2 \delta + \frac{N+\omega}{\omega} \|W\|^2 \delta + \frac{N}2 \, \frac{\|z\|^2}{\omega} \, \|W\|^2 \delta +  \mathcal O(\delta^2), \\
\TVar_2 & = N \, \|W\|^2 \, (1 + \|z\|^2/\omega)  \delta + \mathcal O(\delta^2),
\end{align*}
for sufficiently small $\delta$, which produce the desired bound. Since the derivations to obtain these bounds are somewhat lengthy, we split them into two different sections below.

\subsection{Bounding $\TVar_2$: MH ratio approximations for $\xi$}
We first record a couple of useful auxiliary results. The first result is a well-known eigenvalue perturbation bound due to Weyl. 
\begin{lemma}[Weyl]\label{lem:Weyl}
Let $A, E$ be $n \times n$ Hermitian matrices. Then, for $i = 1, \ldots, n$, 
$$
| \nu_i(A+E) - \nu_i(A) | \le \|E\|,
$$
where $\nu_i(A)$ denotes the $i$th eigenvalue of $A$, and $\|\cdot\|$ denotes the operator norm of a matrix. 
\end{lemma}
Next, we present a simple yet useful result to bound the difference between MH acceptance probabilities. 
\begin{lemma}\label{lem:MHprob_diff}
For any $a, b > 0$, 
$$
| \min(a, 1) - \min(b,1) | \le \max \big \{ |(a/b) - 1|, \, |(b/a) - 1| \big\} \le e^{|\Delta|} - 1, 
$$
where $\Delta = \log(a/b)$. 
\end{lemma}
\begin{proof}
First observe that $|\min(a, 1) - \min(b, 1)| \le |a - b|$, which can be verified by enumerating the 4 different cases (i) $a, b < 1$, (ii) $a < 1 < b$, (iii) $b < 1 < a$, and (iv) $a, b > 1$. In case (iv), the left hand side is 0 and the claimed bound is trivially satisfied. In the remaining cases, bound 
\be
|a - b| &= | \{ \max(a,b)/\min(a,b) \} - 1 | \, \min(a, b) \le | \{ \max(a,b)/\min(a,b) \} - 1 | \\
&\le \max \big \{ |(a/b) - 1|, \, |(b/a) - 1| \big\}. 
\ee
This proves the first part. The second part simply follows from the monotonicity of $x \mapsto e^x$. 
\end{proof}

\

As noted in Appendix \ref{sec:App_exap}, we have that 
$$
\alpha_\eta(x,y) = \min\{1, q_\eta(x,y)\}, \quad \alpha_{\eta,\epsilon}(x,y) = \min\{1, q_{\eta,\delta}(x,y)\}
$$ 
with 
\begin{align*}
q_\eta(x,y) = \frac{ |M_y|^{-1/2} \, (\omega + z' M_y^{-1} z)^{-(N+\omega)/2} }{ |M_x|^{-1/2} \, (\omega + z' M_x^{-1} z)^{-(N+\omega)/2 }} \ \frac{y \sqrt{x} \, (1+x)}{x \sqrt{y} \, (1+y)},
\end{align*}
and $q_{\eta,\delta}(x, y)$ is obtained by replacing $M_t$ by $M_{t, \delta}$, 
where, recall that 
\begin{align*}
M_t = I_N + t^{-1} \, WDW', \ M_{t, \delta} = I_N + t^{-1} \, W D_\delta W', \quad t \in \{x, y\}.
\end{align*}
It then follows from Lemma \ref{lem:MHprob_diff} that 
\begin{align*}
\big| \alpha_\eta(x,y) - \alpha_{\eta, \epsilon}(x,y) \big| \le \exp(|\Delta|) - 1,
\end{align*}
where 
\begin{align*}
& \Delta = \log \frac{ q_{\eta,\delta}(x,y)}{q_\eta(x,y)} = \Delta_1 + \Delta_2, \\
& \Delta_1 = \Delta_{1,y} - \Delta_{1,x}, \quad \Delta_{1, t} = - \frac{1}{2} \, \big[ \log |M_{t, \delta}| - \log |M_t| \big], \ t \in \{x, y\}, \\
& \Delta_2 = \Delta_{2,y} - \Delta_{2,x}, \quad \Delta_{2, t} =  - \frac{n+\omega}{2} \, \big[ \log(1 + z' M_{t, \delta}^{-1} z/\omega) - \log(1 + z' M_t^{-1} z/\omega) \big], \ t \in \{x, y\}.
\end{align*}
We shall prove below that
\begin{align}\label{eq:del_bd}
|\Delta| \le N \, \|W\|^2 \, (1 + \|z\|^2/\omega)  \delta.
\end{align}
Observe the right hand side is independent of $\xi$ and $\eta$.  

\

To establish \eqref{eq:del_bd}, we bound 
\begin{align}\label{eq:Delta_bd}
|\Delta| \le \sum_{t \in \{x, y\}} [|\Delta_{1,t}| + |\Delta_{2,t}|]. 
\end{align}
We now proceed to individually bound $|\Delta_{1,t}|$ and $|\Delta_{2,t}|$ for $t \in \{x, y\}$. 

For $t \in \{x, y\}$, we have 
\begin{align*}
\big \vert \log |M_t| - \log |M_{t, \delta}| \big \vert  
& = \bigg \vert \sum_{i=1}^N \big[\log\{1 + t^{-1}\nu_i(W D W')\} - \log\{1 + t^{-1} \nu_i(W D_{\delta}W')\} \big] \bigg \vert \\
& \le \sum_{i=1}^N \big \vert \log\{1 + t^{-1}\nu_i(W D W')\} - \log\{1 + t^{-1} \nu_i(W D_{\delta}W')\} \big \vert \\
& \le \sum_{i=1}^N \big \vert t^{-1}\nu_i(W D W') - t^{-1} \nu_i(W D_{\delta}W') \big \vert,
\end{align*}
where the last step uses the fact that the map $u \mapsto \log(1+u)$ for $u > 0$ is Lipschitz. Write 
$$
t^{-1} WDW' = t^{-1} WD_{\delta}W' + t^{-1} WD_{<\delta}W',
$$
where $D_{<\delta} = \mbox{diag}\big( (\eta_j^{-1}) \, \mathbf{1}(j \in \mathcal{I}^c)\big)$ retains the entries of $D$ which are thresholded. By Weyl's perturbation bound (see Lemma \ref{lem:Weyl}), for any $i = 1, \ldots, N$, 
$$
\big \vert t^{-1}\nu_i(W D W') - t^{-1} \nu_i(W D_{\delta}W') \big \vert \le t^{-1} \|W D_{< \delta} W'\| \le \delta \|W\|^2, 
$$
where we use the fact that, given our thresholding rule, all non-zero diagonal entries of the matrix $t^{-1} D_{< \delta}$ is bounded by $\delta$ for $t \in \{x, y\}$. Substituting the bound, we obtain, 
\begin{align}\label{eq:Delta1_bd}
\sum_{t \in \{x, y\}} |\Delta_{1,t}| \le N \|W\|^2 \, \delta. 
\end{align}

Next, we bound $|\Delta_{2,t}|$ for $t \in \{x, y\}$. To that end, once again using that the map $u \mapsto \log(1+u)$ is Lipshcitz, bound 
\begin{align*}
\big \vert \log(1 + z' M_t^{-1} z/\omega) - \log(1 + z' M_{t, \delta}^{-1} z/\omega) \big \vert  &\le  \big \vert z'(M_t^{-1} - M_{t, \delta}^{-1})z/\omega \big \vert \\
&\le (\|z\|^2/\omega) \, \| M_t^{-1} - M_{t, \delta}^{-1} \| \\
&\le (\|z\|^2/\omega) \, \| M_t^{-1} (M_{t, \delta} - M_t) M_{t, \delta}^{-1} \| \\
&\le (\|z\|^2/\omega) \, \|M_{t, \delta} - M_t \|,
\end{align*}
where we have used the identity $A^{-1} - B^{-1} = A^{-1} (B-A) B^{-1}$, the bound $\|AB\| \le \|A\| \|B\|$, and the fact that both $\|M_t^{-1}\|$ and $\|M_{t, \delta}^{-1}\|$ are bounded above by 1. Continuing from the last line of the display, $\|M_t - M_{t,\delta}\| = \| t^{-1} WD_{<\delta} W' \| \le \delta \|W\|^2$ using the same argument as in the bound for $\Delta_{1,t}$. Substituting this bound, we obtain, 
\begin{align}\label{eq:Delta2_bd}
\sum_{t \in \{x, y\}} |\Delta_{2,t}| \le (N + \omega) \,(\|z\|^2/ \omega)\, \|W\|^2 \, \delta. 
\end{align}
Substituting \eqref{eq:Delta1_bd} and \eqref{eq:Delta2_bd} in \eqref{eq:Delta_bd}, we obtain \eqref{eq:del_bd}. Now, making a Taylor expansion of $e^{x}-1$ about zero, we obtain for $0< x < 1$
\be
e^{x}-1 = (1 + x + \mathcal O(x^2))-1 = x + \mathcal O(x^2),
\ee
which gives
\be \label{eq:TV2_bound}
\TVar_2 = N \, \|W\|^2 \, (1 + \|z\|^2/\omega)  \delta + \mathcal O(\delta^2)
\ee
for sufficiently small $\delta$. 

\subsection{Bounding $TV_1$}
To bound the total variation distance between $p_2(\cdot \mid \xi, \eta)$ and $p_{2,\epsilon}(\cdot \mid \xi, \eta)$, we use Pinsker's inequality, 
\begin{align}\label{eq:pinsker}
\|p_2(\theta \mid \xi, \eta) - p_{2,\epsilon}(\theta \mid \xi, \eta)\|_{\TV}^2 \le \frac{1}{2} \, \mathrm{KL}\big(p_{2,\epsilon}(\theta \mid \xi, \eta)\,||\,p_2(\theta \mid \xi, \eta)\big), 
\end{align}
and subsequently use the expression for KL between two MNIGs derived in Lemma \ref{lem:KL_MNIG}; note that the shape parameters $a_\delta = a = (N+\omega)/2$ and hence the Lemma applies. 

\

Let us define
\begin{align}\label{eq:KL_terms}
\begin{aligned}
& \mbox{KL}_1 = \tr(\Sigma^{-1} \Sigma_\delta - I_p) - \log |\Sigma^{-1} \Sigma_\delta|,  \\
& \mbox{KL}_2 = (\mu- \mu_\delta)' \Sigma^{-1} (\mu - \mu_\delta) \, \frac{a_\delta}{a'_\delta}, \\
& \mbox{KL}_3 = a_\delta \log(a'_\delta/a') + (a' - a'_\delta)  \, \frac{a_\delta}{a'_\delta}
\end{aligned}
\end{align}
so that 
$$
\mathrm{KL}\big(p_{2,\epsilon}(\cdot \mid \xi, \eta)\,||\,p_2(\cdot \mid \xi, \eta)\big) = 0.5(\mbox{KL}_1 + \mbox{KL}_2) + \mbox{KL}_3. 
$$
We now proceed to bound each of the terms subsequently. 

\subsubsection{Bounds for $\mbox{KL}_1$}

The matrix $\Sigma^{-1} \Sigma_\delta$ is similar to the positive definite matrix $\Sigma^{-1/2} \Sigma_\delta \Sigma^{-1/2}$, and hence its eigenvalues $\{\zeta_j\}_{j=1}^p$ are all positive. Expressing the trace and determinant in terms of the eigenvalues, we obtain, 
\begin{align}\label{eq:KL_1a}
\mbox{KL}_1 = \sum_{j=1}^p( \zeta_j - 1 - \log \zeta_j).  
\end{align}
Now, write 
\begin{align*}
\Sigma_\delta = \Sigma + \Delta, \quad \Delta = \Gamma W' M^{-1} W \Gamma - \big(2 \Gamma W'  - \Gamma_\delta W' M_\delta^{-1} M \big) \, M_\delta^{-1} W \Gamma_\delta, 
\end{align*}
and 
\begin{align*}
\Sigma^{-1} \Sigma_\delta = I_p + \Sigma^{-1} \Delta. 
\end{align*}
Using $\mbox{rank}(B_1 B_2) \le \min\{\mbox{rank}(B_1), \mbox{rank}(B_2)\}$, $\Delta$ is the difference of two matrices with rank at most $N$ each, and using $\mbox{rank}(B_1+B_2) \le \mbox{rank}(B_1) + \mbox{rank}(B_2)$, we can bound $\mbox{rank}(\Delta) \le 2 N$, which then implies $\mbox{rank}(\Sigma^{-1} \Delta) \le 2 N$. Letting $\{\bar{\zeta}_j\}_{j=1}^p$ denote the eigenvalues of $\Sigma^{-1} \Delta$, it then follows that $\bar{\zeta}_j = 0$ for $j \ge 2N$. Since $\zeta_j = 1 + \bar{\zeta}_j$, we conclude that $\zeta_j = 1$ for $j \ge 2N$, and 
\begin{align}\label{eq:KL_1b}
\mbox{KL}_1 = \sum_{j=1}^{2N} \big( \zeta_j - 1 - \log \zeta_j \big) = \sum_{j=1}^{2N} \big[\bar{\zeta}_j - \log(1 + \bar{\zeta}_j) \big].
\end{align} 
Observe that the right hand side is a positive quantity, since $\log(1+x) \le x$ for $x > -1$ and $\bar{\zeta}_j > -1$ for all $j$ (since $\zeta_j > 0$ for all $j$). Using Taylor expansion, it can be further shown that $x - \log(1+x) < x^2$ whenever $|x| \le 1/2$. Using that the magnitude of the eigenvalues of a matrix are bounded by its operator norm, we have $|\bar{\zeta}_j| \le \|\Sigma^{-1} \Delta\|$ for all $j = 1, \ldots, 2N$. Hence, if we can show that $\|\Sigma^{-1} \Delta\|$ is small, we can bound
\begin{align}\label{eq:KL_1c}
\mbox{KL}_1 \le \sum_{j=1}^{2N} |\bar{\zeta}_j|^2 \le 2 N \, \|\Sigma^{-1} \Delta\|^2. 
\end{align}
With this motivation, we now proceed to bound $\|\Sigma^{-1} \Delta\|$. 
To facilitate our bounds, we decompose 
$$
(\Sigma_\delta - \Sigma) = (\Sigma_\delta - \Sigma_\ast) + (\Sigma_\ast - \Sigma),
$$
where 
$$
\Sigma_\ast = \Gamma - \Gamma W'(2 M_\delta^{-1} - M_\delta^{-1} M M_\delta^{-1})W \Gamma.
$$
$\Sigma_\ast$ itself is a covariance matrix, although this isn't used in the subsequent analysis. 
Letting $A = M_\delta^{-1} M M_\delta^{-1}$, 
\begin{align*}
\Gamma_\delta W' A W \Gamma_\delta  - \Gamma W' A W \Gamma 
= \Gamma_\delta W' A W (\Gamma_\delta - \Gamma) + (\Gamma_\delta - \Gamma) W' A W \Gamma. 
\end{align*}
Hence, 
\begin{align*}
\Sigma_\delta - \Sigma_\ast = \underbrace{ 2 \Gamma W' M_\delta^{-1} W (\Gamma-\Gamma_\delta) }_{T_1} + \underbrace{ \Gamma_\delta W' A W (\Gamma_\delta - \Gamma) }_{T_2} + \underbrace{ (\Gamma_\delta - \Gamma) W' A W \Gamma }_{T_3}.
\end{align*}
Recall that $\Sigma^{-1} = (W' W + \Gamma^{-1})$. Let us now calculate 
$$
\Sigma^{-1}(\Sigma_\delta - \Sigma_\ast) = (W' W + \Gamma^{-1}) (T_1 + T_2 + T_3) = H_1 + H_2 + H_3, 
$$
with 
\begin{align}\label{eq:H1_3}
\begin{aligned}
H_1 & = 2 \big[ W' \, W\Gamma W' \, M_\delta^{-1} \, W (\Gamma-\Gamma_\delta) + W' M_\delta^{-1} W (\Gamma - \Gamma_\delta) \big] \\
H_2 & = \big[W' \, W\Gamma_\delta W' \, A \, W (\Gamma-\Gamma_\delta) + \Gamma^{-1} \Gamma_\delta \, W' A W (\Gamma_\delta - \Gamma) \big] \\
H_3 & = \big[W'W \, (\Gamma_\delta - \Gamma) W' A W \Gamma + \Gamma^{-1} \, (\Gamma_\delta - \Gamma) W' A W \Gamma\big]. 
\end{aligned}
\end{align}
Next, 
\begin{align*}
\Sigma_\ast - \Sigma  = \Gamma W' \, \underbrace{ \big[M^{-1} + M_\delta^{-1} M M_\delta^{-1} - 2 M_\delta^{-1} \big]}_{E} \,W \Gamma.
\end{align*}
Hence, 
\begin{align}\label{eq:H_4}
H_4 :\,= \Sigma^{-1}(\Sigma_\ast - \Sigma) = (W' W + \Gamma^{-1}) \Gamma W' E W \Gamma.
\end{align}
Combining \eqref{eq:H1_3} and \eqref{eq:H_4} and using the triangle inequality for the operator norm,
\begin{align*}
\| \Sigma^{-1} \Delta \| = \| \Sigma^{-1}(\Sigma_\delta - \Sigma) \| = \| H_1 + H_2 + H_3 + H_4\| \le \sum_{i=1}^4 \|H_i\|. 
\end{align*}
We now record a Lemma which collects various results required to bound the operator norms of the $H_i$s; a proof is provided in Appendix \ref{sec:pf_lem_var}. 
\begin{lemma}\label{lem:various}
The following inequalities hold: \\
$\mathrm{(i)}$ $\max \big\{ \|M^{-1}\|, \|M_\delta^{-1}\| \big\} \le 1$. \\
$\mathrm{(ii)}$ $\max \big\{ \|M - M_\delta\|, \|M^{-1} M_\delta - I_N\|, \|M_\delta M^{-1} - I_N\|, \|M_\delta^{-1} M - I_N\|, 
\|M M_\delta^{-1} - I_N \| \big\} \le  \|W\|^2 \delta$. \\
$\mathrm{(iii)}$ $ \max \big\{ \| W \Gamma W' \, M^{-1} \|, \|W \Gamma_\delta W' \, M_\delta^{-1} \| \big\} \le 1$. \\
$\mathrm{(iv)}$ $\| W \Gamma W' \, M_\delta^{-1}\| \le 1 + \|W\|^2 \delta$. \\
$\mathrm{(v)}$ Recalling that $A = M_\delta^{-1} M M_\delta^{-1}$, we have $\|A\| \le (1 + \|W\|^2 \delta)$. Further, $\|W \Gamma_\delta W' \, A\| \le (1 + \|W\|^2 \delta)$ and $\|A \, W \Gamma W'\| \le (1 + \|W\|^2 \delta)^2$. 
\end{lemma}
Using Lemma \ref{lem:various}, we now proceed to bound the $\|H_i\|$s; that $\|\Gamma - \Gamma_\delta\| < \delta$ is used throughout, along with the facts $\|B_1 B_2\| = \|B_2 B_1\|$ and $\|B_1 + B_2\| \le \|B_1\| + \|B_2\|$.

\

\noindent {\bf Bound for $\|H_1\|$.} We obtain, using (i) and (iv) in Lemma \ref{lem:various}, 
\begin{align*}
\|H_1\| \le 2 \|W\|^2 \delta \, \big[ \|W \Gamma W' \, M_\delta^{-1}\| + \|M_\delta^{-1}\| \big] \le 2 \|W\|^2 \delta \big[2 + \|W\|^2 \delta \big]. 
\end{align*}

\

\noindent {\bf Bound for $\|H_2\|$.} We obtain, using (v) in Lemma \ref{lem:various} and the fact that $\|\Gamma^{-1} \Gamma_\delta\| \le 1$, 
\begin{align*}
\|H_2\| \le \|W\|^2 \delta \, \big[\|W \Gamma_\delta W' \, A\| + \|A\| \big] \le 2 \|W\|^2 \delta \, \big[1 + \|W\|^2 \delta\big]. 
\end{align*}

\

\noindent {\bf Bound for $\|H_3\|$.} We obtain, using (v) in Lemma \ref{lem:various} and the fact that $\|\Gamma^{-1} \Gamma_\delta\| \le 1$, 
\begin{align*}
\|H_3\| \le \|W\|^2 \delta \, \big[\|A \, W \Gamma W' \| + \|A\| \big] \le  \|W\|^2 \delta \, \big[(1 + \|W\|^2 \delta)^2 +  (1 + \|W\|^2 \delta)\big]. 
\end{align*}

\

\noindent {\bf Bound for $\|H_4\|$.} We have, 
\begin{align*}
\| H_4\| 
& = \| W \Gamma (W' W + \Gamma^{-1}) \Gamma W' \, E\| \\
& = \| W \Gamma W' \, (I_N + W \Gamma W') \, E\| \\
& = \| W \Gamma W' \, ME\|.
\end{align*}
Now, $M E = I_N + M M_\delta^{-1} M M_\delta^{-1} - 2 M M_\delta^{-1} = (M M_\delta^{-1} - I_N)^2$. Substituting in the above display, and once again invoking Lemma \ref{lem:various}, 
\begin{align*}
\|H_4\| 
& = \| M_\delta^{-1} W \Gamma W' \, (M - M_\delta)M_\delta^{-1} (M - M_\delta) \| \\
& \le \|W \Gamma W' M_\delta^{-1} \| \, \|M - M_\delta\|^2 \\
& \le (1 + \|W\|^2 \delta) \, (\|W\|^2 \delta)^2. 
\end{align*}

\

\noindent {\bf Bound for $\|\Sigma^{-1} \Delta\|$.} Collecting the bounds for $\|H_i\|$ and substituting in the display before Lemma \ref{lem:various} plus some simplifying algebra yields, 
\begin{align}\label{eq:siginvdel_final}
\|\Sigma^{-1} \Delta\| \le (\|W\|^2 \delta) \, \big[3 + 3(1 + \|W\|^2 \delta) + 2(1 + \|W\|^2 \delta)^2 \big] = 8 \|W\|^2 \delta + \mathcal O(\delta^2) 
\end{align}
for sufficiently small $\delta$.

\subsubsection{Bound for $\mbox{KL}_2$}

Focus first on $(\mu - \mu_\delta)' \Sigma^{-1} (\mu - \mu_\delta)$. We have $\mu - \mu_\delta = (\Gamma W' M^{-1} - \Gamma_\delta W' M_\delta^{-1})z$. Write 
$$
\Gamma W' M^{-1} - \gamma_\delta W' M_\delta^{-1} = \underbrace{ (\Gamma - \Gamma_\delta)W' M^{-1}}_{U} + \underbrace{ \Gamma_\delta W' M^{-1} (I_N - M M_\delta^{-1}) }_V.
$$
We can now write
\begin{align*}
(\mu - \mu_\delta)' \Sigma^{-1} (\mu - \mu_\delta) 
& = z' (U+V)'\Sigma^{-1}(U+V) z \\
& \le \|(U+V)'\Sigma^{-1}(U+V)\| \, \|z\|^2 \\
& \le \| \Sigma^{-1/2}(U+V)\|^2 \, \|z\|^2 \\
& \le 2 (\| \Sigma^{-1/2} U\|^2 + \| \Sigma^{-1/2} V\|^2) \, \|z\|^2 \\
& = 2 (\|U' \Sigma^{-1} U\| + \|V' \Sigma^{-1} V\|) \, \|z\|^2,
\end{align*}
where we used the inequality $\|B_1 + B_2\|^2 \le 2(\|B_1\|^2 + \|B_2\|^2)$. Next, 
\begin{align*}
\|U' \Sigma^{-1} U\| 
& = \| M^{-1} W (\Gamma - \Gamma_\delta) (W'W + \Gamma^{-1}) (\Gamma - \Gamma_\delta) W' M^{-1} \| \\
& = \|(\Gamma - \Gamma_\delta) (W'W + \Gamma^{-1}) \,  (\Gamma - \Gamma_\delta) W' M^{-2} W\| \\
& \le \|(\Gamma - \Gamma_\delta) (W'W + \Gamma^{-1}) \| \, \|W\|^2  \delta\\
& \le \|W\|^2  \delta (1 + \|W\|^2  \delta) =  \|W\|^2  \delta + \mathcal O(\delta^2), 
\end{align*}
for sufficiently small $\delta$, where we have used conclusions of Lemma \ref{lem:various} in multiple places and in the last step, we used $\|(\Gamma - \Gamma_\delta) \Gamma^{-1} \| \le 1$ since it is a diagonal matrix with zeros and ones on the diagonal. Similarly, it can be verified that $\|V' \Sigma^{-1} V \| \le \|W\|^2  \delta (1 + \|W\|^2  \delta)$. So then it follows
\be
\KL_2 \le  \left(2 \|W\|^2  \delta + \mathcal O(\delta^2) \right) \frac{N+\omega}{\omega} = 2 \frac{N+\omega}{\omega} \|W\|^2 \delta + \mathcal O(\delta^2).
\ee
where the last factor is an upper bound on $a_\delta/a'_\delta$ which originates from bounding $a'_\delta$ below by $\omega/2$. 

\subsubsection{Bound for $\mbox{KL}_3$}
Using $\log(x) \le (x-1)$ for $x > 0$, we have, 
\begin{align*}
\mbox{KL}_3 
& \le a_\delta \big\{ (a'_\delta/a' - 1) + (a'/a'_\delta - 1) \big\} \\
& \le 2 a_\delta |a' - a'_\delta|,
\end{align*}
since $a', a'_\delta > 1$. Since $|a' - a'_\delta| = |z' (M^{-1} - M_\delta^{-1})z|/\omega \le (\|z\|^2/\omega) \, \|W\|^2 \delta$, we have 
\begin{align*}
\mbox{KL}_3 \le N \, (\|z\|^2/\omega) \, \|W\|^2 \delta.
\end{align*}

\subsubsection{Summing up}
We now combine the bounds to obtain the final result. We have
\be
\KL_1+\KL_2 +\KL_3 \le 8 \|W\|^2 \delta + 2 \frac{N+\omega}{\omega} \|W\|^2 \delta + N \, (\|z\|^2/\omega) \, \|W\|^2 \delta +  \mathcal O(\delta^2)
\ee
so by Pinsker's inequality
\be
\TVar_1^2 \le 4 \|W\|^2 \delta + \frac{N+\omega}{\omega} \|W\|^2 \delta + \frac{N}2 \, \frac{\|z\|^2}{\omega} \, \|W\|^2 \delta +  \mathcal O(\delta^2)
\ee
and so finally, combining with \eqref{eq:TV2_bound} -- which contributes only factors of order $\delta^2$ or smaller after squaring -- via \eqref{eq:TVCombine}, we obtain
\be
\sup_x \| \delta_x \P - \delta_x \P_\epsilon\|_{\TVar} &= \sqrt{4 \|W\|^2 \delta + \frac{N+\omega}{\omega} \|W\|^2 \delta + \frac{N}2 \, \frac{\|z\|^2}{\omega} \, \|W\|^2 \delta} +  \mathcal O(\delta) \\
&= \sqrt{\delta} \|W\| \sqrt{ 4 + \frac{N+\omega}{\omega} + \frac{N}2 \, \frac{\|z\|^2}{\omega} } + \mathcal O(\delta), 
\ee
since none of the bounds depend upon the remaining state variable $\eta$, giving the result.

%\subsection{Proof of auxiliary lemmata}

\section{Some integrals, inequalities, \& proofs of auxiliary lemmas}

\subsection{Incomplete Gamma function}
The incomplete Gamma function $\Gamma(a,x)$ for $x > 0$ is defined as 
\be\label{eq:incomp_gam}
\Gamma(a,x) = \int_x^\infty t^{a-1} e^{-t} dt. 
\ee
When $a = 0$, this reduces to the exponential integral function $\En_1(x) = \int_x^\infty t^{-1} e^{-t} dt$. 

\

We record an integral from Gradstheyn and Ryzhik (GR 3.383.10), 
\be\label{eq:gr_383_10}
\int_{0}^{\infty} \frac{x^{\nu - 1} e^{-\mu x}}{x + \beta} \,dx = \beta^{\nu-1} e^{\beta \mu} \, \Gamma(\nu) \, \Gamma(1-\nu, \beta \mu),
\ee
for $\nu, \mu, \beta > 0$.

\

We record a result relating the ratio of certain incomplete gamma functions. 
\begin{lem}\label{lem:imcomp_gamrat2}
Fix $c \in (0, 1/2]$ and $b \in (0, 1)$. For any small $\varepsilon > 0$, there exists a positive constant $C_\varepsilon$ such that 
\be
r_{b,c}(x) :\,= \frac{e^{-x}}{x^c} \, \frac{\Gamma(c, bx)}{\Gamma(0, x+bx)} \le \varepsilon x^{-c} + C_\varepsilon, \quad \forall \, x \in (0, \infty). 
\ee
\end{lem}
\begin{proof}
For $x \ge 1/2$, bound $\Gamma(0, x) \ge \int_x^{2x} e^{-t}/t \,dt \ge e^{-x} (1 - e^{-x})/(2x) \ge e^{-x}/(8x)$, where we used that for $x \ge 1/2$, $1 - e^{-x} \ge 1/4$. Also bound $\Gamma(c, x) = \int_x^\infty t^{c-1} e^{-t} dt \le x^{c-1} \int_x^\infty e^{-t} dt = x^{c-1} e^{-x}$. Substituting these bounds, we have for $x \ge 1/2$ that 
\be\label{eq:inv_abs_mom}
\frac{e^{-x}}{x^c} \, \frac{\Gamma(c, bx)}{\Gamma(0, x+bx)} \le 8 b^{c-1} (1+b). 
\ee
We also have that $\lim_{x \to 0} \frac{\Gamma(c, bx)}{\Gamma(0, x+bx)} = 0$. Pick $\delta > 0$ such that 
$\frac{\Gamma(c, bx)}{\Gamma(0, x+bx)} < \varepsilon$ for all $x < \delta$. We can then bound 
\be 
r_{b,c}(x) \le \varepsilon x^{-c} + \max\{r_{b,c}(\delta), 8 b^{c-1} (1+b)\}, \quad \forall \, x\in (0, \infty). 
\ee
\end{proof}

\subsection{Normal inverse moments \& Hypergeometric function} \label{sec:NormalHypergeo}
We state a formula for inverse absolute moments of a normal distribution. Let $X \sim \No(\mu, \sigma^2)$. One has, for $\nu > -1$, 
\be\label{eq:normal_invmom}
E(|X|^\nu) = \sigma^\nu 2^{\nu/2} \frac{\Gamma\left( \frac{\nu+1}{2} \right)}{\sqrt{\pi}} \, e^{-\frac{\mu^2}{2\sigma^2}} \, M\left( \frac{\nu+1}2, \frac12; \frac{\mu^2}{2\sigma^2} \right), 
\ee
where
\be
M(\alpha,\gamma;z) &= \mathstrut_1 F_1(\alpha;\gamma;z) :\,= \sum_{n=0}^\infty \frac{(\alpha)_n}{(\gamma)_n n!} z^n
\ee
is the confluent hypergeometric function of the first kind, with  
$$
(x)_n \equiv \frac{\Gamma(x+n)}{\Gamma(x)},
$$
the ascending factorial. Letting
\be
\mathbf{M}(\alpha,\gamma;z) &= \frac{M(\alpha,\gamma;z)}{\Gamma(\gamma)}, 
\ee
one has, if $\gamma > \alpha > 0$ then
\be
\mathbf{M}(\alpha,\gamma;z) &= \frac{1}{\Gamma(\alpha)\Gamma(\gamma-\alpha)} \int_0^1 e^{zt} t^{\alpha-1} (1-t)^{\gamma-\alpha-1} dt. 
\ee
%From this integral representation, it is immediate that if $\gamma > \alpha > 0$, then $M(\alpha, \gamma; z)$ is a monotone increasing function of $z$. 
We state a useful result below. 
\begin{lem}\label{lem:conf_dec}
Fix $\gamma > \alpha > 0$. The function
\be 
z \mapsto e^{-z} M(\alpha, \gamma; z)
\ee
is a non-increasing function of $z$ for $z \ge 0$. In particular, $e^{-z} M(\alpha, \gamma; z) \le 1$ for any $z > 0$. 
\end{lem}
\begin{proof}
We can write, based on the above integral representation, 
\be
e^{-z} M(\alpha, \gamma; z) 
&= \frac{\Gamma(\gamma)}{\Gamma(\alpha)\Gamma(\gamma-\alpha)} \,
\int_0^1 e^{-z(1-t)} t^{\alpha-1} (1-t)^{\gamma-\alpha-1} dt \\
&= \frac{\Gamma(\gamma)}{\Gamma(\alpha)\Gamma(\gamma-\alpha)} \,
\int_0^1 e^{-zt} t^{\gamma-\alpha-1} (1-t)^{\alpha-1} dt,
\ee
which is a non-increasing function of $z$ on $[0, \infty)$. 

The second part follows from evaluating the last line of the above display at $z=0$, which reduces to the integral of a beta pdf at $0$, so that $e^{-z} M(\alpha, \gamma; z)\vert_{z=0} = 1$. The bound then follows. 
\end{proof}

\subsection{Proof of Lemma \ref{lem:KL_MNIG}}\label{sec:KL_lemma}
We have 
\begin{align*}
& \int p_0(\beta, \sigma^2) \log \frac{p_0(\beta, \sigma^2)}{p_1(\beta, \sigma^2)} \, d \beta d \sigma^2 \\
& = \int p_0(\beta \mid \sigma^2) \, p_0(\sigma^2) \bigg[ \log \frac{p_0(\beta \mid \sigma^2)}{p_1(\beta \mid \sigma^2)} + \log \frac{p_0(\sigma^2)}{p_1(\sigma^2)} \bigg] \,d\beta d \sigma^2 \\
& = \int \mbox{KL}\bigg(p_0(\cdot \mid \sigma^2)\,||\,p_1(\cdot \mid \sigma^2)\bigg) \, p_0(\sigma^2) \, d\sigma^2 + \int p_0(\sigma^2) \log \frac{p_0(\sigma^2)}{p_1(\sigma^2)} \, d\sigma^2. 
\end{align*}
First, using normality of $p_i(\cdot \mid \sigma^2)$ and the standard expression for the KL divergence between two multivariate normals, 
\begin{align*}
\mbox{KL}\bigg(p_0(\cdot \mid \sigma^2)\,||\,p_1(\cdot \mid \sigma^2)\bigg) 
= \frac{1}{2} \bigg[ \tr(\Sigma_1^{-1} \Sigma_0 - I_p) - \log |\Sigma_1^{-1} \Sigma_0| + \frac{(\mu_1 - \mu_0)'\Sigma_1^{-1}(\mu_1 - \mu_0)}{\sigma^2} \bigg]. 
\end{align*}
Thus, 
\begin{align*}
& \int \mbox{KL}\bigg(p_0(\cdot \mid \sigma^2)\,||\,p_1(\cdot \mid \sigma^2)\bigg) \, p_0(\sigma^2) \, d\sigma^2 \\
& = \frac{1}{2} \bigg[ \tr(\Sigma_1^{-1} \Sigma_0 - I_p) - \log |\Sigma_1^{-1} \Sigma_0| + (\mu_1 - \mu_0)'\Sigma_1^{-1}(\mu_1 - \mu_0) \, \frac{a_0}{a'_0} \bigg]. 
\end{align*}
Next, using that $a_0 = a_1$, 
\begin{align*}
& \int p_0(\sigma^2) \log \frac{p_0(\sigma^2)}{p_1(\sigma^2)} \, d\sigma^2 \\
& = \int p_0(\sigma^2) \, [ (a_0 \log a'_0 - a_0 \log a'_1) + (a'_1 - a'_0) (\sigma^2)^{-1} ] \, d\sigma^2 \\
& = a_0 \log(a'_0/a'_1) + \frac{(a'_1 - a'_0) a_0}{a'_0}.
\end{align*}

\subsection{Proof of Lemma \ref{lem:various}}\label{sec:pf_lem_var}
We make multiple usage of the following facts. For matrices $A$ and $B$ of compatible size, $\|AB\| = \|BA\| \le \|A\|\, \|B\|$ and $\|A+B\| \le \|A\| + \|B\|$. For invertible $A$, $\|A^{-1}\| = 1/s_{\min}(A)$. For a symmetric p.d. matrix $A$, its eigenvalues and singular values are identical. \\[1ex]
(i) Follows since $s_{\min}(M)$ and $s_{\min}(M_\delta)$ are both bounded below by 1. \\[1ex]
(ii) First, $\|M - M_\delta\| = \| W (\Gamma - \Gamma_\delta) W'\| \le \|W\|^2 \delta$ since $\Gamma - \Gamma_\delta$ is a diagonal matrix with the non-zero entries bounded by $\delta$. The remaining 4 inequalities have near identical proofs so we only prove one of them. We have $\| M_\delta^{-1} M - I_N \| = \| M_\delta^{-1} (M_\delta - M)\| \le \|M - M_\delta\|$ by (i). \\[1ex]
(iii) Writing $W \Gamma W' \, M^{-1} = I_N - M^{-1}$, all its eigenvalues are bounded above by 1. Similarly for the second part. \\[1ex]
(iv) Bound $\| W \Gamma W' \, M_\delta^{-1}\| \le \|W \Gamma W' \, M^{-1} \| + \|W \Gamma W' \, M^{-1} \, (I_N - M M_\delta^{-1})\| \le \|W \Gamma W' \, M^{-1} \| \, \big( 1 + \|I_N - M M_\delta^{-1}\| \big)$. Conclude from (ii) and (iii). 
\\[1ex]
(v) First, bound $\| A \| \le \| M_\delta^{-1} M \| \, \|M_\delta^{-1} - M^{-1} \| + \|M_\delta^{-1}\|$. The bound then follows from (i) and (ii) and noting that $(M_\delta^{-1} - M^{-1}) = M_\delta^{-1} (M - M_\delta) M^{-1}$. For the third bound, write 
$\|A \, W \Gamma W'\| = \|M_\delta^{-1} M \, M_\delta^{-1} W\Gamma W'\| \le \|M_\delta^{-1} M\| \, \|M_\delta^{-1} W\Gamma W'\|$. The bound then follows from (ii) and (iii). The bound for $\|W \Gamma_\delta W' \, A\|$ follows similarly. 

%\FloatBarrier
%\newpage
\bibliographystyle{apalike}
\bibliography{lg-algos}

\newpage

\beginsupplement
\noindent\begin{LARGE}\textbf{Supplementary Materials}\end{LARGE} \\

\noindent This supplement contains derivation of the rejection sampler and some additional figures as described in the main text.

\section{Rejection sampler for local scales}
Fix $\varepsilon \in (0, 1)$ and consider sampling from the density  
\be
h_\varepsilon(t) = C_\varepsilon \, \frac{e^{- \varepsilon t}}{1+t}, \quad t > 0, 
\ee
where the normalizing constant $C_\varepsilon = e^{-\varepsilon}/\mbox{Ei}(\varepsilon)$, with $\mbox{Ei}(x) = \int_x^\infty e^{-t}/ t \, dt = \Gamma(0,x)$ the exponential integral function. The constant $C_\varepsilon$ is a decreasing function of $\varepsilon$, with $C_1 \approx 1.6$ and $C_\varepsilon < 1$ for $\varepsilon < 0.40$. 

First we record useful fact about the density $h_\varepsilon$. If $X \sim \mbox{Expo}(\varepsilon)$ with $E(X) = 1/\varepsilon$, then $P(X > b/\varepsilon) = e^{-b}$ for any $b > 0$. We show a similar upper bound for 
$h_\varepsilon$. Bound 
\begin{align*}
C_\varepsilon \, \int_{b/\varepsilon}^{\infty} \frac{e^{- \varepsilon t}}{1+t} \, dt 
\le \frac{C_\varepsilon}{1 + b/\varepsilon} \, \int_{b/\varepsilon}^{\infty} e^{- \varepsilon t} \, dt 
= \frac{C_\varepsilon}{1 + b/\varepsilon} \, \frac{1}{\varepsilon} \, e^{-b}  \le C_\varepsilon \, \frac{e^{-b}}{b}. 
\end{align*}
Let 
$$
f(x) :\,= f_\varepsilon(x) = \varepsilon x + \log(1+x), \quad x > 0, 
$$
be the negative log-density up to constants. It is easily verified that $f$ is an increasing concave function on $(0, \infty)$. We now develop a lower bound to $f$. 

For any real-valued function $g$ and an interval $[\munderbar{v}, \bar{v}] \subset \mbox{dom}(g)$, recall that the line segment on the interval $[\munderbar{v}, \bar{v}]$ joining $g(\munderbar{v})$ and $g(\bar{v})$ is given by  
$$
x \mapsto g(\munderbar{v}) + \frac{g(\bar{v}) - g(\munderbar{v})}{\bar{v}-\munderbar{v}} \, (x - \munderbar{v}), \quad x \in [\munderbar{v}, \bar{v}]. 
$$
Fix $0 < a < 1 < b$, and set 
$$
A = f(a/\varepsilon), \ I = f(1/\varepsilon), \ B = f(b/\varepsilon). 
$$
Also, set 
$$
\lambda_2 = \frac{I-A}{(1-a)/\varepsilon}, \ \lambda_3 = \frac{B-I}{(b-1)/\varepsilon}. 
$$
With these notations, set
\begin{align*}
f_{L, \varepsilon}(x) :\, = f_L(x) = 
\begin{cases}
\log(1+x) &  x \in [0, a/\varepsilon),\\
A + \lambda_2(x - a/\varepsilon) & x \in [a/\varepsilon, 1/\varepsilon), \\
I + \lambda_3 (x - b/\varepsilon) & x \in [1/\varepsilon, b/\varepsilon), \\
B + \varepsilon(x - b/\varepsilon) & x \ge b/\varepsilon. 
\end{cases}
\end{align*}
Some comments about the approximation $f_L$. First, $f_L$ is an increasing function and is piecewise linear on $[a/\varepsilon, \infty)$. It has a jump discontinuity at $a/\varepsilon$ and is continuous everywhere else.
$f_L$ is identical to $\log(1+x)$ on $[0, a/\varepsilon)$, linearly interpolates between (i) $f(a/\varepsilon)$ and  $f(1/\varepsilon)$ on $[a/\varepsilon, 1/\varepsilon)$ and (ii) $f(1/\varepsilon)$ and  $f(b/\varepsilon)$ on $[1/\varepsilon, b/\varepsilon)$, and equals $\varepsilon x + \log(1+b/\varepsilon)$ on $[b/\varepsilon, \infty)$. By construction, $f_L \le f$ on $[0, a/\varepsilon)$, and the concavity of $f$ implies $f_L \le f$ on $[a/\varepsilon, \infty)$, implying that $f_L$ is globally bounded from above by $f$. 
%The maximum error in approximation in the first segment $\sup_{x \in [0, a/\varepsilon)} [f(x) - f_L(x)] = a$. \\

\begin{figure}[h!]
%\centering
\hspace{-0.2in}
\includegraphics[width=0.9\columnwidth]{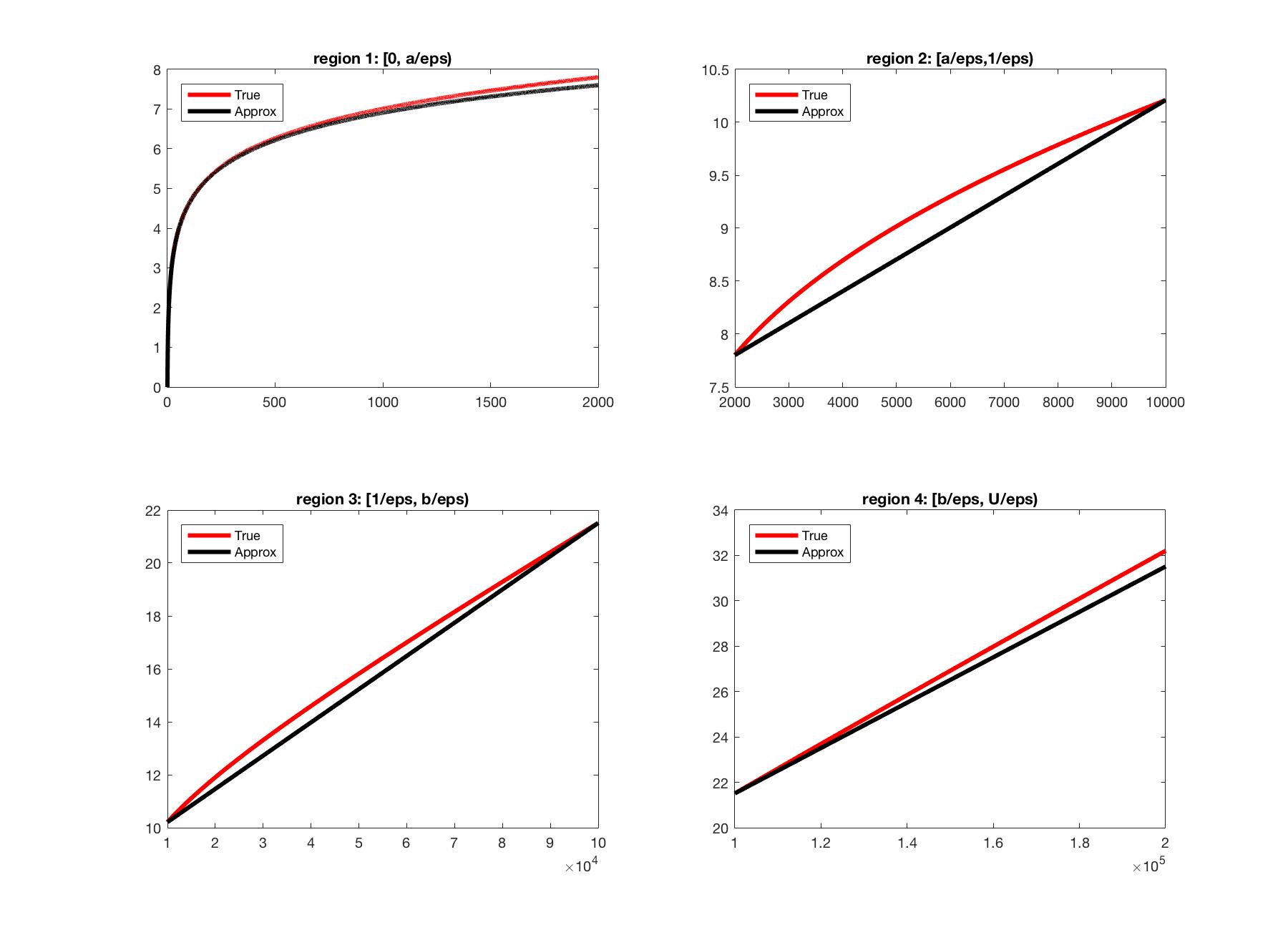}
\caption{Comparison of $f$ and $f_L$ with $\varepsilon = 10^{-4}$, $a = 1/5$, and $b = 10$. }
\end{figure}

Let $h_L(x) = e^{-f_L(x)}/\nu$ for $x \in (0, \infty)$, with $\nu = \int_0^\infty e^{- f_L(x)} \, dx$. A rejection sampling algorithm to sample from $h$ proceeds as follows: \\[1ex]
(i) draw $z \sim h_L$ and $u \sim U(0, 1)$ independently. \\
(ii) Accept $z$ as a sample from $h$ if $u < e^{-(f - f_L)(z)}$. Otherwise, back to step (i). \\

We now describe sampling from $h_L$. To that end, let us first calculate the normalizing constant $\nu$. We have, 
\begin{align*}
\nu & = \nu_1 + \nu_2 + \nu_3 + \nu_4, \\
\nu_1 & = \int_0^{a/\varepsilon} \frac{dx}{1+x} = \log(1 + a/\varepsilon), \\ 
\nu_2 & = e^{- A} \, \int_{a/\varepsilon}^{1/\varepsilon} e^{- \lambda_2 (x - a/\varepsilon) } \, dx
= \lambda_2^{-1} \, e^{- A} \, \big[1 - e^{- (I - A)} \big], \\
\nu_3 & = e^{-I} \, \int_{1/\varepsilon}^{b/\varepsilon} e^{- \lambda_3 (x - 1/\varepsilon) } \, dx
= \lambda_3^{-1} \, e^{- I} \, \big[1 - e^{- (B - I)} \big], \\
\nu_4 & = e^{- B} \, \int_{b/\varepsilon}^{\infty} e^{- \varepsilon (x - b/\varepsilon) } \, dx = \varepsilon^{-1} \, e^{- B}.
\end{align*}

We can thus write $h_L$ as a mixture of four densities,
\begin{align*}
h_L = (\nu_1/\nu) \, h_1 + (\nu_2/\nu) \, h_2 + (\nu_3/\nu) \, h_3 + (\nu_4/\nu) \, h_4, 
\end{align*}
where 
\begin{align*}
h_1(x) &= \frac{1}{\nu_1} \, \frac{\ind_{[0, a/\varepsilon)}(x)}{1+x} , \\
h_2(x) &= \frac{1}{\nu_2} \, e^{- A} \, e^{- \lambda_2 (x - a/\varepsilon)} \, \ind_{[a/\varepsilon, 1/\varepsilon)}(x), \\
h_3(x) &= \frac{1}{\nu_3} \, e^{- I} \, e^{- \lambda_3 (x - 1/\varepsilon)} \, \ind_{[1/\varepsilon, 1b/\varepsilon)}(x), \\
h_4(x) &=\frac{1}{\nu_4} \, e^{- B} \, e^{- \varepsilon (x - b/\varepsilon)} \, \ind_{[b/\varepsilon, \infty)}(x). \\
\end{align*}
Observe that $h_2, h_3$ and $h_4$ are truncated exponential densities. We now describe the inverse cdf method to sample from a truncated exponential. 
\\[2ex]
{\bf Sampling from truncated exponential.} Let $\mbox{Expo}(\lambda, \munderbar{v}, \bar{v})$ denote the distribution with density 
$$
\gamma(x) = \frac{ \lambda e^{-\lambda (x - \munderbar{v})}}{H}, \quad x \in [\munderbar{v}, \bar{v}],
$$
where $\lambda > 0$, $0 \le \munderbar{v} < \bar{v} \le \infty$, and $H = 1 - e^{-\lambda(\bar{v}-\munderbar{v})}$ (Note:when $\bar{v} = \infty$, this means $H = 1$). The cdf 
$$
F_\gamma(x) = \frac{1 - e^{-\lambda(x-\munderbar{v})}}{H}, \quad x \in [\munderbar{v}, \bar{v}]. 
$$
The inverse-cdf method to sample from $\mbox{Expo}(\lambda, \munderbar{v}, \bar{v})$ is then given by: \\
Sample $u \sim U(0, 1)$ and set 
$$
x = \munderbar{v} + \frac{-\log(1-uH)}{\lambda} = \munderbar{v} + \frac{ -\log(1-u + u \, e^{-\lambda(\bar{v}-\munderbar{v})})}{\lambda}.
$$

After some simplification, the value of $H$ corresponding to $h_2$ and $h_3$ is respectively, 
$$
H_2 = 1 - e^{-(I-A)}, \quad H_3 = 1 - e^{-(B-I)}. 
$$
\\[2ex]
The density $h_1$ can also be sampled using inverse cdf method. We have \\[1ex]
{\bf Sampling from $h_1$:} The cdf of $h_1$ is 
$$
F_1(x) = \frac{\log(1+x)}{\nu_1}, \quad x \in [0, a/\varepsilon). 
$$
The inverse cdf sampler sets: \\
draw $u \sim U(0, 1)$ and set $x = e^{u \nu_1} -1 = (1 + a/\varepsilon)^u - 1$.

\section{Extra Figures}
Here we provide additional figures relevant to the statistical performance of time-averaging estimators from the Approximate algorithm.

\begin{figure}[h]
\centering
\includegraphics[width=0.7\textwidth]{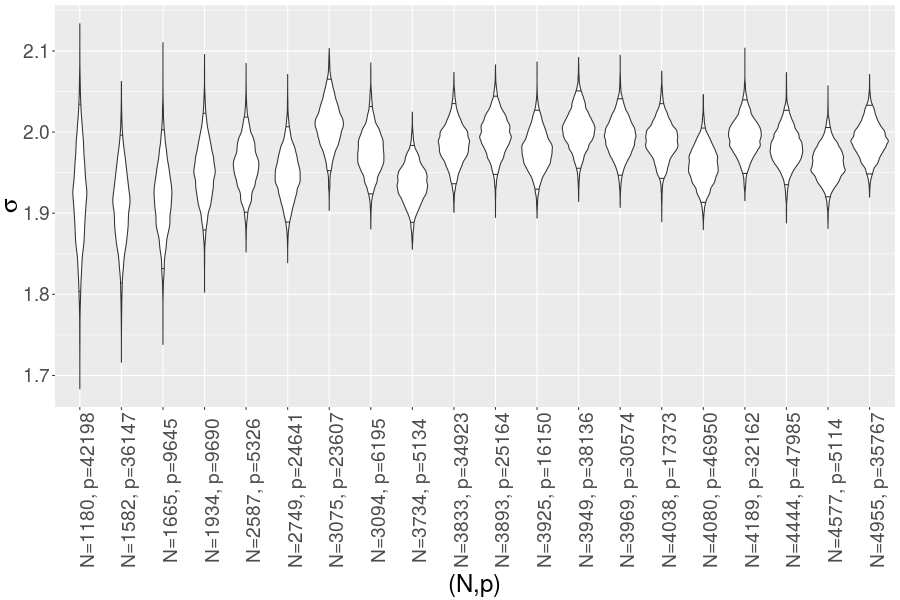}
\caption{Marginals for the residual standard deviation $\sigma$ over 20 values of $N,p$ using the approximate algorithm. The small horizontal lines indicate the 
0.025 and 0.975 approximate posterior quantiles. The true value is 2 in all cases. } \label{fig:sigma_intervals}
\end{figure}

\begin{figure}[h]
\centering
\includegraphics[width=0.9\textwidth]{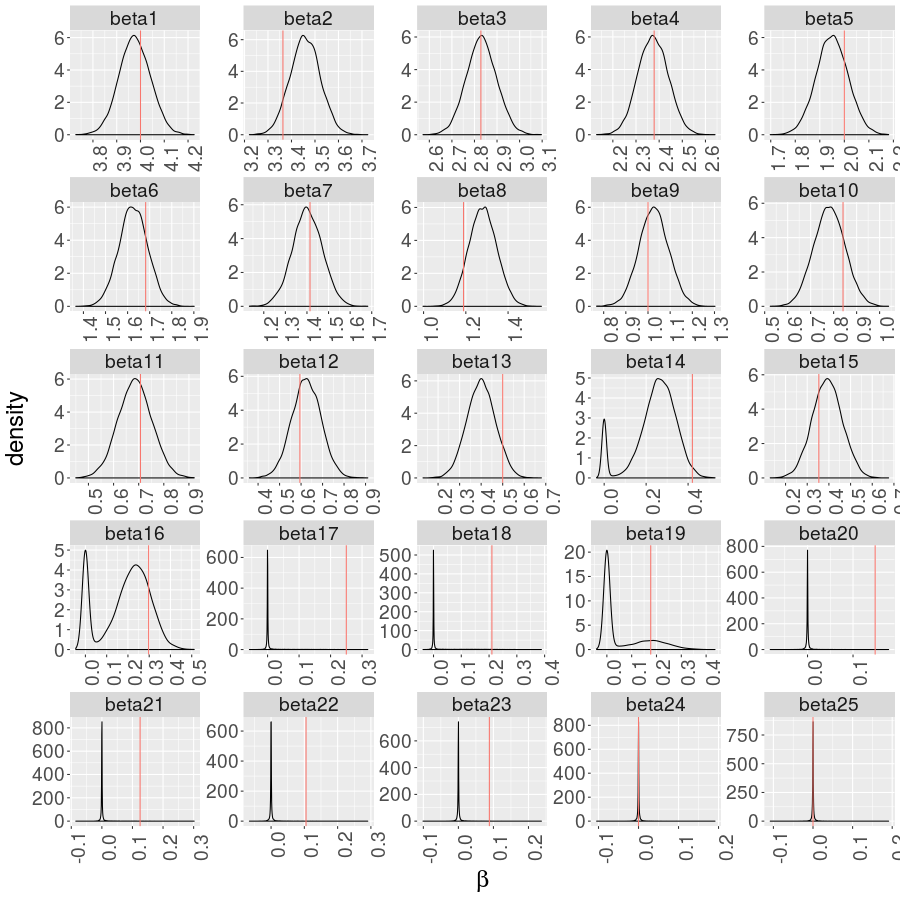}
\caption{Marginals for the first 25 entries of $\beta$ for $N=1000,p=5000$, true value indicated with red line. Approximate algorithm.} \label{fig:beta_marginals_small}
\end{figure}

\begin{figure}[h]
\centering
\includegraphics[width=0.9\textwidth]{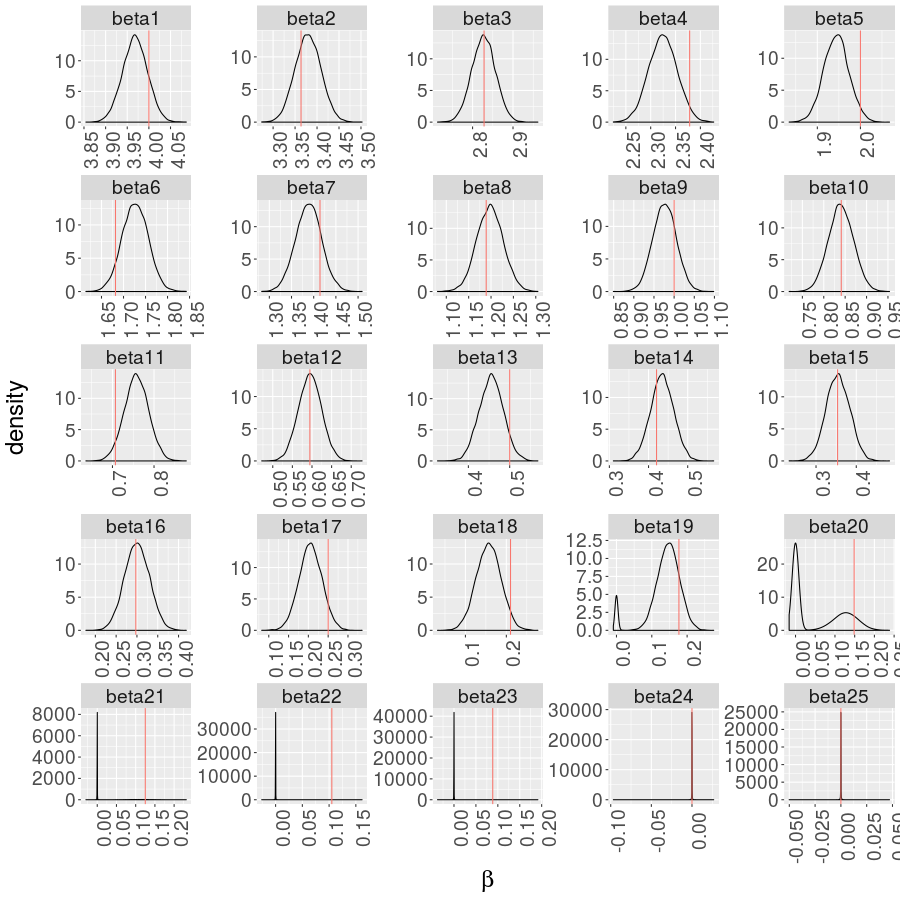}
\caption{Marginals for the first 25 entries of $\beta$ for $N=5000,p=50,000$, true value indicated with red line. Approximate algorithm.} \label{fig:beta_marginals_big}
\end{figure}

\end{document}